\documentclass[11pt]{article}
\usepackage{graphicx}
\usepackage{pgfplots}

\usepackage[colorlinks=true,citecolor=blue]{hyperref}
\usepackage{natbib}
\usepackage{graphicx}
\usepackage{amsfonts}
\usepackage{amsmath}
\usepackage{amssymb}
\usepackage{url}
\usepackage{fancyhdr}
\usepackage{indentfirst}
\usepackage{enumerate}
\usepackage{titlesec}
\usepackage{amsthm}
\usepackage{dsfont}
\usepackage[misc]{ifsym}

\usepackage[final]{changes}

\theoremstyle{definition}

\newtheorem{example}{Example}

\newtheorem{axiom}{Axiom} 
\newtheorem{assumption}{Assumption}

\theoremstyle{plain}
\newtheorem{theorem}{Theorem}
\newtheorem{lemma}{Lemma}
\newtheorem{proposition}{Proposition}
\newtheorem{corollary}{Corollary}
\theoremstyle{remark}
\newtheorem{remark}{Remark}

\usepackage{cases}

\def\laweq{\buildrel \mathrm{d} \over =}
\def\ac{\buildrel \mathrm{ac} \over \sim}

\theoremstyle{definition}

\def\N{\mathbb{N}}

\def\p{\mathbb{P}}
\def\E{\mathbb{E}}

\def\R{\mathbb{R}}
\def\A{\mathcal{A}}
\def\M{\mathcal{M}}

\def\C{\mathcal{C}}
\def\X{\mathcal{X}}

\def\X{\mathcal{X}}

\def\d{\mathrm{d}}
\DeclareMathOperator*{\esssup}{ess\text{-}sup}
\DeclareMathOperator*{\essinf}{ess\text{-}inf}
\def\Q{\mathcal{Q}}

\newcommand{\ex}{\mathrm{Ex}}

\usepackage{bm}
\usepackage{tikz-qtree}
\usepackage{tikz}

 \usepackage[onehalfspacing]{setspace}

\def\id{\mathds{1}}

\title{Disappointment concordance  and duet expectiles}
\author{Fabio Bellini\thanks{Department of Statistics and Quantitative Methods, University of Milano-Bicocca, Italy. \Letter~{\scriptsize\url{fabio.bellini@unimib.it}}}
\and Tiantian Mao\thanks{Department of Finance and Statistics, University of Science and Technology of China, China.  \Letter~{\scriptsize\url{tmao@ustc.edu.cn}}}
	\and
	Ruodu Wang\thanks{Department of Statistics and Actuarial Science, University of Waterloo,  Canada. \Letter~{\scriptsize\url{wang@uwaterloo.ca}}}
	\and
	Qinyu Wu\thanks{Department of Statistics and Actuarial Science, University of Waterloo,  Canada. \Letter~{\scriptsize\url{q35wu@uwaterloo.ca}}}
}
\date{\today}

\pgfplotsset{compat=1.18}

\begin{document}
	\maketitle
	\begin{abstract}  
We introduce an axiom of disappointment-concordance (disco) aversion for a preference relation over acts in an Anscombe-Aumann setting. This axiom means that the decision maker, facing the sum of two acts, dislikes the situation where both acts realize simultaneously as disappointments. 
Our main result is that, under strict monotonicity and continuity, the axiom of disco aversion characterizes preference relations represented by a new class of functionals belonging to the Gilboa-Schmeidler family, which we call the duet expectiled utilities.
When the outcome space is the real line, a duet expectiled utility 
becomes a duet expectile, which involves two endogenous probability measures. It further  becomes a usual expectile, a statistical quantity popular in regression and risk measures, when these two probability measures coincide. We discuss properties of duet expectiles and connections with fundamental concepts including probabilistic sophistication, risk aversion, and uncertainty aversion.

\textbf{Keywords}: Disappointment aversion, probabilistic sophistication, risk aversion, expectiles, ambiguity.
\end{abstract}

	\noindent\rule{\textwidth}{0.5pt}
	

\section{Introduction}

Fix a measurable space $(\Omega,\mathcal F)$. Acts are defined as real-valued bounded measurable functions on $(\Omega,\mathcal F)$, representing monetary payoffs, and the space of acts is denoted by $\X$. 
A decision maker has a preference relation over $\mathcal X$, 
which is a weak order $\succsim$  (the associated $\succ$ and $\sim$ are defined as usual). 
Constant acts are identified with their values. 
While we work  with real outcomes in the main part of the paper, the case of more general outcomes as in the \cite{AA63} framework will be considered in Section \ref{sec:AA} with no additional challenge.\footnote{In order to clarify the exposition, we separate the case of real outcomes (which, as we will see, is sufficient for axiomatizing the notion of duet expectile) to the case of general outcomes, which will involve a Bernoullian utility as in \cite{CGMMS11}.}

We begin by introducing the notion of disappointment event of an act $X \in \X$, defined by  
$$D_X =\{\omega\in \Omega: X(\omega)\prec X\}.$$
The disappointment event of 
$X$ is the set of states of the world in which the decision maker is worse off after the realization of 
$X$ compared to the ex-ante situation, and this can be seen as disappointing.
Clearly, the disappointment set varies with $X$ and is subjective since it is defined in terms of preferences. 
If a unique certainty equivalent 
$c_X$ of $X$ exists and $\succsim$ is increasing, then $$D_X=\{\omega\in \Omega:X(\omega) \prec c_X\}=\{\omega\in \Omega:X(\omega) < c_X\},$$
and hence the disappointment event is simply accounted relatively to the certainty equivalent.
 {We borrow the term \emph{disappointment} from \cite{G91},
who introduced the notion of elation-disappointment decomposition of a lottery, given by its decomposition in a mixture of two lotteries, the first with prizes preferred to the original lottery (elation part) and the second with prizes not preferred to the original lottery (disappointment part). Although based on the same conceptual principle, a slight difference to our definition is that Gul formulated the disappointment set as a subset of outcomes, while our disappointment event is a set of states of the world. We will discuss throughout the paper the connection with Gul's theory in detail. 
}

We move to the introduction of one of the main notions of the paper. We say that two acts are   disappointment-concordant (shortly, disco)  if their disappointment sets coincide, that is,
$$
X,Y \in \mathcal X \text { are  disco  if }  D_X=D_Y.
$$
 The interpretation is straightforward: Disco acts are pairs of acts whose outcomes are disappointing in the same states of the world. Notice that as a consequence, the outcomes are also elating in the same states of the world, that is on the complement of the disappointment event.\footnote{We could equivalently define the notion of disco acts starting from the notion of elation event, and also, as it will become clear in the following, both sets could be equivalently defined by using non-strict preferences.} It may be  interesting to compare the notion of disco acts with the notion of comonotonic acts, which is central in the axiomatization of Choquet-expected utility given in \cite{S89}. For real-valued acts, comonotonicity is an objective property, which can be defined without any reference to preferences. On the contrary, the notion of disco acts is based on disappointment that is subjective by definition.

We can now present the main novel axiom of the paper, which we call disco aversion. 
 \begin{equation}
 \label{eq:CAintro}
      \mbox{$X$ and $Y$ are disco and  $Y' \sim Y $} \implies X+Y' \succsim X+Y. 
 \end{equation}
{Disco aversion prescribes that in a sum of acts, the case in which $X$ and $Y$ are disco is least favoured among all $Y'$ indifferent to $Y$. The basic idea is that when combining acts in a sum, as for example in the construction of a portfolio of financial payoffs, the decision maker dislikes the situation in which disappointments may happen together, and instead prefers any situation in which a disappointment on $X$ can be compensated by an elation on $Y'$ or vice versa. From a financial point of view,} disco aversion reflects the natural principle that pooling random payoffs together may reduce the total risk associated with these payoffs by an effect of diversification, pioneered by \cite{M52}, but not all forms of pooling are equally desirable.
Pooling  payoffs that incur in losses simultaneously amplifies the severity of the loss scenario,
making the adverse situation even more dire than without pooling.
An example of such a situation is  the excessive use of leverage in a financial market, the opposite of hedging. 
Such a pooling situation is considered dangerous in risk management, and recognized by financial regulators in their regulatory documents (e.g., \cite{BASEL19}). 
The axiom of disco aversion describes a decision principle against pooling risks whose loss events occur together, given other things equal. 
Although $Y'\sim Y$, 
the decision maker with an existing payoff $X$ prefers $Y'$ over $Y$, because of the dependence between $X$ and $Y$.
Therefore, disco aversion can be viewed as an attitude towards dependence. Attitudes towards dependence are closely related to attitudes towards risk, as recently studied by \cite{MMWW23}; however, dependence concepts studied there are defined by their joint distributions (thus objective), whereas our concept is specific to the preference relation itself (thus subjective).

 {Consistently with this discussion, it will turn out in Theorem \ref{th:riskaverse1} that under mild additional axioms disco aversion implies preference for mixing  (often called simply convexity), i.e. 
 \begin{align*}
 \label{eq:UAintro}
      \mbox{$X \sim Y$} \implies \alpha X +(1-\alpha) Y \succsim X,
 \end{align*}
that is the fundamental axiom for max-min expected utility given in \cite{GS89}. Indeed, it will turn out that our preferences belong to the Gilboa-Schmeidler class and their representation as worst-case expected utilities have explicit formulas given in Subsection \ref{subsec:GS}.
}

 {Our notion of disco aversion can be contrasted with the notion of disappointment aversion in \cite{G91}. Gul's theory aims at providing a minimal extension of the expected utility theory that is able to accommodate the Allais paradox. In brief, its main axiom is a weakening of the classic independence axiom, that is required only for lotteries that have a common  elation-disappointment decomposition, and the resulting numerical representation of preferences is of the implicit utility type and has an additional disappointment aversion parameter. In contrast with our setting, Gul's theory is not directly based on a disappointment aversion axiom, but instead the attitude toward disappointment emerges as a consequence, in the sense that in the numerical representation of Gul's preferences expected utilities of disappointment and elation prizes are asymmetrically weighted. We refer to \cite{CDO20} for an interesting alternative explicit representation of Gul's preferences as infima or maxima of suitable certainty equivalents. In Sections \ref{subsec:Gul} and \ref{sec:AA}, we will show that Gul's preferences are special cases of preferences in our setting.}

It turns out that under mild additional conditions disco aversion is sufficient to pin down a new form of representing functional. Our main results, Theorems \ref{th-main} and \ref{th:main-general}, yield that a preference relation  satisfying disco aversion \eqref{eq:CAintro}, strict monotonicity and continuity (precise definitions are postponed to Section \ref{sec:axiom}) admits a numerical representation of the following type: 
 \begin{equation}
 \label{eq:REPintro}
   X \succsim Y \iff \ex_\alpha^{P,Q}(X) \geq \ex_\alpha^{P,Q}(Y)
 \end{equation}
where $\alpha \in (0,1/2]$, $P$ and $Q$ are two probability measures on $(\Omega, \mathcal{F})$, and 
$\ex_\alpha^{P,Q}(X)$ is the unique $y\in \R$ such that
 \begin{equation}
 \label{eq:DUETintro}
\alpha \E^P[(X-y)_+]=  (1-\alpha)\E^{Q}[(y-X)_+] .
 \end{equation}
Further, the constant $\alpha$ and the probability measures $P$ and $Q$ are uniquely determined by the preference relation on $\X$.  
 The functional $\ex_\alpha^{P,Q}$ is a generalization of the \emph{expectiles} introduced by \cite{NP87} and axiomatized by \cite{Z16} and \cite{DBBZ16} in the risk measure literature; see Section \ref{sec:solo} for more details. 
Since two probability measures are involved, we call the quantity $\ex_\alpha^{P,Q}(X)$ a \emph{duet expectile}; the usual (solo) expectiles arise in the case $P=Q$. 

{Remarkably, in contrast with the aforementioned axiomatizations of expectiles, here the probability measures $P$ and $Q$ are endogenously determined by the preferences. 
In the case of an infinite $\Omega$ we will assume in Theorem \ref{th:main-general} the existence of a reference measure $\p$ but its only role is the identification of the class of null sets that in turn are necessary for the definition of strict monotonicity. 
}

 The three axioms of disco aversion, strict monotonicity and continuity thus imply a novel type of probabilistic sophistication, in which the numerical representation of preferences depends on two probability measures $P$ and $Q$. For each act $X$, the probability measure $P$ is applied to elation states while the probability measure $Q$ is applied to disappointment states. The relative weight of the two resulting expectations is then controlled by the agent's endogenous parameter $\alpha$. Also here, it is suggestive to compare with the structure of Gul's preferences, where a different weight is given to expected utilities of elation and disappointment prizes. Gul's preferences are directly defined on lotteries, so by construction a unique probability measure is exogenously given. On the other hand, Gul's preferences indeed are based on a monetary utility, while in our framework for the moment the representing functional is, in a rough sense, a two-probability generalization of the usual expectation. We will incorporate utility functions into our theory in Section \ref{sec:AA}, by defining the preferences on a larger outcome set of compound acts, as in the Anscombe-Aumann (AA) framework.  

Before moving to the AA framework, a natural question arising: If there is a simple, preference-based additional condition sufficient for $Q=P$, that is for the 'collapsing' of duet expectiles on the usual expectiles. In Theorem \ref{th:p-ex} we show that under an additional axiom of event independence, used already in \cite{dF31} to axiomatize subjective probability, this is indeed the case, and thus we provide as a byproduct a novel axiomatization of the solo expectiles. 

In Theorem \ref{th:riskaverse1} we show  that preferences represented by duet expectiles satisfy weak risk aversion with respect to both probability measures $P$ and $Q$. It turns out that strong risk aversion does not hold  either under $P$ or under $Q$, but it holds in a joint sense, as shown in Proposition \ref{prop-jointRA}.

 As mentioned above, in Section \ref{sec:AA} we generalize our setting by considering acts $f \colon \Omega \to \mathcal{C}$, where $\mathcal{C}$ is a convex set of consequences, in the classical AA setting represented by lotteries on a set of prizes. Using the general techniques introduced by \cite{CGMMS11}, under the additional axioms of monotonicity and constant-independence, we find that preference over acts satisfies
\begin{equation*}
f\succsim g\iff\ex_{\alpha}^{P,Q}(u\left(  f\right)  )\geq
\ex_{\alpha}^{P,Q}(u\left(  g\right)  ),
\end{equation*}
where $u \colon \mathcal{C} \to \R$ is a cardinally unique affine utility function. 

 The paper is structured as follows. In Section \ref{sec:axiom} we introduce the main axioms, and Section \ref{sec:main} presents our main characterization results. In Section \ref{sec:solo} we discuss sufficient conditions for solo expectiles. In Section \ref{sec:discuss} we discuss several related issues, i.e., probabilistic sophistication with multiple probabilities, various notions of risk aversion and the interpretation of Gul's preferences as (solo) expectiled utilities.  In Section \ref{sec:AA} we extend the result to an AA framework and discuss duet expectiled utilities in the framework og \cite{GS89}, explicitly determining the set of relevant priors. 
 Section \ref{sec:proofth1} outlines the proofs of our main results, and highlights the technical novelty and challenges. 
 Section \ref{sec:concl} concludes the paper.   The appendices contain full proofs, additional results, and counter-examples.

\section{Setting and axioms}\label{sec:axiom}
Let $(\Omega,\mathcal F)$ be a  measurable space,
and $\mathcal X$ be a {convex} set of real measurable functions from $\Omega$ to $\R$, called acts or random variables, equipped with a norm $\Vert \cdot \Vert$.
That is, we consider the outcome space $\R$. 
The case of a more general outcome space as in the Anscombe--Aumann framework will be studied in Section \ref{sec:AA}.

Specifically, we will consider two cases:
Either $\Omega$ is finite with $n$ elements, where we assume that $n\ge 4$, in which case $\X$ can be identified with $\R^n$ with the maximum norm,
or $\Omega$ is infinite, in which case $\X=L^\infty(\Omega,\mathcal F,\p)$ is a set of all bounded acts under some fixed  reference atomless probability $\p$ endowed with the $L^\infty$-norm as usual.\footnote{A probability space is atomless if there exists a continuously distributed random variable (such as uniformly distributed ones) on this space, which is a minor assumption.} 
The assumption $n\ge 4$ in the finite case is necessary for some proofs\footnote{Specifically, the main characterization result (Theorem \ref{th-main}) does not hold for $n=3$, and Proposition \ref{prop-cv} in Appendix \ref{sec:properties} fail when $n=2$. All these will be explained with counter-examples in Appendix \ref{app:A}.} and harmless for the interpretation of our results.

A \emph{preference relation} is a weak order $\succsim$ on $\X$,\footnote{That is, for $X,Y,Z\in \X$, (a)
either $X\succsim Y$ or $Y\succsim X$; (b) $X\succsim Y$ and $Y\succsim Z$
imply $X\succsim Z$.} with asymmetric part $\succ$  and symmetric part $\sim$.
Constant acts in $\X$ are identified with real numbers when applying $\succsim$.
The order $\ge$ on $\R$ is naturally extended to $\X$ in the point-wise sense. 
Events like $\{\omega\in\Omega: X(\omega)<c\}$ will be written as  
$\{X<c\}$ for simplicity.
We start by introducing the following natural and simple axiom in the case that $\Omega$ is finite.

\renewcommand\theaxiom{SM}
\begin{axiom}[Strict monotonicity]
\label{ax:SM}
If  $X\ge Y$ and $X\neq Y$, then $X\succ Y$.
\end{axiom} 

For the case of an atomless probability space $(\Omega,\mathcal F,\p)$, Axiom \ref{ax:SM} has to be slightly modified. The role of $\p$ is simply to determine the null sets of $(\Omega,\mathcal F)$, and $\p$-almost surely equal acts are identical by definition of $L^\infty(\Omega, \mathcal{F}, \p)$. As a consequence, 
in this setting Axiom \ref{ax:SM} becomes 
$$
\mbox{If  $X\ge Y$ and $\p(X>Y)>0$ then $X\succ Y$.}
$$ 
Axiom \ref{ax:SM} is essential for the proof of all the main characterization results. We will give a counter-example in Appendix \ref{app:SMmain} to illustrate that the representation in Theorem \ref{th-main} does not hold under the weaker axiom of (non-strict) monotonicity, that is $X\ge Y$ implies $X\succsim Y$.

Another standard and natural axiom is continuity.  In this axiom, convergence is with respect to the maximum norm when $\Omega$ is finite, and with respect to the 
$L^\infty$-norm when $(\Omega,\mathcal F,\p)$ is an atomless probability space.

\renewcommand\theaxiom{C}
\begin{axiom}[Continuity]
\label{ax:C}
If $X_n \succsim Y_n$ for each $n\in \N$ and $X_n \to X$ and $Y_n \to Y$, then $X \succsim Y$.
 \end{axiom}  

Axioms \ref{ax:SM} and \ref{ax:C} have a few immediate and well-known consequences 
summarized in the following lemma.  






 
\begin{lemma}\label{lem:S}
For a preference relation $ \succsim $ on $\X$, 
Axioms \ref{ax:SM} and \ref{ax:C} imply
\begin{enumerate}[(i)]
\item[(i)] for each $X\in \X$, there exists a unique $c_X\in \R$ satisfying $X \sim c_X$;
\item[(ii)]   $X \succsim Y$ if and only if $ c_X\ge c_Y$.
\end{enumerate} 
\end{lemma}
The quantity $c_X$ 
is the certainty equivalent of $X$. Recalling the definition of disappointment event given in the Introduction, Lemma \ref{lem:S} shows that, under Axioms \ref{ax:SM} and \ref{ax:C}, it follows that
\begin{align*}
D_X=\{\omega\in\Omega: X(\omega)\prec X\}=\{X<c_X\},
\end{align*}
so the disappointment event of $X$ can be expressed in terms of its certainty equivalent. 

As anticipated in the Introduction, the main new axiom of the paper is \emph{disco aversion}, formally stated below.
Recall that two acts $X$ and $Y$ are 
\emph{disappointment-concordant},  or \emph{disco} for short, if 
$D_X=D_Y$. Under Axioms \ref{ax:SM} and \ref{ax:C} it follows immediately from Lemma \ref{lem:S} that $X$ and $Y$ are disco if and only if $\{X< c_X\}=\{Y<c_Y\}$.

   
\renewcommand\theaxiom{DA}
\begin{axiom}[Disco aversion] 
	\label{ax:DA} 
	If  $X$ and $Y$ are disco and $Y' \sim Y $, then $X+Y' \succsim X+Y$.
\end{axiom} 


The interpretation of disco acts and the motivation for disco aversion have been discussed in the Introduction.
Remarkably, disco aversion is defined for a decision maker who is  not necessarily probabilistically sophisticated in the sense of \cite{MS92}; nevertheless, disco aversion intuitively has some connections with the classic notion of risk aversion (\cite{RS70}). This connection will be revealed  by our main results and discussed in Section \ref{sec:discuss}.

We end the section by collecting further notation used in the rest of the paper. 
We denote by
$\mathcal M_{1}=\mathcal M_{1}(\Omega,\mathcal F)$ the set of all probability measures on $(\Omega,\mathcal F)$. For $A\in\mathcal F$, $\id_{A}$ is the indicator of event $A$.
For $P,Q\in\mathcal M_1$ and $\lambda_1,\lambda_2\ge 0$, $\lambda_1 P\le \lambda_2 Q$ means that $\lambda_1 P(S)\le \lambda_2 Q(S)$ for all $S\in\mathcal F$. The notation $P\ac Q$ means that $P$ and $Q$ are mutually absolutely continuous.
For $P\in\mathcal M_1$ and $X\in\X$, we write $\E^P[X]$ for the expected value of $X$ under $P$.
We write $[n]=\{1,\dots,n\}$ for  $n\in \N$,
 $x_+=\max\{x,0\}$ and 
 $x_-=\max\{-x,0\}$.
 
\section{Representation results}
\label{sec:main}

 \subsection{The case of a finite space}
 
We will first focus on the case of a finite sample space, 
and the case of an atomless space will be studied in Section \ref{sec:general}. 
Let $\Omega=[n]$ with $n\ge 4$, $\mathcal F$ be the power set of $\Omega$,  and  $\X$ be the set of all real random variables on $(\Omega,\mathcal F)$.  
The choice of the norm $\Vert \cdot \Vert$  on $\X$ does not matter in this setting, and as anticipated for simplicity we will take it as the maximum norm.
For a measure $P \in \M$, $P>0$ means that $P(i)>0$ for each $i=1, \dots, n$. 
  
\begin{theorem}\label{th-main}
Let $\Omega= \{1, \dots, n\} $ with $n\ge 4$.
For a   preference  relation $ \succsim$ on $\X$, Axioms \ref{ax:SM}, \ref{ax:C} and \ref{ax:DA} hold if and only if there exist $P,Q\in\mathcal M_1$ and $\alpha \in (0,1/2]$ with $0<\alpha P\le (1-\alpha)Q$
such that  
\begin{align*}
X \succsim Y~ \Longleftrightarrow ~ \ex_\alpha^{P,Q}(X) \ge \ex_\alpha^{P,Q}(Y), 
\end{align*}
where $\ex_\alpha^{P,Q}(X)$ is the unique number $y\in \R$ such that
\begin{align}
\label{eq:def-ex}
	  \alpha \E^P[(X-y)_+]=  (1-\alpha)\E^{Q}[(y-X)_+] .
\end{align}
Further, $P,Q$ and $\alpha$ are uniquely determined.
\end{theorem}


A first remarkable point of Theorem \ref{th-main} is that from a rather parsimonious set of axioms it derives a novel form of probabilistic sophistication, based on two probability measures instead of a single one. 
To interpret the function $\ex_\alpha^{P,Q}$,
the probabilities $P$ and $Q$ represent how the decision maker subjectively evaluates the expected  gain relative to the certainty equivalent in  the elation event and the expected loss in the disappointment event, respectively. In other words, the decision maker can have different subjective probabilities for disappointment and elation events.  
Therefore, the quantity
$\ex_\alpha^{P,Q}(X)$ is a balancing point of disappointment, measured by its expected relative loss weighted by $1-\alpha$,  
and elation,
measured by its expected relative gain weighted by  $\alpha$. For this reason, $\alpha$ is called a disappointment weight parameter.  The condition $\alpha\in (0,1/2]$ reflects the fact that disappointment is weighted on average more than elation.
The condition $0<\alpha P\le (1-\alpha)Q$ further suggests that any state $\omega$ is subjectively weighted more when it belongs to a disappointment event, compared to when it belongs to an elation event. 

 {The quantity $\ex_\alpha^{P,Q}$ can be equivalently defined for any $\alpha \in (0,1)$ and $P,Q \in \M_1$ as
\begin{align}\label{eq-genex}
	\ex^{P,Q}_\alpha(X)=\inf\left\{y\in \R: \alpha\E^P[(X-y)_+]= (1-\alpha) \E^{Q}[(y-X)_+] \right \},~~~X\in \X.
\end{align} 
Indeed, the existence of at least one $y$ satisfying the equality in \eqref{eq-genex} 
follows from the continuity of 
$y\mapsto \E^P[(X-y)_+]$ and 
$y\mapsto \E^Q[(y-X)_+]$ and their limiting behaviour as $y\to \infty$ and $y\to -\infty$. 
The uniqueness of such a $y$ may fail for general $P$ and $Q$, 
so for definiteness the initial infimum has to be added. Nevertheless, under the axioms in  Theorem \ref{th-main}, which lead to the additional condition $0<\alpha P\le (1-\alpha)Q$, $y$ is indeed unique, so $\ex_\alpha^{P,Q}$ can be more simply defined as in \eqref{eq:def-ex}. 
If $Q=P$, the duet expectile $\ex_\alpha^{P,P}$ coincides with the usual expectile introduced by \cite{NP87}, denoted by $\ex_\alpha^P$. It will be called ``solo expectile" in the rest of the paper and discussed further in  Section \ref{sec:solo}.

Interestingly, as it will be shown  in Theorem \ref{th:riskaverse1} of Section \ref{sec:discuss}, the condition $\alpha P \leq (1-\alpha) Q$ turns out to be equivalent to a weak form of risk aversion with respect to ${P}$ and ${Q}$. If we flip the preference in Axiom \ref{ax:DA} by requiring that 
$$
\mbox{if  $X$ and $Y$ are disco and $Y' \sim Y $, then $X+Y' \precsim X+Y$},
$$
\noindent
that is if we replace the axiom of disco aversion with an axiom of \emph{disco seeking},
by the same proof of Theorem \ref{th-main} we get a characterization of 
$\ex_\alpha^{P,Q}$ with $\alpha P\ge (1-\alpha) Q>0$ and $\alpha \in [1/2,1)$.

 We give an overview of the proof  of Theorem \ref{th-main} in Section \ref{sec:proofth1} and the details are reported in Section \ref{subsubsec:Th1}. As far as we know, the proof techniques are very different from existing ones in the literature and involve new mathematical treatment of consistency across multiple probability measures.

\subsection{The case of a standard probability space}
\label{sec:general}

Next, we focus on the case in which $(\Omega,\mathcal F,\p)$ is a standard probability space, meaning that there exists a uniform random variable $V \colon \Omega \to (0,1)$ such that $\sigma(V)=\mathcal F$, where $\sigma(V)$ is the $\sigma$-algebra generated by $V$. Standard probability spaces are special cases of atomless probability spaces. The assumption of a standard probability space might be not necessary for the validity of Theorem \ref{th:main-general}, but it is required within the context of our current proof techniques.


To derive the main characterization result in a standard probability space, we introduce an additional continuity axiom as follows, which holds trivially if $\Omega$ is finite.

\renewcommand\theaxiom{WMC}
\begin{axiom}[Weak monotone continuity]
\label{ax:MC}
If $m\in\R$ and $\{A_n\}_{n\in\N}\subseteq \mathcal F$ with $A_1\supseteq A_2\supseteq\dots$ and $\bigcap_{n\in\N}A_n=\varnothing$, 
then there exists $n_0\in\N$ such that $m\id_{A_{n_0}}+\id_{A_{n_0}^c}\succ 0$.
 \end{axiom}  


Axiom \ref{ax:MC} is a weaker version of monotone continuity, where the latter was introduced by \cite{A70} to obtain a Subjective Expected Utility representation that utilizes a $\sigma$-additive subjective probability measure. Specifically, let $m$ and $\{A_n\}_{n\in\N}$ be defined in Axiom \ref{ax:MC}. Monotone continuity implies that if $X,Y\in\X$ and $X\succ Y$, then there exists $n_0\in\N$ such that $m\id_{A_{n_0}}+X\id_{A_{n_0}^c}\succ Y$. It is not hard to see that monotone continuity, when letting $X=1-Y=1$ and considering a preference relation satisfying $1\succ 0$, naturally leads to Axiom \ref{ax:MC}. Monotone continuity has been instrumental in restricting the set of priors in multiple priors models to $\sigma$-additive measures, as evidenced in the works of \cite{CMMT05} for $\alpha$-maximin expected utility preferences, \cite{MMR06} for variational preferences, and \cite{S11} for multiplier preferences.
Below we present our main characterization result in this subsection.

\begin{theorem}\label{th:main-general}
Let $(\Omega,\mathcal F,\p)$ be a standard probability space. 
For a preference relation $ \succsim$ on $\X$, 
Axioms \ref{ax:SM}, \ref{ax:C}, \ref{ax:MC} and \ref{ax:DA} hold if and only if there exist $P,Q\in\mathcal M_1$ and $\alpha\in(0,1/2]$ with $\alpha P\le (1-\alpha)Q$ and $P\ac Q\ac \p$
such that $\succsim$ is represented by $\ex_{\alpha}^{P,Q}$ defined as in \eqref{eq:def-ex}.
Further, $P,Q$ and $\alpha$ are uniquely determined.
\end{theorem}

\section{The solo expectiles}
\label{sec:solo}

Solo expectiles are a popular class of risk functions introduced in \cite{NP87} as an asymmetric generalization of the mean.  
 As we already mentioned, for $P\in \mathcal M_1$, the solo expectile at level $\alpha \in (0,1)$ denoted by $\ex_\alpha^P(X)$ is the unique $y\in \R$ such that
\begin{equation*}
\alpha \E^P[(X-y)_+]=  (1-\alpha)\E^{P}[(y-X)_+].   
\end{equation*} 


The duet expectile obtained in Theorem \ref{th:main-general} involves two measures. 
If we further assume  probabilistic sophistication under $\p$ (that is, all identically distributed acts are equally preferable), 
then the characterization in Theorem   \ref{th:main-general}
leads to a solo expectile  with respect to the objective probability measure $\p$. 
Probabilistic sophistication 
is usually formulated through subjective probability measures (\cite{MS92}), and it will be further discussed in Section \ref{sec:probab}.


\begin{proposition}
\label{prop:classic}
Let $(\Omega,\mathcal F,\p)$ be a standard probability space. 
For a  preference relation $ \succsim$ on $\X$, probabilistic sophistication and Axioms \ref{ax:SM}, \ref{ax:C} and \ref{ax:DA} hold if and only if there exists  $\alpha \in(0,1/2]$ 
such that $\succsim$ is represented by $\ex_\alpha^{\p}$.

\end{proposition}

Next, we show how it is possible to pin down the solo expectile among duet expectiles without a prior requiring probabilistic sophistication, by means of the following additional axiom.

\renewcommand\theaxiom{EI}
\begin{axiom}[Event independence] 
	\label{ax:EI}   
 For disjoint $A,B,C\in \mathcal F$, $\id_A  \succsim \id_B \Longrightarrow \id_{A\cup C}  \succsim \id_{B\cup C}$.
\end{axiom}  

Axiom \ref{ax:EI} is the key property used by \cite{dF31} to rationalize subjective probability. 
This axiom means that if an event $A$ is seen as more likely (or favorable) than another event $B$, then adding a disjoint event $C$ to both does not change the preference between the two events. 
Axiom \ref{ax:EI} alone does not imply probabilistic sophistication, but together with the other axioms in our framework, it is sufficient for pinning down preferences based on the solo expectiles. 

\begin{theorem}\label{th:p-ex}
Let $(\Omega,\mathcal F,\p)$ be a standard probability space. 
For a  preference  relation $ \succsim$ on $\X$, Axioms \ref{ax:SM}, \ref{ax:C}, \ref{ax:MC}, \ref{ax:DA} and \ref{ax:EI} hold if and only if there exist $P\in\mathcal M_1$ and $\alpha \in(0,1/2]$ with 
$P\ac \p$
such that $\succsim$ is represented by $\ex_\alpha^{P}$. 
\end{theorem}



Being a special case of the duet expectiles,
solo expectiles are widely studied in the statistical literature and the risk measure literature. 
First, solo
expectiles are an asymmetric generalization of the mean, that corresponds to the case $\alpha=1/2$, which was the original motivation of \cite{NP87}.
Second, they belong to the  class of so-called coherent risk measures (\cite{ADEH99}). For further properties of solo expectiles we refer e.g., to \cite{BKMR14,BDB17}. 
Third, solo expecitles are a class of elicitable functionals. 
Elicitability of a statistical functional means that it can be defined as the minimizer of a suitable expected loss,
and it is useful in statistics and machine learning; see \cite{G11} and \cite{FK21}.    
Solo expectiles  admit a well-known axiomatic characterization established by \cite{Z16}  
as the only coherent risk measures that are elicitable. As a consequence of elicitability, solo expectiles satisfy also
the following property, called convex level sets (CxLS) by \cite{DBBZ16}:
$$
\ex_\alpha^P(X) = \ex_\alpha^P(Y)  \implies \ex_\alpha^P(X \id_A + Y \id_{A^c}) = \ex_\alpha^P(X),
$$
whenever $X$, $Y$ and $\id_A$ are independent. 
The CxLS property expresses convexity of the level sets with respect to mixtures and is strictly related to the betweenness axiom introduced in \cite{C89} and \cite{D86} as a possible weakening of the independence axiom of the Von Neumann and Morgenstern expected utility theory. As it has been shown in \cite{DBBZ16} the CxLS property also axiomatizes solo expectiles among 
coherent risk measures under probabilistic sophistication.

Our Proposition \ref{prop:classic} shows that assuming probabilistic sophistication and other standard properties, solo expectiles can be characterized through 
the property
\begin{equation}\label{eq:aversion-rm}
U(X+Y')\ge U(X+Y),
\end{equation}
where $U$ is the unique certainty equivalent of $\succsim$, and $X,Y$ are disco satisfying $U(Y')=U( Y)$. 
As shown by \cite{BCCT21},
$\ex_\alpha^P (X+Y)=\ex_\alpha^P (X)+\ex_\alpha^P (Y)$
if $X$ and $Y$ are disco, and hence \eqref{eq:aversion-rm} corresponds to superadditivity. 
A related result holds for the standard risk measure in financial regulation, the Expected Shortfall  (ES, also known as CVaR), in the sense that the class of ES is also characterized by \eqref{eq:aversion-rm} where $X,Y$ are \emph{concentrated risks} satisfying $U(Y')=U( Y)$; see \cite{WZ21} and \cite{HWWW21} for precise definitions and results.  
The notion of concentrated risks is objective (specified by their joint distribution), whereas our concept of disco aversion is subjective as it is specific to the preference relation itself.

\section{Discussions}
\label{sec:discuss}

 
\subsection{Probabilistic sophistication with multiple probabilities}\label{sec:probab}
The definition of duet expectiles with probability measures as in  \eqref{eq-genex}
leads to the simple fact that a duet expectile is a scenario-based risk functional in the sense of \cite{WZ21a}.
For a collection of probability measures $\mathcal Q\subseteq \M_1$, we write $X \laweq_{\mathcal Q} Y$ to indicate that $X$ and $Y$ are identically distributed under each $P\in \mathcal Q$. 
A \emph{$\mathcal Q$-based}
risk functional is a mapping $\rho \colon \X\to \R$   such that $\rho(X)=\rho(Y)$ whenever $X \laweq_{\mathcal Q} Y$; we will also use this term for the corresponding preference relation.
When  $\Q=\{P\}$, such a functional is called law-based (under $P$) in the literature.
It is clear that $\ex_\alpha^{P,Q}$ is a $(P, Q)$-based risk functional.
Therefore, Theorems \ref{th-main} and \ref{th:main-general} are the first results in the literature that give rise to an axiomatic characterization of a  $\mathcal Q$-based risk functional for a non-singleton $\mathcal Q$. Moreover, keeping in mind that the role played by $\mathbb{P}$ is purely instrumental to the formulation of Axiom \ref{ax:SM}, we have that $ P$ and $ Q$ are endogenous; i.e., they correspond to subjective probabilities of the decision maker.

The above observation inspires  a generalized notion of  probabilistic sophistication.
For $n\in \N$, 
we say that a preference $ \succsim$ is \emph{$n$-probabilistically sophisticated} if 
there exists a set $\mathcal Q$ of cardinality $n$ such that $X \sim Y$ whenever $X\laweq _{\mathcal Q} Y$ (thus, it is $\mathcal Q$-based).
 Note that the case of $n=1$ is the notion of probabilistic sophistication defined by \cite{MS92}. 
 Intuitively, $n$-probabilistic sophistication means that the decision maker has a collection of $n$ subjective probabilities that jointly determine  her choices over acts. 
With this terminology, from Theorems \ref{th-main} and \ref{th:main-general} it follows that Axioms \ref{ax:SM}, \ref{ax:C},
 \ref{ax:MC} (on a standard probability space), and  \ref{ax:DA}
together 
imply $2$-probabilistic sophistication. 
 
Classic models of subjective probability, such as those of \cite{dF31} and \cite{S54}, 
yield preferences satisfying $1$-probabilistic sophistication, i.e., the existence of a probability that governs decisions. 
Many popular models of preferences under ambiguity, such as the max-min expected utility model of \cite{GS89},
 the  smooth ambiguity model of \cite{KMM05}, and
the variational preferences of \cite{MMR06}, involve a set of probabilities, called the ambiguity set. These preferences are $n$-probabilistically sophisticated if the ambiguity set has  $n$ elements. 
The axioms for these models do not identify the cardinality of the ambiguity set, in sharp contrast to Theorem \ref{th-main}. As far as we are aware, Theorem \ref{th-main} is the first result that yields $n$-probabilistic sophistication for some fixed $n\ge 2$.

In the context of investment  decisions in a complete financial market, $2$-probabilistic sophistication also appears naturally.
In such a setting, an expected-utility decision maker's preferences over acts, representing here terminal asset payoffs, are determined by 
a (possibly subjective) ``real-world" probability $P$, and a ``risk-neutral" probability   $Q$ induced by the market prices. 
Using the above formulation, the decision maker is indifferent between two assets that have the same distribution under $P$ (implying the same expected utility) and the same distribution under $Q$ (implying the same price); thus $2$-probabilistic sophistication holds. 
 
\subsection{Risk aversion and convexity}
\label{sec:52}

We first recall the notion of strong risk aversion of \cite{RS70}.
For a probability $P$, 
a $P$-based preference or its numerical representation is \emph{strongly risk averse}  under  $P$ if it is monotone with respect to second order stochastic dominance under $P$; that is, for $X,Y\in\mathcal X$, $X\ge_{\rm ssd}^P Y$ implies $X\succsim Y$, where $X\ge_{\rm ssd}^P Y$ means that $\E^P[u(X)]\ge \E^P[u(Y)]$ for all increasing and concave utility functions $u \colon \R \to \R$ such that both expectation exist.

The preference relation represented by $\ex_\alpha^{P}$ for $\alpha \in (0,1/2]$ is strongly risk averse.\footnote{See e.g., Theorem 5 of \cite{BKM18}.}
Assume  Axioms \ref{ax:SM} and \ref{ax:C}  for $\succsim$ that is law-based under $P$.
By Proposition \ref{prop:classic},  disco aversion implies risk aversion $(\alpha \le 1/2)$,
and conversely, disco seeking implies risk seeking $(\alpha \ge 1/2)$.\footnote{Moreover, $\ex_\alpha^{P}$ is convex if and only if $\alpha \in [1/2,1)$, and it is concave if and only if $\alpha \in (0,1/2]$. In Section \ref{sec:properties} we will see that a similar statement holds for $\ex_{\alpha}^{P,Q}$.}
Therefore, the concept of disco aversion is intimately connected to risk aversion for risk preferences. 

Since the duet expectile preferences are $2$-probabilistically sophisticated, the traditional notions of risk aversion are not applicable. Nevertheless, in Theorem \ref{th:riskaverse1} below we show that duet expectile preferences are uncertainty averse and weakly risk averse in a natural sense.
 
We say that a preference relation $ \succsim$ is \emph{weakly risk averse} (resp.~\emph{weakly risk seeking}) under a probability measure $P$ if $\E^P[X] \succsim X$  (resp.~$ \E^P[X]\precsim X$) for all $X\in \X$.  

We say that a preference relation $\succsim$ is \emph{uncertainty averse} (resp.~\emph{uncertainty seeking}) or simply
\emph{convex} (resp.~\emph{concave}) if $X\sim Y$ implies $\lambda X+(1-\lambda)Y\succsim X$ (resp.~$X\sim Y$ implies $X\succsim\lambda X+(1-\lambda)Y$) for all $\lambda\in[0,1]$. 
In the context of decision making under ambiguity, this property has been introduced by  \cite{GS89} and is fundamental in other decision models such as that of  \cite{MMR06};
see also the related discussions of  \cite{E99} and \cite{GM02} on ambiguity aversion defined differently.
Note that convexity of a preference relation is equivalent to the quasi-concavity of its numerical representation, which is again equivalent to concavity under translation invariance (see, e.g., \cite{CMMM11}).


Based on Theorem \ref{th-main} or Theorem \ref{th:main-general}, the following result shows that 
Axioms \ref{ax:SM}, \ref{ax:C}, \ref{ax:MC} (on standard probability space) and \ref{ax:DA} imply both weak risk aversion under  $ P$ and $ Q$
and convexity.

\begin{theorem}
\label{th:riskaverse1}
Suppose that $\succsim$ is represented by $\ex_{\alpha}^{P,Q}$ for some $P,Q\in \M_1$ and $\alpha\in(0,1)$, and Axiom \ref{ax:SM} holds.
The following are equivalent:
\begin{enumerate}[(i)]
\item[(i)] $\succsim$ 
is weakly risk averse  under   $ P$;
\item[(ii)] $\succsim$ 
is weakly risk averse  under   $ Q$;
\item[(iii)] $\succsim$ is convex;
\item[(iv)] $\alpha\in(0,1/2]$ and $\alpha P\le (1-\alpha)Q$.
\end{enumerate}  
\end{theorem}

\begin{remark}
In Theorem \ref{th:riskaverse1},  weak  risk aversion and uncertainty aversion can be replaced by weak risk seeking and uncertainty seeking, respectively,
if $\alpha\in(0,1/2]$ and $\alpha P\le (1-\alpha)Q$ are replaced by $\alpha\in[1/2,1)$ and $\alpha P\ge(1-\alpha) Q$, following the same proof.
Moreover, the implications from (iv)  to (i), (ii) and (iii) hold on a general space without Axiom \ref{ax:SM}, while Axiom \ref{ax:SM} is necessary for the converse—inferring (iv) from any of (i), (ii), or (iii) (see Appendix \ref{app:SMprop} for the detailed explanation).
\end{remark}



In contrast to weak risk aversion in Theorem \ref{th:riskaverse1}, the  preference $\succsim$ represented by duet expectile $\ex_\alpha^{P,Q}$ does not necessarily satisfy strong risk aversion under either $ P$ or $ Q$.  In fact, strong risk aversion   under $P$   would imply that $\succsim$ is law-based under $ P$, which $\ex_\alpha^{P,Q}$ does not satisfy  unless $P=Q$   (see Proposition \ref{prop:classic}). 


Although strong risk aversion does not hold under either $P$ or $Q$, a special form of joint strong risk aversion holds. 
We say that a preference relation  $\succsim $ is \emph{strongly risk averse jointly under $(P,R)\in \mathcal M_1\times \mathcal M_1$} 
if $X\ge_{\rm ssd}^P Y$
and 
$X\ge_{\rm ssd}^R Y$
imply $X \succsim Y$. 

\begin{proposition}\label{prop-jointRA}
Suppose that $\succsim$ is represented by $\ex_\alpha^{P,Q}$ for some $P,Q\in \M_1$ and $\alpha\in(0,1/2)$  satisfying $\alpha P\le (1-\alpha)Q$ and $P\ne Q$. 
Then, $\succsim $ is strongly risk averse jointly under $( P,(Q-\alpha(P+Q))/(1-2\alpha))$.
\end{proposition}

Jointly strongly risk averse preference is a robust preference in the sense of \cite{DR10}. 
For a set of probability measures $\mathcal Q$,
\cite{DR10} introduced
a type of risk aversion characterized by the condition that $X\ge_{\rm ssd}^P Y$ for all $P\in\mathcal Q$ implies $X\succsim Y$, which is named as $\mathcal Q$-risk aversion in our context.
\cite{DR10} focused on stochastic optimization problems where the constraint follows the principle of $\mathcal Q$-risk aversion. Clearly, $\mathcal Q$-risk aversion is equivalent to joint strong risk aversion when $\mathcal Q$ is composed of two elements.
Therefore, 
Proposition \ref{prop-jointRA} illustrates that Theorems \ref{th-main} and \ref{th:main-general} together imply $\{P,R\}$-risk aversion, where $P,R$ are two endogenous measures.

A natural question arises as to whether $\ex_\alpha^{P,Q}$ satisfies joint strong risk aversion under $(P,Q)$. We find that it does not, and a counter-example is provided in Appendix \ref{app:jointSRA}.

 \subsection{Gul's preferences as expectiled utilities}
\label{subsec:Gul}
We have already pointed out a number of similarities and differences with respect to the theory of disappointment averse preferences in \cite{G91}. Gul's preferences are defined on lotteries that are finitely supported probability measures on a compact set of monetary prizes. We refer to \cite{CDO20} for a more general treatment and for additional results and interpretations. The numerical representation of preferences obtained by Gul has the following form: 
$$
U(p) = \gamma(\lambda) \sum_{x} u(x)q(x) +  ( 1-\gamma(\lambda) ) \sum_{x} u(x)r(x),
$$
where $p$ is a lottery whose elation-disappointment decomposition is
$
(\lambda,q,r),
$\footnote{Suppose that $\lambda\in[0,1]$, and $p,q,r$ are lotteries.  $(\lambda,q,r)$ is an elation-disappointment decomposition of $p$ if $q\in B(p)$, $r\in W(p)$ and $\lambda q+(1-\lambda)r=p$, where $B(p)=\{s: x~{\rm in~the~support~of~}s~{\rm implies}~x\succsim p\}$ and $W(p)=\{s: x~{\rm in~the~support~of~}s~{\rm implies}~p\succsim x\}$.}
$u : \R \to \R$ is a cardinally unique continuous monetary utility, and $\gamma \colon [0,1] \to [0,1]$ is given by
$$
\gamma(\lambda) = \frac{\lambda}{1+\beta(1-\lambda)},
$$
with $\beta \in (-1, \infty)$ is a disappointment aversion coefficient. To our knowledge, it has not been  recognized in the literature that Gul's preferences can be equivalently formulated as expectiled utilities. Indeed, a useful way of representing expectiles is the following formula:
\begin{equation}
\ex_\alpha^P(X)  = \frac{\alpha \E^P\left[X \id_{\{X\ge \ex_\alpha^P(X)\}}\right]+
(1-\alpha) \E^P\left[X \id_{\{X\le \ex_\alpha^P(X)\}}\right]}
{\alpha P(X \ge \ex_\alpha^P(X)) + (1-\alpha) P(X \le \ex_\alpha^P(X))} 
\end{equation}
(see e.g. \cite{BKMR14}),
and with straightforward calculations it follows that Gul's representing functional can be written as
$$
U(p) = \ex_\alpha^P(u(X)),
$$
where $\alpha = 1/(2 + \beta)$ and $X$ is any random variable having distribution $p$ under $P$. If $\beta=0$ then $\alpha=1/2$, consistently having expected utility as a special case of expectiled utility.  
In the next section we show how to derive duet expectiled utilities in a more general AA framework, thus to include Gul's preferences into our model formulated in the AA framework.  

\section{The Anscombe--Aumann framework}
\label{sec:AA}

\subsection{Setup} 

We first describe a standard extension of the Anscombe and Aumann setup as in \cite{CGMMS11}.
Extending the setting in Section \ref{sec:axiom}, we continue to 
work with the set $\Omega$ of  {states of the world} and the $\sigma$-algebra $\mathcal F$
of subsets of $\Omega$. Moreover, we consider a convex set $\C$ of
\emph{consequences} or \emph{outcomes}. We denote by $B_{0}\left(  \C\right)  =B_{0}\left(
\Omega,\mathcal F,\C\right)  $, the set of all the \emph{(simple) acts}: finite-valued
functions $f:\Omega\rightarrow \C$ which are $\mathcal F$-measurable.  
 {We endow $B_0(\C)$ with the maximum norm}.
Analogously, we
denote by $B_{0}\left(  \mathbb{R}\right)  $ the set of all real-valued
$\mathcal F$-measurable simple functions, so that $u\left(  f\right)  \in
B_{0}\left(  \mathbb{R}\right)  $ for all $u:\C\rightarrow\mathbb{R}$.
Given any $x\in \C$, we denote by $x\in B_{0}\left(  \C\right)  $ also the
constant act such that $x(\omega)=x$ for all $\omega\in \Omega$. That is, we  identify $\C$ with the subset of the constant acts in
$B_{0}\left(  \C\right)  $. 
When $\mathcal C=\R$, we are back to the setting of Section \ref{sec:axiom}, with the additional assumption that acts are finite-valued (which always holds true in finite spaces).

We model the decision maker's \emph{preference relation}  
by a weak order $\succsim$ on $B_{0}\left(  \C\right) $.  
In this setting, disappointment concordance with respect to $\succsim$ has a similar definition, that is, $f,g\in B_0(\mathcal C)$ are disappointment-concordant (disco for short) if $\{\omega\in \Omega: f(\omega)\prec f\}=\{\omega\in \Omega: g(\omega)\prec g\}$.

Let $\Gamma \subseteq \R$ be an interval. We say that
a functional $I \colon B_0(\Gamma) \to \R$ is 
\emph{monotone}, if $f \geq g$ implies $I(f) \geq I(g)$;
 \emph{continuous}, if it is continuous with respect to the maximum norm;
\emph{normalized}, if $I(c)=c$ for all $c \in \Gamma$.

\subsection{Axioms and representation} 

 We introduce several axioms as outlined in \cite{CGMMS11}.





\renewcommand\theaxiom{AA-M}
\begin{axiom}[Monotonicity]
\label{ax:AA-M}
If $f,g\in B_{0}\left(  \C\right)  $ and
$f(\omega)\succsim g(\omega)$ for all $\omega\in \Omega$, then $f\succsim g$.
\end{axiom} 

\renewcommand\theaxiom{AA-I}
\begin{axiom}[Independence]\label{ax:AA-I} 
If $x,y,z\in \C$ and $\lambda\in(0,1]$, then $x\succ y$ implies $\lambda x+(1-\lambda)z\succ \lambda y+(1-\lambda)z$.

\end{axiom} 

\renewcommand\theaxiom{AA-C}
\begin{axiom}[Continuity]
\label{ax:AA-C} 
If $f,g,h\in B_{0}\left(  \C\right)$ and $f\succ g\succ h$, then there are $\alpha,\beta\in(0,1)$ such that $\alpha f+(1-\alpha)h\succ g\succ \beta f+(1-\beta)h$.

\end{axiom}

\begin{remark}
Axiom \ref{ax:AA-I} and Axiom \ref{ax:AA-C} are called  risk independence and Archimedean properties in \cite{CGMMS11}. These  two axioms can  be   replaced by the following properties: 
\begin{enumerate}[(a)] 
\item[(a)] For $x, y, z \in \C$ and $x \sim y$, it holds that  $x/2 + z/2 \sim y/2 + z/2$; 
\item[(b)] For $f, g, h \in B_{0}(\C)$, the sets $\{\alpha \in [0, 1] : \alpha f + (1 - \alpha) g \succsim h\}$ and $\{\alpha \in [0, 1] : h \succsim \alpha f + (1 - \alpha) g\}$ are closed. 
\end{enumerate}
These two properties were originally introduced by \cite{HM53} for the mixture space theorem.
\end{remark}

\begin{lemma}[Proposition 1 of \cite{CGMMS11}]
\label{lm-CGMMS11}
 A preference relation $\succsim$\ on $B_{0}\left(
\C\right)  $ satisfies Axioms \ref{ax:AA-M}, \ref{ax:AA-I} and \ref{ax:AA-C} if and only if there exist an affine function
$u:\C\rightarrow\mathbb{R}$ and a monotone, continuous, and normalized function
$I:B_{0}\left(  u\left(  \C\right)  \right)  \rightarrow\mathbb{R}$ such that,
given any $f,g\in B_{0}\left(  \C\right)$,
$$
f\succsim g\iff I\left(  u\left(  f\right)  \right)  \geq I\left(  u\left(
g\right)  \right)  .
$$
In this case, $u$ is cardinally unique, and $I$ is unique given $u$.
\end{lemma}
Next, we formulate the axiom of disco aversion in the AA setting, which has a slightly different form compared to Axiom \ref{ax:DA}, because we replaced the sum by a convex combination to work with a convex set of acts.




\renewcommand\theaxiom{AA-DA}
\begin{axiom}[Disco aversion]
\label{ax:AA-DA} 
If $f,g,h\in B_{0}\left(  \C\right)
$ are such that $f$ and $g$ are disco, then  
$$
g\sim h~\Longrightarrow~\alpha f+\left(  1-\alpha\right)  h\succsim\alpha f+\left(
1-\alpha\right)  g, \mbox{~~~for all $\alpha\in\left(  0,1\right)  $.}
$$ 

\end{axiom} 

Finally, we impose the usual non-degeneracy axiom to exclude trivial preferences. 
\renewcommand\theaxiom{AA-ND}
\begin{axiom}[Non-degeneracy]
\label{ax:AA-ND} 
$f\succ g$ for some $f,g\in B_{0}\left(\C\right)$.
\end{axiom}



As in Section \ref{sec:axiom}, we assume that either $\Omega$ is finite with $n$ elements satisfying $n\ge 4$ or $\Omega$ is infinite with a fixed reference probability measure $\p$ such that $(\Omega,\mathcal F,\p)$ is a standard probability space. In the infinite case, $\p$-almost surely equal acts are seen as the identical. Below we present a stronger version of Axiom \ref{ax:AA-M}, which will be used in Theorems 
\ref{thm:diesel} and   \ref{thm:AASPS}.

\renewcommand\theaxiom{AA-SM}
\begin{axiom}[Strict monotonicity]
\label{ax:AA-SM}
If $f,g\in B_{0}\left(\C\right)$, then  $f(\omega)\succsim g(\omega)$ for all $\omega\in \Omega$ and $f(\omega)  \succ g (\omega) $ for some $\omega\in \Omega$ (in  a standard probability space, this becomes $\p(\omega\in \Omega: f (\omega) \succ g (\omega))>0$) imply $f\succ g$.
\end{axiom}


\begin{theorem}\label{thm:diesel}
Suppose that $\Omega$ is finite.
Let $\succsim$ be a preference relation on $B_{0}\left(
\C\right) $. The following conditions are equivalent:

\begin{enumerate}
\item[(i)] $\succsim$ satisfies Axioms \ref{ax:AA-SM}, \ref{ax:AA-I}, \ref{ax:AA-C}, \ref{ax:AA-DA} and \ref{ax:AA-ND}.

\item[(ii)] There exists a non-constant affine function $u:\C\rightarrow
\mathbb{R}$, two stricly positive probability measures $P$ and $Q$, and a
disappointment-weight $\alpha\in\left(0,1/2\right]$ with $0<\alpha P\leq(1-\alpha)Q$, such that given any $f,g\in B_{0}\left( \C\right)  $,%
\begin{equation}
f\succsim g\iff\ex_{\alpha}^{P,Q}(u\left(  f\right)  )\geq
\ex_{\alpha}^{P,Q}(u\left(  g\right)  ).
\end{equation}
In this case, $u$ is cardinally unique, while $P$, $Q$, and $\alpha$ are unique.
\end{enumerate}
\end{theorem}

Similar to Section \ref{sec:general}, for the infinite case we introduce an  extra continuity axiom.

\renewcommand\theaxiom{AA-MC}
\begin{axiom}[Monotone continuity]
\label{ax:AA-MC}
For all $x,y,z\in\mathcal C$ such that $x\succ y$, if $\{A_n\}_{n\in\N}\subseteq \mathcal F$ satisfies $A_1\supseteq A_2\supseteq\dots$ and $\bigcap_{n\in\N}A_n=\varnothing$, 
then $z\id_{A_{n_0}}+x\id_{A_{n_0}^c}\succ y$ for some $n_0\in \N$.
\end{axiom}


\begin{theorem}\label{thm:AASPS}
Suppose that $(\Omega,\mathcal F,\p)$ is a standard probability space.
Let $\succsim$ be a preference relation on $B_{0}\left(
\C\right) $. The following conditions are equivalent:

\begin{enumerate}
\item[(i)] $\succsim$ satisfies Axioms \ref{ax:AA-SM}, \ref{ax:AA-I}, \ref{ax:AA-C}, \ref{ax:AA-DA}, \ref{ax:AA-ND} and \ref{ax:AA-MC}.

\item[(ii)] there exists a non-constant affine function $u:\C\rightarrow
\mathbb{R}$, two probability measures $P$ and $Q$, and a
disappointment-weight $\alpha\in\left(0,1/2\right] $ with $P\ac Q\ac \p$ and $\alpha P\leq(1-\alpha)Q$, such that given any $f,g\in B_{0}\left(  \C\right)  $,%
\begin{equation}
\label{numrepAA}
f\succsim g\iff\ex_{\alpha}^{P,Q}(u\left(  f\right)  )\geq
\ex_{\alpha}^{P,Q}(u\left(  g\right)  ).
\end{equation}
In this case, $u$ is cardinally unique, while $P$, $Q$, and $\alpha$ are unique.
\end{enumerate}
\end{theorem}

If we go back to the setting of Section \ref{sec:axiom} and restrict the preference relation in \eqref{numrepAA} to deterministic lotteries, that is if we consider only acts such that $f(\omega) = \delta_{X(\omega)}$ with some finite-valued function $X:\Omega\to \R$, then the preference induced on $\X$ has the numerical representation 
$$
X\succsim Y\iff\ex_{\alpha}^{P,Q}(v\left(  X\right)  )\geq
\ex_{\alpha}^{P,Q}(v\left(  Y\right)  ),
$$
where
$
v(x)= u(\delta_x)\mbox{~for all $x\in\R$}.
$
Hence, if we require axioms in Theorem \ref{thm:AASPS} on an enlarged outcome set, we arrive at a preference relation that includes as a special case duet expectiled utilities, which in turn include as a special case Gul's expectiled utility, as discussed in Section \ref{subsec:Gul}.

\subsection{Duet expectiled utility as min-max utility preferences} \label{subsec:GS}
For a preference relation, 
it is straightforward to see that all axioms in statement (i) of Theorem \ref{thm:AASPS} imply axioms A.1--A.6 in \cite{GS89},  
the main point being that disco aversion implies uncertainty aversion,
which can be seen by taking $f=g$ in Axiom \ref{ax:AA-DA}; see also results in Section \ref{sec:52}. As a consequence, duet expectiled utilities have a representation as worst-case expected utilities.
Moreover, the set of relevant priors can be obtained explicitly by means of Proposition \ref{prop:dual} in Appendix \ref{sec:DR}. Precisely, when $\alpha P\le (1-\alpha)Q$, it holds that
$$
\ex_{\alpha}^{P,Q}(u(f)) = \inf_{R\in \mathcal R}\E^R[u(f)],
$$
where
$$\mathcal R =\left \{R\in \M_1, R \ll  \p,  (1-\alpha)\frac{\d R}{\d P}(\omega) \ge   \alpha\frac{\d R}{\d Q}(\omega') \mbox{~~$P\times Q$-almost surely}\right\},
$$
where $\d R/\d P$ and $\d R/\d Q$ denote Radon-Nikodym derivatives.
Therefore, the duet expectiled utility preferences belong to  the class of
min-max utility preferences of \cite{GS89}, and hence they are 
also special cases of the variational preferences of \cite{MMR06}.

\section{Overview of proofs of the main theorems}\label{sec:proofth1}


We first provide an overview of the proofs of our main results, Theorems \ref{th-main} and \ref{th:main-general}, which will help to prove other theorems in Sections \ref{sec:solo} and \ref{sec:AA}. 
The full proofs of these two theorems are presented in Appendix \ref{app:proofs}.
A brief sketch of the proofs of Theorems \ref{th:p-ex} and \ref{th:riskaverse1} will be explained later in the section.
Theorems \ref{thm:diesel} and \ref{thm:AASPS} 
are shown by first using Lemma \ref{lm-CGMMS11} to   convert the setting of real outcomes used in Theorems \ref{th-main}--\ref{th:riskaverse1}
to the AA framework, and then proving positive homogeneity that allows us to connect Axioms \ref{ax:DA} and   \ref{ax:AA-DA}. 
The detail proofs of Theorems \ref{thm:diesel} and \ref{thm:AASPS} are presented in Appendix \ref{sec:proof-AA}.

\subsection{Two lemmas and some properties of duet expectiles}\label{sec:7.1}
 
We first present two lemmas, 
Lemmas \ref{lm-property} and   \ref{prop-uniquedef}, which will appear in the proofs of all theorems.  Lemma \ref{lem:S} guarantees the existence of a unique certainty equivalent based on the stated axioms, which is denoted by $U$.

\begin{lemma}\label{lm-property}
Assume that 
$(\Omega,\mathcal F,\p)$ is finite with $\p>0$ or $(\Omega,\mathcal F,\p)$ is an atomless probability space. 
Axioms \ref{ax:SM}, \ref{ax:C} and \ref{ax:DA} imply the following properties of  $U$.
\begin{itemize}
	\item [(i)] 
Internality:  $U(X)\in(\essinf_{\p} X,\esssup_{\p} X)$ for all nondegenerate $X\in\mathcal X$; 

	\item[(ii)] Translation invariance: $U(X+m)=U(X)+m$ for all $X\in\X$ and $m\in\R$;
	\item[(iii)] Lipschitz continuity: $|U(X)-U(Y)|\le \|X-Y\|$ for all  $X,Y\in\mathcal X$;
	\item[(iv)]  Superadditivity: $U(X+Y)\ge U(X)+U(Y)$ for all $X,Y\in\mathcal X$;
	\item[(v)] Disco additivity: If $\{X<U(X)\}=\{Y<U(Y)\}$, then $U(X+Y)=U(X)+U(Y)$;	\item[(vi)] Positive homogeneity: $U(\lambda X)=\lambda U(X)$ for all $X\in\mathcal X$ and $\lambda> 0$.
\end{itemize}
\end{lemma}

Although nontrivial, the proof of Lemma \ref{lm-property} follows from some standard techniques in decision theory; see Appendix \ref{app:pf-lem} for details.

We denote by $\mathcal M=\mathcal M(\Omega,\mathcal F)$ 
the set of all non-zero $\sigma$-additive measures on $(\Omega,\mathcal F)$.
For $P\in\mathcal M$, let $\widetilde{P}\in\mathcal M_1$ be the normalized version of $P$, 
i.e., $\widetilde{P}=P/{P(\Omega)}$, and we write $\E^P[X]$ for the integral of $X$ with respect to $P$. 
For $P,Q\in\mathcal M_1$,
it is immediate to see that $\ex_\alpha^{P,Q} =\ex^{\alpha P, (1-\alpha) Q}$, where 
\begin{align*}
\ex^{R,H}(X):=\inf\left \{x\in \R: \E^{R}[(X-x)_+]\le \E^{H}[(x-X)_+] \right\},~~~R,H\in\mathcal M,~X\in \X.
\end{align*}
One the other hand, for $P,Q\in\mathcal M$, it holds that $\ex^{P,Q} =\ex_\alpha^{\widetilde P, \widetilde Q}$ with $\alpha=P(\Omega)/(P(\Omega)+Q(\Omega))\in (0,1)$. To simplify the presentation, we will also use the term $\ex^{P,Q}$ with $P,Q\in\mathcal M$ to represent duet expectiles.
To prove the necessity statement of Theorem \ref{th-main},
we need to verify two claims:
\begin{itemize}
\item[(a)] With the conditions in Lemma \ref{lm-property}, $U=\ex^{P,Q}$ for some $P,Q\in\mathcal M$ with $0<P\le Q$, where $\ex^{P,Q}(X)$ is the unique number $y\in\R$ such that 
$\E^P[(X-y)_+]=\E^Q[(y-X)_+]$.  This is the main step of the proof.

\item[(b)] The triplet $(\widetilde{P},\widetilde{Q},\alpha)$ with $\alpha=P(\Omega)/(P(\Omega)+Q(\Omega))$ is the unique one such that $\ex^{P,Q}(X)=\ex_{\alpha}^{\widetilde{P},\widetilde{Q}}(X)$ for all $X\in\mathcal  X$.  
This claim is shown in Lemma \ref{prop-uniquedef} below.
\end{itemize} 
The proof of Theorem \ref{th:main-general} also involves corresponding versions of claims (a) and (b).

\begin{lemma}\label{prop-uniquedef}
Assume that 
$\Omega$ is finite or $(\Omega,\mathcal F,\p)$ is an atomless probability space, and $\succsim$ is represented by $\ex_\alpha^{P, Q}$ for some $(\alpha,P,Q)\in  (0,1)\times\mathcal M_1\times \mathcal M_1$. 
Then $(\alpha,P,Q)$ is unique under Axiom \ref{ax:SM}. 
\end{lemma}
 
The proof of Lemma \ref{prop-uniquedef} can be shown by constructing, for each set $T\subseteq \Omega$ of states with non-zero probability under $P$, several bi-atomic random variables $X$ and solving equations
  $\ex^{P,Q}_\alpha(X) =\ex^{P',Q'}_{\alpha'}(X) $, and this yields $P'(T)=P(T)$, $Q'(T)=Q(T)$ and $\alpha'=\alpha$.
A detailed proof is presented in Appendix \ref{app:pf-lem}. 

The following properties of duet expectiles will also be used in the proofs of our main results.
 In what follows, 
for $P,Q\in\mathcal M$,
 we say that $\ex^{P,Q}$
 is strictly monotone if
the preference represented by $\ex^{P,Q}$ satisfies Axiom \ref{ax:SM}.
 For $P\in \M$, let 
$\|\cdot\|_{1}^P$ be the  $L^1$-norm under $P$, i.e.,  
$\Vert X\Vert^P_1=\E^P[|X|]$.
Denote by $P\vee Q$  the maximum of two measures $P$ and $Q$, namely the smallest element $R$ in $\M$ satisfying $R\ge P$ and $R\ge Q$. 
Similarly, denote by $P\wedge Q$   the minimum   of $P$ and $Q$. 

\begin{proposition}\label{prop:technical}
Let $P,Q\in\mathcal M$.
\begin{itemize} 
\item [(i)] If   $ (P\wedge Q)(\Omega) \ne 0$, then $L^1$-continuity of $\ex^{P,Q}$ holds as 
$$|\ex^{P,Q}(X)- \ex^{P,Q}(Y)|\le  \frac{ \|X-Y\|^{P\vee Q}_1}{(P\wedge Q)(\Omega)},~~~\mbox{for all }X,Y\in \X.$$ 
\item  [(ii)] $\ex^{P,Q}$ is strictly monotone if and only if 
$P,Q>0$ in case $\Omega$ is finite, or $P\ac Q\ac \p$ in case $\Omega$ is infinite.   
\item  [(iii)]
Assume that $\ex^{P,Q}$ is strictly monotone. Then, $\ex^{P,Q}$ is concave (resp.~convex) if and only if $P\le Q$ (resp.~$P\ge Q$).   
\end{itemize}
\end{proposition}
Proving the  three statements in Proposition \ref{prop:technical} mainly involves techniques in the theory risk measures, postponed to Appendix \ref{sec:properties}; see Propositions \ref{prop-gextrivial}, 
\ref{prop-finSM}, and \ref{prop-cv}.

\subsection{Overview of the proof of Theorem \ref{th-main}}
\label{sec:pf-main-ov}

In this section, $\Omega$ is finite with at least $4$ states. 
The main direction  to prove is necessity.
Sufficiency is checked in Appendix \ref{subsubsec:Th1}. 

For  the necessity statement, as we have discussed above, it remains to prove Claim (a). Using translation invariance of $U$, it suffices to verify the claim for all acts in the following set
\begin{align*}
\mathcal X^0:=\{X\in\mathcal X: U(X)=0\}.
\end{align*}
To make use of disco additivity,   it is natural to consider the set of all acts  in $\mathcal X^0$ that share the same disappointment event $S\in\mathcal F$. We  denote this set by
\begin{align*}
\mathcal X_{S}:=\{X\in\mathcal X^0: \{X\ge 0\}=S\},~~S\in\mathcal F.
\end{align*}
Clearly, $\mathcal X_{\Omega}=\{0\}$, $\mathcal X_{\varnothing}=\varnothing$, and $\mathcal X_S\neq \varnothing$ for all $S\in\mathcal F\setminus\{\Omega,\varnothing\}$.
For a fixed $S$, we are able to characterize $U$ on $\mathcal X_S$ in the next lemma.
\begin{lemma}\label{lm:EonS}
For any $S\in\mathcal F$, there exist $P_S$ and $Q_S$ in $\mathcal M$
such that $P_S(i)Q_S(j)>0$ for $i\in S,~j\in S^c$ and
$U(X)=\ex^{P_S,Q_S}(X)$ for all $X\in\mathcal X_S$. 
\end{lemma} 

Lemma \ref{lm:EonS} is a crucial step in the proof of Theorem \ref{th-main}, and it requires several sophisticated steps, which we outline below, omitting some details, which are presented in Appendix \ref{subsubsec:Th1}.

The cases $S\in\{\Omega,\varnothing\}$ are trivial. For $S\in\mathcal F\setminus\{\Omega,\varnothing\}$, 
denote by $a_S=U(\id_S)$, and define $\mathcal X_{S}^+=\{X\in\mathcal X: 0<\max X<\min X/a_S\}$ and
$\phi_S:\mathcal X_{S}^+\to \R$ as
\begin{align*}
	\phi_S(X):=U(X\id_S),~~X\in\mathcal X_{S}^+.
\end{align*}

Using Lemma \ref{lm-property} parts (i), (v) and (vi), we can verify that 
the disappointment event of $X\id_S$ is $S^c$ for $X\in \X^+_S$ (here, we used the condition $\max X<\min X/a_S$),
and $\phi_S$ is monotone and additive on  the convex cone $  \X^+_S$.
Then, we   extend $\phi_S$ to $\mathcal X$  by using the translation invariance property, resulting in a monotone and additive function on $\X$. 
 This   yields $\phi_S=\E^{P_S}$ on $\mathcal X_S^+$ for some $P_S\in\mathcal M$ with $P_S(S^c)=0$ by the representation of linear functions on $\X$. By strict monotonicity of $U$ and construction of $\phi_S$,
we   get $P_S(i)>0$ for $i \in S$. 
Analogously, we can use a similar argument to show that $\psi_S=\E^{Q_S}$ on $\mathcal X_S^-$, where
\begin{align*}
\mathcal X_{S}^-=\{X\in\mathcal X: \max X/(1-a_S)<\min X<0\}~~{\rm and}~~\psi_S(X)=U(X\id_{S^c})
\end{align*}
and $Q_S\in\mathcal M$ satisfies $Q_S(i)=0$ for $i \in S$ and $Q_S(i) >0$ otherwise. 

For $X\in\mathcal X_S$, we can find a largh enough $\eta>0$ be such that $X-\eta\in\mathcal X_S^-$ and $X+\eta\in\mathcal X_S^+$. Then, the following acts are disco with $S^c$ as the common disappointment event:
\begin{align*}
X,~-\eta\id_{S^c},~\eta\id_{S},~(X-\eta)\id_{S^c},~(X+\eta)\id_{S},~-\eta\id_{S^c}+\eta\id_{S}.
\end{align*}
Using $U(X)=0$ for $X\in \X_S$ and by decomposition into the above five terms, for which $U$ is additive due to Lemma \ref{lm-property} (v), we can verify
\begin{align}\label{eq-repS}
\E^{P_S}[X_+]=\E^{Q_S}[X_-],~~ \mbox{for all }X\in\mathcal X_S,
\end{align} 
which is an equivalent formulation of  
 $U(X)=\ex^{P,Q}(X)$ for all $X\in \mathcal X_S$, and thus the lemma holds. 

With Lemma \ref{lm:EonS}, we have established the representation of $U$ on $\mathcal X_S$. The next nontrivial task is to demonstrate that $P_S$ and $Q_S$ can be chosen independently of $S$. Note the representation given in \eqref{eq-repS} remains valid no matter what values  are assigned to 
$P_S(i)$ and $Q_S(j)$  for $i \in S^c$ and $j \in S$. 
Additionally, \eqref{eq-repS} holds under a scaling transformation of $P_S$ and $Q_S$. Thus, it is natural to begin with a normalization. Denote by $\mathcal S_1=\{S\in\mathcal F: 1\in S\}$ and $\mathcal S_1^c=\mathcal F\setminus\mathcal S_1$. We assume that
\begin{align}\label{ass-norm}
P_S(1)=1 {~\rm for~} S\in \mathcal S_1~~{\rm and}~~ Q_S(1)=1 {~\rm for}~S\in \mathcal S_1^c.
\end{align}
Further, we define 
\begin{align}\label{eq-PP}
&P(i)=P_{\{i\}}(i),~~Q(i)=Q_{\Omega\setminus\{i\}}(i)~~{\rm for}~i\in[n]\setminus\{1\}; \\\label{eq-QQ}
&\widetilde{P}(i)=P_{\{1,i\}}(i),~~
\widetilde{Q}(i)=Q_{\Omega\setminus\{1,i\}}(i)~~{\rm for}~ i\in[n].
\end{align}
The next lemma justifies consistency between $P_S$ and $Q_S$ across different choices of $S$.

\begin{lemma}\label{lm-inP}
The following three statements hold:
\begin{itemize}
\item[(i)] $P_S(i)=P(i)$ for all $S\in \mathcal S_1^c$ and $i\in S$, and $Q_S(j)=Q(j)$ for all $S\in \mathcal S_1$ and $j\in S^c$.

\item[(ii)] $P_S(i)=\widetilde{P}(i)$ for all $S\in \mathcal S_1$ and $i\in S$, and $Q_S(j)=\widetilde{Q}(j)$ for all $S\in \mathcal S_1^c$ and $j\in S^c$.

\item[(iii)] $P(i)/\widetilde{P}(i)=\widetilde{Q}(j)/Q(j)$ for all $i,j\ge 2$ and $i\neq j$.
\end{itemize}
\end{lemma}

To prove Lemma \ref{lm-inP}, the continuity   in Lemma \ref{lm-property} (iii) plays a crucial role. The main idea is explained below.
We construct two sequences of acts, one in $ \mathcal{X}_S$ and another one in $ \mathcal{X}_{S'}$,  both converging to the same limit, where $S$ and $S'$ are different sets. By \eqref{eq-repS}, we can derive some conditions on $P_S$ and $Q_S$ from the first sequence, and on $P_{S'}$ and $Q_{S'}$ from the second. Given that both sequences converge to the same limit, Lemma \ref{lm-property} (iii) is useful to establish a connection between $P_S$ and $P_{S'}$, as well as between $Q_S$ and $Q_{S'}$. 
Appendix \ref{subsubsec:Th1} contains the complete proof of Lemma \ref{lm-inP}.

With Lemma \ref{lm-inP}   established,
we can get  from part (iii) that there exists $\lambda>0$ such that $\lambda=P(i)/\widetilde{P}(i)=\widetilde{Q}(i)/Q(i)$ for all $i\ge 2$. 
 This is the place where we use the assumption that 
 the size of $\Omega$ is no less than $4$.  
The next result shows that the measures in the representation \eqref{eq-repS} can be chosen  independently of $S$.

\begin{lemma}\label{lm-repX0}
Let $P(1)=\lambda$. Recall that $P(i)$ for $i\in[n]\setminus \{1\}$ and $\widetilde{Q}(j)$ for $j\in[n]$ are defined in \eqref{eq-PP} and \eqref{eq-QQ}, respectively. We have $U(X)=\ex^{P,\widetilde{Q}}(X)$ for all $X\in\mathcal X^0$, or, equivalently,
\begin{align*}
\E^{P}[X_+]=\E^{\widetilde{Q}}[X_-]~~\mbox{ for all } X\in\mathcal X^0.
\end{align*}
\end{lemma}

\begin{proof}[Proof of Lemma \ref{lm-repX0}]
For $S\in\mathcal S_1$ and $X\in\mathcal X_S$,
we have $\E^{\widetilde P}[X_+]-\E^{Q}[X_-]=0$ by applying Lemma \ref{lm-inP} (i) and (ii). 
By \eqref{ass-norm} and  \eqref{eq-QQ}, we have $\widetilde{P}(1)=1$, and hence, $P(1)=\lambda \widetilde{P}(1)=\lambda$.
Moreover, it holds that 
$X(1)\ge 0$ and $P(i)/\widetilde{P}(i)=\widetilde{Q}(i)/Q(i)=\lambda$ for $i\ge 2$. Therefore,
\begin{align*}
\E^{P}[X_+]-\E^{\widetilde Q}[X_-]=\lambda(\E^{\widetilde P}[X_+]-\E^{Q}[X_-])=0.
\end{align*}
This yields $\E^{P}[X_+]=\E^{\widetilde Q}[X_-]$ for $X\in\mathcal X_S$ with $S\in\mathcal S_1$. For $S\in\mathcal S_1^c$ and $X\in\mathcal X_S$, it follows immediately from Lemma \ref{lm-inP} (i) and (ii) that $\E^{P}[X_+]=\E^{\widetilde Q}[X_-]$. This completes the proof.
\end{proof}

Combining the above lemma and translation invariance in Lemma \ref{lm-property} (ii),
we have concluded that $U$ has the form $\ex^{P,Q}$ for $P,Q\in\mathcal M$ on $\mathcal X$, that is, the main part of the necessity statement is proved.
 It remains to show that $0<P\le Q$. The condition $P,Q>0$ follows from  Proposition \ref{prop:technical} (ii)
and the condition $P\le Q$ follows from Proposition \ref{prop:technical}   (iii) and the fact that $U$ is concave in 
Lemma \ref{lm-property}.

\subsection{Overview of the proof of Theorem \ref{th:main-general}}\label{sec:ovth3}
Next, we provide a  sketch of the proof of the necessity statement in Theorem \ref{th:main-general}.
All detail analyses are postponed to Appendix \ref{app:Th2}. 

We introduce some notation below.  Denote by $\mathcal U$ the set of all uniform random variables on $(0,1)$. 
We iteratively subdivide the interval $(0,1)$ into finer finite partitions, such that 
the Lebesgue measure of each partition element asymptotically approaches zero. Such a sequence of partitions are denoted by 
$\{\Psi_k\}_{k\in\N}$. 
For each $V\in\mathcal U$, we can construct a 
filtration $\left\{  \mathcal F
_{k}^{V}\right\}_{k\in\mathbb{N}}$ in $\mathcal F$, where $\mathcal F_{k}^{V}=\sigma(V^{-1}\left( \Psi_{k})\right)  $ for all $k\in\mathbb{N}$. As usual, $\mathcal F_\infty^V=\sigma(\bigcup_{k\in\N}\mathcal F_k^V)$.

The desired statement can be shown in four steps. 

Step 1 is the most challenging, and we aim to show that for each $V\in\mathcal U$ there exist $P^V, Q^V\in\mathcal M$ such that $U(X)=\ex^{P^V,Q^V}(X)$ for all $X\in \bigcup_{k\in\N}L^\infty(\Omega,\mathcal F_k^V,\p)$. Theorem \ref{th-main} and the martingale convergence theorem will be applied in this step. Specifically, we need to ensure that the sequence is uniformly integrable when applying the martingale convergence theorem, and Axiom \ref{ax:MC} is used to establish uniform integrability.

In Step 2, we will use the assumption of a standard probability space. Since $(\Omega, \mathcal F, \p)$ is a standard probability space, there exists $V_0 \in \mathcal U$ such that $\sigma(V_0) = \mathcal F$, and consequently, $\mathcal F_{\infty}^{V_0} = \mathcal F$ (see Appendix \ref{app:Th2} for details on this equivalence). Note that $P^V$ and $Q^V$ are defined on $\mathcal F_{\infty}^V$. In this step, we aim to verify that $P^V=P^{V_0}$ and $Q^V=Q^{V_0}$ on $\mathcal F_{\infty}^V$ for all $V\in\mathcal U$, which ensures that $U(X)=\ex^{P^{V_0},Q^{V_0}}(X)$ for all $X\in \bigcup_{V\in\mathcal U,~k\in\N} L^\infty(\Omega,\mathcal F_k^V,\p)$.

Step 3 is to demonstrate that the representation in Step 2 can be extended to $\X$, where a continuity property of duet expectile in Proposition \ref{prop:technical} (i) will be used.

Up to now, we know that the preference ralation can be represented by a duet expectile $\ex_\alpha^{P,Q}$.
The last step is to verify that the triplets $(\alpha,P,Q)\in (0,1)\times \mathcal M_1\times\mathcal M_1$ satisfies all conditions in the theorem, similarly to the case of Theorem \ref{th-main}, again using
Lemmas \ref{lm-property} and \ref{prop-uniquedef} and Proposition \ref{prop:technical}.

\subsection{Brief proof sketch   of Theorem  \ref{th:p-ex}}
Based on Theorem \ref{th:main-general}, it suffices to verify the following result.

\begin{proposition}\label{prop-id}
Let $\mathcal{X}=L^\infty(\Omega, \mathcal{F}, \mathbb{P})$ with $\mathbb{P}$ atomless, and let $P,Q\in\mathcal M$ be such that $P\ac Q \ac \mathbb{P}$. Then a preference relation $\succsim$ that can be represented by $\ex^{P,Q}$ satisfies Axiom \ref{ax:EI} if and only if $P=\lambda Q$, for some $\lambda >0$. 
\end{proposition}
Necessity of Proposition \ref{prop-id} is the main direction, which will be shown by contradiction. 
Assuming that $P \neq \lambda Q$ for all $\lambda > 0$, we construct some disjoint sets $A$ and $B$ 
according to the value of the Radon-Nikodym derivative $\d \widetilde Q/\d \widetilde P$ such that $\d \widetilde Q/\d \widetilde P>1$ on $A$ and  $\d \widetilde Q/\d \widetilde P<1$ on $B$.
Let $C=(A\cup B)^c$. 
They lead to a possible contradiction  $\id_A \succsim \id_B$ and $\id_{B \cup C} \succ \id_{A \cup C}$. See Section \ref{app:th3} for detailed computations.

\subsection{Brief proof sketch  of Theorem  \ref{th:riskaverse1}}

For $X\in\mathcal X$,
denote by $x=\ex_\alpha^{P,Q}(X)$ and $\beta=1-\alpha$, and we have $\alpha\E^P[(X-x)_+]=\beta\E^Q[(x-X)_+]$.
The following relation plays a key role in the proof of Theorem  \ref{th:riskaverse1}:
\begin{align*}
&\alpha(\E^P[X] - x) =\alpha(\E^P[(X-x)_+] -\E^P[(x-X)_+]) =
\E^{\beta Q-\alpha P}[(x-X)_+];\\
&\beta(\E^Q[X] - x) =\beta\left(\E^Q[(X-x)_+] -\E^Q[(x-X)_+\right) 
=\E^{\beta Q-\alpha P} [(X-x)_+].
\end{align*}
The equivalence between (i), (ii) and (iv) follows immediately from the above relations.

To establish (iii) $\Leftrightarrow$ (iv), we first demonstrate that (iii) holds if and only if $\ex_\alpha^{P,Q}$ is concave. Then, using Proposition \ref{prop:technical}, we conclude that (iii) $\Leftrightarrow$ (iv). The detailed analysis is deferred to Appendix \ref{app:sec5}.

\section{Conclusion}
\label{sec:concl}

The paper offers many new notions of relevance to decision theory, and more broadly, to the theory of statistical functionals.  
We introduce the concepts of disappointment concordance and disco aversion. They lead to a new  class of functionals, the duet expectiles, that can be seen as generalizations of the solo expectile, popular in the literature of asymmetric regression and in financial risk management. The first two characterization results, Theorems \ref{th-main} and \ref{th:main-general}, state that, under strict monotonicity and continuity, the axiom  of disco aversion  characterizes preference relations represented by duet expectiles. A duet expectile involves two endogenous probability measures, and it becomes a usual expectile under an additional axiom of event-independence (Theorem \ref{th:p-ex}). 
The characterization of duet expectiles leads to  new notions of probabilistic sophistication and risk aversion for multiple probability measures (Theorem \ref{th:riskaverse1}).
Putting into the AA framework, we characterize duet expectiled utilities (Theorems \ref{thm:diesel} and \ref{thm:AASPS}), which include the decision model of \cite{G91} as a special case. 
Duet expectiles have many convenient technical properties, and in particular they are coherent risk measures in the sense of \cite{ADEH99}.
Demonstrated by the many results in this paper, the study of duet expectiles gives rise to new mathematical techniques
as well as conceptual novelties, 
many of which   require  even further exploration. 

\begin{appendix}

\section{Full proofs of Theorems \ref{th-main} and \ref{th:main-general}}\label{app:proofs}
In this appendix, we first provide the proofs of Lemmas \ref{lem:S}, \ref{lm-property}, and \ref{prop-uniquedef}, which are useful in establishing our main results, Theorems \ref{th-main} and \ref{th:main-general}. Following that, we present the complete proofs of these two theorems.



 
 \subsection{Lemmas \ref{lem:S}, \ref{lm-property}  and \ref{prop-uniquedef}}

\label{app:pf-lem}

\begin{proof}[Proof of Lemma \ref{lem:S}] 
(i) Suppose that $\Omega$ is finite.
Let $m,M \in \R$ be such that 
$m<\inf _{\omega\in \Omega}X (\omega)$ 
and $M> \sup_{\omega\in \Omega} X(\omega)$.
Clearly, 
 $m < X< M$, and hence $M\succ  X \succ  m$ by Axiom \ref{ax:SM}. By Axiom \ref{ax:C} and the fact that $ \succsim$ is total, there exists $c_X\in [m,M]$ such that $c_X \sim X$. Axiom \ref{ax:SM} then gives the uniqueness of $c_X$.
 For an atomless probability space, we can share a similar proof to the case that $\Omega$ is finite by substituting $\inf_{\omega\in\Omega} X(\omega)$ and $\sup_{\omega\in\Omega} X(\omega)$ for the essential infimum and essential supremum of $X$, respectively.

(ii) The if part follows directly from Axiom \ref{ax:SM}. The only if part follows from the uniqueness of $c_X$ and $c_Y$ in (i). \end{proof}


\begin{proof}[Proof of Lemma \ref{lm-property}]  

(i) This result follows immediately from Axiom \ref{ax:SM} and the fact that $U(m)=m$ for all $m\in\R$.

(ii) The case that $X$ is a degenerate act is trivial. Let us consider a nondegenerate act.
By Axiom \ref{ax:DA}, we have $U(X+m)\ge U(U(X)+m)=U(X)+m$ for all $X\in\X$ and $m\in\R$ as $X \sim U(X)$, and $U(X)\id_{\Omega}$ and $m\id_{\Omega}$ are disco.
On the other hand, let $X\in\X$ be a nondegenerate act. Denote by $A=\{X<U(X)\}$, and define $Y_{\epsilon}=m\id_{A}+(m+\epsilon)\id_{A^c}$ with $\epsilon> 0$. It follows from (i) that $\{Y_\epsilon<U(Y_\epsilon)\}=A$. Hence, using Axiom \ref{ax:DA}, we have 
\begin{align*}
	U(X+Y_\epsilon)\le U(U(X)+Y_\epsilon),~~\forall \epsilon>0.
\end{align*}
Let $\epsilon\to 0$, and the above inequality yields $U(X+m)\le U(X)+m$
where we have used Axiom \ref{ax:C}. Hence, we have concluded that $U(X+m)=U(X)+m$ for all $X\in\X$ and $m\in\R$.

(iii) By Axiom \ref{ax:SM} and noting that $U$ satisfies translation invariance in (ii), the result follows immediately from Lemma 4.3 of  \cite{FS16}.

(iv) The case that $Y$ is degenerate has been verified in (ii). For any nondegenerate $Y\in\mathcal X$, denote by $A=\{Y<U(Y)\}$, and define $X_{\epsilon}=\epsilon \id_{A^c}$, $\epsilon>0$.
It follows from (i) that  
$\{X_{\epsilon}<U(X_{\epsilon})\}=A$.
Let $X=X_\epsilon-U(X_\epsilon)$ and $X'=Z-U(Z)$ for an arbitrary $X\in\mathcal X$. It follows from the translation invariance in (ii) that $U(X)=U(X')=0$. Also note that
$\{X<U(X)\}=\{X_\epsilon-U(X_\epsilon)< 0\}=A$, which implies that $X,Y$ are disco.
Using Axiom \ref{ax:DA}, we have
\begin{align*}
U(X_{\epsilon}-U(X_{\epsilon})+Y)=U(X+Y)\le U(X'+Y)
= U(Z-U(Z)+Y),~~\forall \epsilon>0.
\end{align*}
This, combining with the translation invariance in (ii), is equivalent to 
\begin{align*}
	U(X_{\epsilon}+Y)+U(Z)-U(X_{\epsilon})\le U(Z+Y),~~\forall \epsilon>0.
\end{align*}
It is easy to see that $\|X_\epsilon\|\to 0$. Hence, we have $U(X_\epsilon)\to 0$ and $U(X_\epsilon+Y)\to U(Y)$ as $\epsilon\to0$ by applying (iii). Also noting that $Z$ and $Y$ are arbitrarily chosen, we complete the proof of (iv).

(v) Using Axiom \ref{ax:DA} and the translation invariance in (ii), we have $U(X+Y)\le U(U(X')+Y)=U(X)+U(Y)$ for $X \sim X'$ and $\{X<U(X)\}=\{Y<U(Y)\}$. Combining the superadditivity in (iv), we have $U(X+Y)=U(X)+U(Y)$ for all $X,Y\in\X$ satisfying $\{X<U(X)\}=\{Y<U(Y)\}$.

(vi) Note that $\{X<U(X)\}=\{ X< U(X)\}$. It follows from (v) that $U(2X)=2U(X)$. Thus, we have $\{X<U(X)\}=\{2X<U(2X)\}$, and using (v) again, it holds that $U(3X)=3U(X)$. Applying the above arguments iteratively, we have $U(\lambda X)=\lambda U(X)$ for all rational $\lambda>0$. Let now $\lambda>0$ be any real number. There exists a sequence of rational $\{\lambda_n\}$ such that $\lambda_n>0$ and $\lambda_n\to \lambda$. Since $\|\lambda_n X-\lambda X\|\le |\lambda_n-\lambda|\cdot\|X\|\to 0$ as $n\to \infty$, we have $U(\lambda_n X)\to U(\lambda X)$ by applying (iii). On the other hand, $U(\lambda_n X)=\lambda_n U(X)\to \lambda U(X)$, and hence, we have $U(\lambda X)=\lambda U(X)$. This completes the proof.
\end{proof}

\begin{proof}[Proof of Lemma \ref{prop-uniquedef}]
Suppose that $\succsim$ can be represented by $\ex_{\alpha}^{P,Q}$ and Axiom \ref{ax:SM} holds.
Assume that there is $(\alpha',P',Q')\in\mathcal H$ such that $\ex_{\alpha'}^{P',Q'}=\ex_{\alpha}^{P,Q}$. We aim to verify that $(\alpha',P',Q')=(\alpha,P,Q)$. 

  We can verify that the following conditions are equivalent: (a) $P(S)=0$; (b) $Q(S)=0$; (c) $P'(S)=0$; (d) $Q'(S)=0$. This is verified in Proposition \ref{prop:technical} (ii).

To see $P'=P$ and $Q'=Q$,
we only need to consider non-zero measure set, whose component is also a non-zero measure set. For such a set $S_1\in\mathcal F$, it holds that $P(S_1),Q(S_1),P'(S_1),Q'(S_1)\in(0,1)$. It suffices to verify $\alpha=\alpha'$, $P'(S_1)=P(S_1)$ and $Q'(S_1)=Q(S_1)$. Without loss of generality, we can assume that there exist $S_2,S_3$ such that $\bigcup_{i\in[3]}S_i=\Omega$ and the probability of $P,Q,P',Q'$ on these two sets are positive. In the case of atomless probability space, it is clear that such $S_2,S_3$ always exist. In the finite case, if they do not exist, we can replace $S_1$ by $S_1^c$ in the analysis. 
Denote by $p_i=P(S_i)$, $q_i=Q(S_i)$, $p'_i=P'(S_i)$ and $q'_i=Q'(S_i)$ for $i=1,2,3$ and define 
\begin{align*}
X_{i,j}=\id_{S_i}-\frac{\alpha p_i}{(1-\alpha)q_j}\id_{S_j},~(i,j)\in \mathcal D,
\end{align*}
where $\mathcal D=\{(i,j): i,j\in[3],~i\neq j\}$.
It is straightforward to check that $\ex_{\alpha}^{P,Q}(X_{i,j})=0$ for all $(i,j)\in\mathcal D$, and hence, $\ex_{\alpha'}^{P',Q'}(X_{i,j})=0$, which implies
\begin{align}\label{eq-unique1}
	\frac{\alpha p_i}{\alpha' p_i'}=\frac{(1-\alpha)q_j}{(1-\alpha')q_j'},~~(i,j)\in \mathcal D.
\end{align}
Therefore,
\begin{align*}
\frac{\alpha}{\alpha'}=\frac{
\textstyle \sum_{i\in I}\alpha p_i}{
\textstyle \sum_{i\in I}\alpha' p'_i}=\frac{
\textstyle \sum_{j\in J}(1-\alpha) q_j}{
\textstyle\sum_{j\in J}(1-\alpha') q'_j}=\frac{1-\alpha}{1-\alpha'}.
\end{align*}
This yields $\alpha'=\alpha$. It follows from \eqref{eq-unique1} that $i\mapsto  p_i/ p_i'$ and $j\mapsto q_j/q_j'$ are constant for $i,j\in[3]$, and they take the same constant value. Since $\bigcup_{i\in[3]}S_i=\Omega$, we have $\sum_{i\in [3]}p_i=\sum_{i\in [3]}p_i'=\sum_{j\in [3]}q_j=\sum_{j\in [3]}q_j'=1$, and hence, $p_i=p_i'$ and $q_i=q_i'$ for $i=1,2,3$. This implies that $P'(S_1)=P(S_1)$ and $Q'(S_1)=Q(S_1)$. Hence, 
we complete the proof.
\end{proof}

\subsection{Theorem \ref{th-main}}\label{subsubsec:Th1}

\noindent
\emph{Proof of Theorem \ref{th-main}.}
{\bf Sufficiency.} 
Suppose that $\succsim$ can be represented by $\ex^{P,Q}$ for some $P,Q\in\mathcal M$ with $0<P\le Q$.
Axiom \ref{ax:SM}  follows from Proposition \ref{prop:technical} (ii).
Axiom \ref{ax:C} follows by standard algebraic manipulations.
Let us now consider Axiom \ref{ax:DA}. For $X_1 \sim X_2 \sim X_3 \sim 0$, we have $\ex^{P,Q}(X_i)=0$ for $i=1,2,3$. Suppose that $X_1$ and $X_2$ are disco. It holds that $\{X_1<0\}=\{X_2<0\}=:S$ and $\E^{P}[X_i\id_{S}]=\E^{Q}[X_i\id_{S^c}]$ for $i=1,2$.
Noting that $\{X_1+X_2<0\}=S$, we have 
\begin{align*}
&\E^{P}[(X_1+X_2)_+]-\E^{Q}[(X_1+X_2)_-]\\
&=\E^{P}[(X_1+X_2)\id_{S}]-\E^{Q}[(X_1+X_2)\id_{S^c}]\\
&=\E^{P}[X_1\id_{S}]-\E^{Q}[X_1\id_{S^c}]+(\E^{P}[X_2\id_{S}]-\E^{Q}[X_2\id_{S^c}])=0.
\end{align*}
Hence, we have $\ex^{P,Q}(X_1+X_2)=0$.
By Proposition \ref{prop:technical} (iii) and noting that duet expectile is positively homogeneous (straightforward to verify), $0<P\le Q$ implies that $\ex^{P,Q}$ is superadditive. Therefore, we conclude that
\begin{align*}
	0=\ex^{P,Q}(X_1+X_2)=\ex^{P,Q}(X_1)+\ex^{P,Q}(X_3)\le \ex^{P,Q}(X_1+X_3). 
\end{align*}
This implies $X_1+X_3\succsim X_1+X_2$. Note that duet expectile is translation invariant (see Proposition \ref{prop-gextrivial}).
Hence, Axiom \ref{ax:DA} holds, and we complete the proof of sufficiency.

{\bf Necessity}. Suppose that Axioms \ref{ax:SM}, \ref{ax:C} and \ref{ax:DA} hold. As outlined in Section \ref{sec:proofth1},  
to prove 
Theorem \ref{th-main},  it suffices to prove  Lemmas \ref{lm:EonS} and   \ref{lm-inP}. Below we first recall some notations that will be used. Define
\begin{align*}
\mathcal X^0=\{X\in\mathcal X: U(X)=0\}~~{\rm and}~~\mathcal X_{S}:=\{X\in\mathcal X^0: \{X\ge 0\}=S\},~~S\in\mathcal F.
\end{align*}
Clearly, $\mathcal X^0=\bigcup_{S\in\mathcal F}\mathcal X_S$, $\mathcal X_{\Omega}=\{0\}$, $\mathcal X_{\varnothing}=\varnothing$, and $\mathcal X_S\neq \varnothing$ for all $S\in\mathcal F\setminus\{\Omega,\varnothing\}$.

\emph{Proof of Lemma \ref{lm:EonS}:}
The cases $S\in\{\Omega,\varnothing\}$ are trivial. For $S\in\mathcal F\setminus\{\Omega,\varnothing\}$, 
denote by $a_S=U(\id_S)$, and define $\mathcal X_{S}^+=\{X\in\mathcal X: 0<\max X<\min X/a_S\}$ and
$\phi_S:\mathcal X_{S}^+\to \R$ as
\begin{align*}
	\phi_S(X):=U(X\id_S),~~X\in\mathcal X_{S}^+.
\end{align*}
It follows from Lemma \ref{lm-property} (i) that $a_S\in(0,1)$ and $\mathcal X_S^+\neq \varnothing$. Clearly, $\mathcal X_{S}^+$ is closed under addition and $\phi_S(m)=m a_S$ for all $m>0$ by Lemma \ref{lm-property} (vi).
For any $X\in\mathcal X_{S}^+$, we have
\begin{align*}
	0<U(X\id_{S})\le U((\max X) \id_S)=\phi_S(\max X)=(\max X) a_S<\min X,
\end{align*}
This implies that the  disappointment event of $X\id_{S}$ is $S^c$ for all $X\in\mathcal X_S^+$.
Therefore,
\begin{align*}
\phi_S(X+Y)=	U(X\id_{S}+Y\id_{S})=U(X\id_{S})+U(Y\id_{S})=\phi_S(X)+\phi_S(Y),~~\forall X,Y\in\X_S^+,
\end{align*}
where we have used disco additivity (see Lemma \ref{lm-property} (v)) in the second step.
Hence, $\phi_S$ is additive on $\mathcal X_S^+$. Then, we extend $\phi_S$ to $\mathcal X$ by defining
\begin{align*}
	\widehat\phi_S(X)=\phi_S(X+m_X)-m_X a_S~~{\rm with}~m_X=\inf\{m:X+m\in\mathcal X_{S}^+\},~~X\in\mathcal X.
\end{align*}
We assert that 
\begin{align}\label{eq-+phihat}
	\widehat\phi_S(X)=\phi_S(X+m)-m a_S~~{\rm for~any}~m\ge0~{\rm such~that}~X+m\in\mathcal X_{S}^+.
\end{align}
To see this, for $m\ge 0$ such that $X+m\in\mathcal X_{S}^+$,
\begin{align*}
	\phi_S(X+m)-m a_S&=	\phi_S(X+m_X+m-m_X)-m a_S\\
	&=	\phi_S(X+m_X)+\phi_S(m-m_X)-m a_S\\
	&=	\phi_S(X+m_X)+(m-m_X)a_S-m a_S\\
	&=	\phi_S(X+m_X)-m_X a_S
	=\widehat\phi_S(X),
\end{align*}
where the second equality follows from the additivity of $\phi_S$. For any $X,Y\in \mathcal X$, it holds that $X+Y+m_X+m_Y\in\mathcal X_{S}^+$ as $\mathcal X_{S}^+$ is closed under addition. Hence, 
\begin{align*}
	\widehat\phi_S(X+Y)
	&=\phi_S(X+Y+m_X+m_Y)-(m_X+m_Y)a_S\\
	&=\phi_S(X+m_X)+\phi_S(Y+m_Y)-(m_X+m_Y)a_S
	=\widehat\phi_S(X)+\widehat\phi_S(Y),
\end{align*}
where the second equality follows from the additivity of $\phi_S$. This implies that $\widehat\phi_S: \mathcal X\to\R$ is additive. Note that $\phi_S:\X_S^+\to\R$ is monotone as $U$ is monotone. By the representation of $\widehat\phi_S$ in \eqref{eq-+phihat} and noting that $m$ can be large enough, it follows that $\widehat\phi_S$ is monotone. 
Hence, there exists $P_S\in\mathcal M$ such that $\widehat\phi_S(X)=\E^{P_S}[X]$ for $X\in\X$, and we have
\begin{align*}
	\phi_S(X)=U(X\id_S)=\E^{P_S}[X],~~\forall X\in\X_S^+.
\end{align*}
We can require $P_S(S^c)=0$ in the above representation since $X=0$ on $S^c$ for $X\in \X^+_S$. 
Then we aim to show that $P_S(i)>0$ for $i\in S$. 
Define $X=(1+a_S)\id_S+2a_S\id_{S^c}$.
Since $a_S\in(0,1)$, we have $X\in\X_S^+$. Let $\epsilon>0$ (small enough) be such that $1+a_S+\epsilon<(2a_S-\epsilon)/a_S$, and define $Y_i=X+\epsilon \id_{\{i\}}$ for $i\in S$.
It is straightforward to check that $Y_i\in\X_S^+$ for $i\in S$.
Hence,
\begin{align*}
	\phi_S(Y_i)=\E^{P_S}[Y_i]=\E^{P_S}[X]+\epsilon P_S(i)=\phi_S(X)+\epsilon P_S(i),~~i\in S.
\end{align*}
Noting that $Y_i\id_S-X\id_S=\epsilon\id_{\{i\}}$, it follows from the strict monotonicity that
\begin{align*}
\phi_S(Y_i)-\phi_S(X)=U(Y_i\id_S)-U(X\id_S)>0.
\end{align*}
Combining with the above two equations, we have $P_S(i)>0$ for $i\in S$.

Next, we define  $\mathcal X_{S}^-=\{X\in\mathcal X: \max X/(1-a_S)<\min X<0\}$ and $\psi_S:\mathcal X_{S}^-\to \R$ as 
\begin{align*}
	\psi_S(X):=U(X\id_{S^c}),~~X\in\mathcal X_{S}^-.
\end{align*}
We aim to show that $\psi_S(X)=\E^{Q_S}[X]$ with some $Q_S\in\mathcal M$ for $X\in\mathcal X_S^-$, where 
$Q_S$ satisfies $Q_S(i)>0$ for $i\in S^c$.
The analysis is similar to the previous case of $\psi_S: \mathcal X_S^+\to \R$, and we provide a brief proof here. It follows from translation invariance of $U$ that $U(\id_S-1)=U(\id_S)-1=a_S-1$. Hence, we have 
\begin{align*}
\psi_S(m)&=U((-m)(\id_S-1))=(-m)U(\id_S-1)=(1-a_S)m,~~m<0,
\end{align*}
where the third equality follows from positive homogeneity of $U$. 
For any $X\in\mathcal X_{S}^-$, we have
\begin{align*}
	0&> U(X\id_{S^c})\ge U((\min X)(1-\id_S))=\psi(\min X)=(1-a_S)\min X>\max X.
\end{align*}
This implies that disappointment event of $X\id_{S^c}$ is $S^c$ for all $X\in\mathcal X_S^-$.
Next, we extend $\psi_S$ to $\mathcal X$ by defining
\begin{align*}
	\widehat\psi_S(X)=\phi_S(X-m_X)+m_X a_S~~{\rm with}~m_X=\inf\{m:X-m\in\mathcal X_{S}^-\},~~X\in\mathcal X.
\end{align*}
By the similar arguments in the case of $\phi_S:\X_S^+\to\R$, one can check that $\widehat\psi_S$ is monotone and additive. Hence, there exists $Q_S\in\mathcal M$ such that $\widehat\psi_S=\E^{Q_S}$, and this implies
\begin{align*}
	\psi_S(X)=U(X\id_{S^c})=\E^{Q_S}[X],~~\forall X\in\X_S^-.
\end{align*}
Similarly, we can verify that $Q_S(i)>0$ for $i\in S^c$.

Finally, for $X\in\mathcal X_S$, it is clear that $X-\eta\in\mathcal X_S^-$ and $X+\eta\in\mathcal X_S^+$ for some large enough $\eta>0$. It is straightforward to verify that the following acts are disco with $S^c$ as the same   disappointment event:
\begin{align*}
X,~-\eta\id_{S^c},~\eta\id_{S},~(X-\eta)\id_{S^c},~(X+\eta)\id_{S},~-\eta\id_{S^c}+\eta\id_{S}.
\end{align*}
We have the following equality chains:
\begin{align*}
-\E^{Q_S}[\eta]+\E^{P_S}[\eta]
&=U(X)+\psi_S(-\eta)+\phi_S(\eta)\\
&=U(X)+U(-\eta\id_S)+U(\eta\id_S)
=U((X-\eta)\id_{S^c}+(X+\eta)\id_S)\\
&=U((X-\eta)\id_{S^c})+U((X+\eta)\id_S)
=\psi_S(X-\eta)+\phi_S(X+\eta)\\
&=\E^{Q_S}[X]+\E^{P_S}[X]-\E^{Q_S}[\eta]+\E^{P_S}[\eta]\\
&=-\E^{Q_S}[X_-]+\E^{P_S}[X_+]-\E^{Q_S}[\eta]+\E^{P_S}[\eta],
\end{align*}
where the first equality follows from $U(X)=0$, and we have used disco additivity of $U$ in the third and fourth equalities, and the last equality holds because $P_S(i)Q_S(j)>0$ and $P_S(j)=Q_S(i)=0$ for $i\in S$ and $j\in S^c$. 
This completes the proof.
\qed

\emph{Proof of Lemma \ref{lm-inP}:}
(i) 
By Lemma \ref{lm:EonS}, we have 
\begin{align}\label{eq-Rstep1}
	\E^{P_S}[X_+]=\E^{Q_S}[X_-],~~\forall X\in\X_S,
\end{align} 
where $P_S(i)Q_S(j)>0$ for $i\in S$ and $j\in S^c$. 
For $i,j\in[n]$ with $i\neq j$, let $\eta_{i,j}>0$ be such that $\id_{\{i\}}-\eta_{i,j}\id_{\{j\}}\in\mathcal X^0$, and indeed, using \eqref{eq-Rstep1} yields $\eta_{i,j}=P_{\Omega\setminus\{j\}}(i)/Q_{\Omega\setminus\{j\}}(j)$. 
For $S\in\mathcal F$,  $i\in S$, $j\in S^c$ and $\epsilon>0$, define
\begin{align*} 
X_{S,i,j}^\epsilon=\id_{\{i\}}-f_{S,i,j}(\epsilon)\id_{\{j\}}-\epsilon\id_{S^c},
\end{align*}
where $f_{S,i,j}:\R_+\to\R$ is a function such that $X_{S,i,j}^\epsilon\in\mathcal X^0$. 
Indeed, for small enough $\epsilon>0$, we have $f_{S,i,j}(\epsilon)=(P_S(i)-\epsilon Q_S(S^c))/Q_S(j)>0$, which implies $X_{S,i,j}^\epsilon\in\mathcal X_S$.
Since $U$ is strictly monotone, it is clear that 
$f_{S,i,j}$ is strictly decreasing. 
Lemma \ref{lm-property} (iii) implies that $ f_{S,i,j}(\epsilon)\to\eta_{i,j}$ and $X_{S,i,j}^\epsilon\to \id_{\{i\}}-\eta_{i,j}\id_{\{j\}}$ as $\epsilon\downarrow 0$. 
For $S\in \mathcal S_1^c$ and $i\in S$,
note that we have assumed $Q_S(1)=Q_{\{i\}}(1)=1$ (see \eqref{ass-norm}), and hence,
\begin{align*}
0&=\E^{P_S}[(X_{S,i,1}^\epsilon)_+]-\E^{Q_S}[(X_{S,i,1}^\epsilon)_-]
=P_S(i)-f_{S,i,1}(\epsilon)Q_S(1)-\epsilon Q_S(S^c)\\
&=P_S(i)-f_{S,i,1}(\epsilon)-\epsilon Q_S(S^c)\to P_S(i)-\eta_{i,1}~~{\rm as}~\epsilon\to0
\end{align*}
and 
\begin{align*}
0&=\E^{P_{\{i\}}}[(X_{\{i\},i,1}^\epsilon)_+]-\E^{Q_{\{i\}}}[(X_{\{i\},i,1}^\epsilon)_-]\\
&=P_{\{i\}}(i)-f_{{\{i\}},i,1}(\epsilon)Q_{\{i\}}(1)-\epsilon Q_{\{i\}}(\Omega\setminus\{i\})\\
&=P_{\{i\}}(i)-f_{{\{i\}},i,1}(\epsilon)-\epsilon Q_{\{i\}}(\Omega\setminus\{i\})\to P_{\{i\}}(i)-\eta_{i,1}
~~{\rm as}~\epsilon\to0.
\end{align*}
This gives $P_S(i)=P_{\{i\}}(i)=P(i)$ for $S\in\mathcal S_1^c$ and $i\in S$.

For $S\in\mathcal S_1$ and $i\in S^c$, we have $X_{S,1,i}^\epsilon\in\mathcal X_S$ and $X_{\Omega\setminus\{i\},1,i}^\epsilon\in \X_{\Omega\setminus\{i\}}$. Note that we have assumed $P_S(1)=P_{\Omega\setminus\{i\}}(1)=1$ (see \eqref{ass-norm}). It holds that
\begin{align*}
0&=\E^{P_S}[(X_{S,1,i}^\epsilon)_+]-\E^{Q_S}[(X_{S,1,i}^\epsilon)_-]=P_S(1)-f_{S,1,i}(\epsilon)Q_S(i)-\epsilon Q_S(S^c)\\
&=1-f_{S,1,i}(\epsilon)Q_S(i)-\epsilon Q_S(S^c)\to 1-\eta_{1,i}Q_S(i)~~{\rm as}~\epsilon\to0
\end{align*}
and
\begin{align*}
0&=\E^{P_{\Omega\setminus\{i\}}}[(X_{\Omega\setminus\{i\},1,i}^\epsilon)_+]-\E^{Q_{\Omega\setminus\{i\}}}[(X_{\Omega\setminus\{i\},1,i}^\epsilon)_-]\\
&=P_{\Omega\setminus\{i\}}(1)-f_{{\Omega\setminus\{i\}},1,i}(\epsilon)Q_{\Omega\setminus\{i\}}(i)-\epsilon Q_{\Omega\setminus\{i\}}(i)\\
&=1-f_{{\Omega\setminus\{i\}},1,i}(\epsilon)Q_{\Omega\setminus\{i\}}(i)-\epsilon Q_{\Omega\setminus\{i\}}(i)\to 1-\eta_{1,i}Q_{\Omega\setminus\{i\}}(i)~~{\rm as}~\epsilon\to0.
\end{align*}
The above two equations imply $Q_S(i)=Q_{\Omega\setminus\{i\}}(i)=Q(i)$ for all $S\in\mathcal S_1$, $i\in S^c$.

(ii) For $S\in\mathcal S_1$ and $i\in S$, let $j\in S^c$. It holds that $X_{S,i,j}^\epsilon\in\mathcal X_S$ and $X_{\{1,i\},i,j}^\epsilon\in \X_{\{1,i\}}$. Hence,
\begin{align*}
	0&=\E^{P_S}[(X_{S,i,j}^\epsilon)_+]-\E^{Q_S}[(X_{S,i,j}^\epsilon)_-]
	=P_S(i)-f_{S,i,j}(\epsilon)Q_S(j)-\epsilon Q_S(S^c)\\
	&=P_S(i)-f_{S,i,j}(\epsilon)Q(j)-\epsilon Q(S^c)\to P_S(i)-\eta_{i,j}Q(j)~~{\rm as}~\epsilon\to0
\end{align*}
and
\begin{align*}
	0&=\E^{P_{\{1,i\}}}[(X_{\{1,i\},i,j}^\epsilon)_+]-\E^{Q_{\{1,i\}}}[(X_{\{1,i\},i,j}^\epsilon)_-]\\
	&=P_{\{1,i\}}(i)-f_{{\{1,i\}},i,j}(\epsilon)Q_{\{1,i\}}(j)-\epsilon Q_{\{1,i\}}(\Omega\setminus\{1,i\})\\
	&=P_{\{1,i\}}(i)-f_{{\{1,i\}},i,j}(\epsilon)Q(j)-\epsilon Q(\Omega\setminus\{1,i\})\to P_{\{1,i\}}(i)-\eta_{i,j}Q(j)~~{\rm as}~\epsilon\to0.
\end{align*}
We conclude that $P_S(i)=P_{1,i}(i)=\widetilde P(i)$
for $S\in \mathcal S_1$, $i\in S$.
Similarly, for $S\in\mathcal S_1^c$ and $i\in S^c$, 
let $j\in S$.
We have $X_{S,j,i}^\epsilon\in\mathcal X_S$ and $X_{\Omega\setminus\{1,i\},i,j}^\epsilon\in \X_{\Omega\setminus\{1,i\}}$. It holds that
\begin{align*}
	0&=\E^{P_S}[(X_{S,j,i}^\epsilon)_+]-\E^{Q_S}[(X_{S,j,i}^\epsilon)_-]=P_S(j)-f_{S,j,i}(\epsilon)Q_S(i)-\epsilon Q_S(S^c)\\
	&=P(j)-f_{S,j,i}(\epsilon)Q_S(i)-\epsilon Q_S(S^c)\to P(j)-\eta_{j,i}Q_S(i)~~{\rm as}~\epsilon\to0
\end{align*}
and
\begin{align*}
	0&=\E^{P_{\Omega\setminus\{1,i\}}}[(X_{\Omega\setminus\{1,i\},i,j}^\epsilon)_+]-\E^{Q_{\Omega\setminus\{1,i\}}}[(X_{\Omega\setminus\{1,i\},i,j}^\epsilon)_-]\\
	&=P_{\Omega\setminus\{1,i\}}(j)-f_{{\Omega\setminus\{1,i\}},j,i}(\epsilon)Q_{\Omega\setminus\{1,i\}}(i)-\epsilon Q_{\Omega\setminus\{1,i\}}(\{1,i\})\\
	&=P(j)-f_{{\Omega\setminus\{1,i\}},j,i}(\epsilon)Q_{\Omega\setminus\{1,i\}}(i)-\epsilon Q_{\Omega\setminus\{1,i\}}(\{1,i\})\\
 &\to P(j)-\eta_{j,i}Q_{\Omega\setminus\{1,i\}}(i)~~{\rm as}~\epsilon\to0.
\end{align*}
Hence, we have $Q_S(i)=Q_{\Omega\setminus\{1,i\}}=\widetilde Q(i)$ for $S\in \mathcal S_1^c$, $i\in S^c$.

(iii) For $i,j\ge 2$ and $i\neq j$, we have $X_{\{1,i\},i,j}^\epsilon\in\X_{\{1,i\}}$ and $X_{\{i\},i,j}^\epsilon\in\X_{\{i\}}$. Note that $\{1,i\}\in\mathcal S_1$ and $\{i\}\in\mathcal S_1^c$. We have
\begin{align*}
0&=\E^{P_{\{1,i\}}}[(X_{\{1,i\},i,j}^\epsilon)_+]-\E^{Q_{\{1,i\}}}[(X_{\{1,i\},i,j}^\epsilon)_-]\\
&=P_{\{1,i\}}(i)-f_{\{1,i\},i,j}(\epsilon)Q_{\{1,i\}}(j)-\epsilon Q_{\{1,i\}}(\Omega\setminus\{1,i\})\\
&=\widetilde P(i)-f_{\{1,i\},i,j}(\epsilon)Q(j)-\epsilon Q(\Omega\setminus\{1,i\})
\to \widetilde P(i)-\eta_{i,j}Q(j)~~{\rm as}~\epsilon\to0
\end{align*}
and 
\begin{align*}
	0&=\E^{P_{\{i\}}}[(X_{\{i\},i,j}^\epsilon)_+]-\E^{Q_{\{i\}}}[(X_{\{i\},i,j}^\epsilon)_-]\\
	&=P_{\{i\}}(i)-f_{\{i\},i,j}(\epsilon)Q_{\{i\}}(j)-\epsilon Q_{\{i\}}(\Omega\setminus\{i\})\\
	&= P(i)-f_{\{i\},i,j}(\epsilon)\widetilde Q(j)-\epsilon \widetilde Q(\Omega\setminus\{i\})
	\to  P(i)-\eta_{i,j}\widetilde Q(j)~~{\rm as}~\epsilon\to0.
\end{align*}
Combining the above two equations, we have
$\widetilde P(i)-\eta_{i,j}Q(j)= P(i)-\eta_{i,j}\widetilde Q(j)=0$ for all $i,j\ge 2$ and $i\neq j$. This completes the proof.\qed

\subsection{Theorem \ref{th:main-general} and a related corollary}\label{app:Th2}
In this section, 
we present the proof of Theorem \ref{th:main-general} and a related corollary that will be used in the proof of Theorem \ref{thm:AASPS}. 
 
Denote by $\mathcal U$ the set of all acts $V$ with uniform distribution on $(0,1)$, i.e., $\p(V\le x)=x$ for $x\in[0,1]$. Denote by $q_X$ the quantile function of $X\in \X$ under $\p$, i.e., $q_X(\alpha)=\inf\{x \in \R: \p(X\le x)\ge \alpha\}$ for $\alpha \in (0,1]$. Also write $q_X(0)= \inf\{x \in \R: \p(X\le x)>0\}$.
Since the probability space is standard, there exists $V_0\in \mathcal U$ such that $\sigma(V_0)=\mathcal F$. Define 
\begin{align*}
\Psi_{k}=\left\{  \left(  \frac{0}{2^{k}},\frac{1}{2^{k}}\right]  ,\left(
\frac{1}{2^{k}},\frac{2}{2^{k}}\right]  ,\dots,\left(  \frac{2^{k}-1}{2^{k}%
},\frac{2^{k}}{2^{k}}\right]  \right\},~~k\in\N
\end{align*}
as the partition of $\left( 0,1\right]  $ into segments of equal length $2^{-k}$. For $V\in\mathcal U$,
\begin{align*}
\Pi_{k}^{V}=V^{-1}\left(\Psi_{k}\right),~~k\in\N
\end{align*}
 is a partition of $\Omega$ in $\mathcal F$ such that $\p(E)=1/2^{k}$ for all $E\in
\Pi_{k}^{V}$. By setting $\mathcal F_{k}^{V}=\sigma(  \Pi
_{k}^{V})  =V^{-1}\left(  \sigma\left(  \Psi_{k}\right)
\right)  $ for all $k\in\mathbb{N}$, we have a filtration $\left\{  \mathcal F
_{k}^{V}\right\}_{k\in\mathbb{N}}$ in $\mathcal F$. As usual, $\mathcal F_\infty^V=\sigma(\bigcup_{k\in\N}\mathcal F_k^V)$.
Note that the $\sigma$-algebra $\sigma(
{\bigcup_{k\in\mathbb{N}}}
\Psi_{k})  $ is the Borel $\sigma$-algebra $\mathcal{B}\left(
0,1\right)  $ on $\left(  0,1\right)  $ because $%
{\bigcup_{k\in\mathbb{N}}}
\Psi_{k}$ is countable and separates the points of $(0,1)  $
(see, e.g., Theorem 3.3 of \cite{M57}). Then, in the special case that $V=V_0$, we have
\begin{align}\label{eq-sigma}
\mathcal F=\sigma\left(V_0\right)   &  =V_0^{-1}\left(  \mathcal{B}\left(
0,1\right)  \right)  =V_0^{-1}\left(  \sigma\left(
{\displaystyle\bigcup_{k\in\mathbb{N}}}
\Psi_{k}\right)  \right)  =\sigma\left(  V_0^{-1}\left(
{\displaystyle\bigcup_{k\in\mathbb{N}}}
\Psi_{k}\right)  \right) \notag\\
&=\sigma\left(
{\displaystyle\bigcup_{n\in\mathbb{N}}}
V_0^{-1}\left(  \Psi_{k}\right)  \right)
 =\sigma\left(
{\displaystyle\bigcup_{k\in\mathbb{N}}}
\Pi_{k}^{V_0}\right)  =\sigma\left(
{\displaystyle\bigcup_{k\in\mathbb{N}}}
\mathcal F_{k}^{V_0}\right)  =\mathcal F_{\infty}^{V_0}.
\end{align}

\begin{proof}[Proof of Theorem \ref{th:main-general}]
{\bf Sufficiency.}
Suppose that $\succsim$ can be represented by $\ex^{P,Q}$ for some $P,Q\in\mathcal M$ with $P\ac Q\ac\p$ and $0<P\le Q$.
The proofs of Axioms \ref{ax:SM}, \ref{ax:C} and \ref{ax:DA} are similar to sufficiency of Theorem \ref{th-main}, and we omit them here. Below we will verify Axiom \ref{ax:MC}. Let $m\in\R$, $\{A_n\}_{n\in\N}$ with $A_1\supseteq A_2\supseteq\dots$ and $\bigcap_{n\in\N}A_n=\varnothing$. The case of $m\ge 1$ is trivial, and we assume that $m<1$.
By standard calculation and noting that $P,Q\in\mathcal M$ are $\sigma$-additive, we have 
$$
\ex^{P,Q}(m\id_{A_{n}}+\id_{A_{n}^c})=\frac{P(A_n^c)+mQ(A_n)}{P(A_n^c)+Q(A_n)}\to 1~~{\rm as}~n\to\infty.
$$
Hence, for large enough $n$ we have $\ex^{P,Q}(m\id_{A_{n}}+\id_{A_{n}^c})>0=\ex^{P,Q}(0)$, and this yields Axiom \ref{ax:MC}.


{\bf Necessity.} 
As outlined in Section \ref{sec:ovth3}. The proof of necessity proceeds in four steps, which we will verify in detail below.


\textbf{Step 1}:
For $V\in\mathcal U$,
note that the certainty equivalent $U$ confined to $L^\infty(\Omega,\mathcal F_k^{V}, \p)$
satisfies Axioms \ref{ax:SM}, \ref{ax:C} and \ref{ax:DA}. 
By Theorem \ref{th-main},
there exist equivalent measures $P_k^{V}$ and $Q_k^{V}$ on $\mathcal F_k^{V}$ such that the certainty equivalent of $\succsim$ is as follows:
\begin{align}\label{eq-rep}
U = \ex^{P_k^{V},Q_k^{V}} \mbox{ on $L^\infty(\Omega,\mathcal F_k^{V}, \p)$}.
\end{align}
Without loss of generality we can impose $P_k^{V}(\Omega)=1$ for each $k\in \N$.
By the uniqueness result in Lemma \ref{prop-uniquedef}, $P_k^{V}$ and $Q_k^{V}$ are uniquely determined with $P_k^V\le Q_k^V$, and $Q_k^V(\Omega)\ge P_k^{V}(\Omega)=1$ is a constant for all $k\in\N$.
Denote by $Y_k^V=\d P_k^{V} /\d \p|_{\mathcal F_k^V}$ and $Z_k^V=\d Q_k^{V} /\d \p|_{\mathcal F_k^V}$, where $\p|_{\mathcal F_k^V}$ is the restriction of $\p$ on $\mathcal F_k^V$.
By the same uniqueness result, $P_k^{V}=P_{\ell}^{V}$ and $Q_k^{V}=Q_{\ell}^{V}$ for $\ell>k$ on $\mathcal F_k^{V}$.
Therefore, the sequences $\{Y_k^V\}_{k\in\N}$ and $\{Z_k^V\}_{k\in\N}$ are both martingales. Next, we aim to show that 
$\{Y_k^V\}_{k\in\N}$ and $\{Z_k^V\}_{k\in\N}$ are both uniformly integrable. 
We assume by contradiction that there exists $\epsilon_0>0$ and $\{k_n\}_{n\in\N}$ such that $\E[Z_{k_n}^V\id_{\{Z_{k_n}^V\ge n\}}]\ge \epsilon_0$ for all $n\in\N$.
Denote by $A_n=\bigcup_{\ell\ge n}\{Z_{k_\ell}^V\ge \ell\}$ for $n\in\N$. Obviously, $\{A_n\}_{n\in\N}$ is a decreasing sequence. It holds that
\begin{align}\label{eq-Ville}
\p(A_n)&=\p\left(\bigcup_{\ell\ge n}\{Z_{k_\ell}^V\ge \ell\}\right)\le \p\left(\sup_{\ell\ge n}Z_{k_\ell}^V\ge n\right)\notag\\
&\le \frac{\E[Z_{k_n}^V]}{n}
=\frac{Q_1^V(\Omega)}{n}\to 0~~{\rm as}~n\to\infty,
\end{align}
where the last inequality follows from Ville's inequality. Denote by $B_n=A_n\setminus \bigcap_{\ell\ge n}A_\ell$. Notice that $\{A_n\}_{n\in\N}$ is a decreasing sequence. It is not hard to see that $\{B_n\}_{n\in\N}$ is a decreasing sequence and $\bigcap_{n\in\N} B_n=\varnothing$. 
Further, let $m=1/\epsilon_0+1$ and $D_n=\{Z_{k_n}^V\ge n\}$. Applying \eqref{eq-Ville}, we have $\p(\bigcap_{n\in\N} A_n)=0$, and hence, $-m\id_{B_n}+\id_{B_n^c}=-m\id_{A_n}+\id_{A_n^c}$ $\p$-a.s. for all $n\in\N$, which implies that they are equal under $\succsim$. Moreover, it follows from Axiom \ref{ax:SM} and $D_n\subseteq A_n$ that $-m\id_{D_n}+\id_{D_n^c}\succsim-m\id_{A_n}+\id_{A_n^c}$. Therefore, for all $n\in\N$
\begin{align*}
U(-m\id_{B_n}+\id_{B_n^c})&=U(-m\id_{A_n}+\id_{A_n^c})
\le U(-m\id_{D_n}+\id_{D_n^c})\\
&=\ex^{P_{k_n}^V,Q_{k_n}^V}(-m\id_{D_n}+\id_{D_n^c})
=\frac{P_{k_n}^V(D_n^c)-mQ_{k_n}^V(D_n)}{P_{k_n}^V(D_n^c)+Q_{k_n}^V(D_n)}\\
&=\frac{P_{k_n}^V(D_n^c)-m\E[Z_{k_n}^V\id_{D_n}]}{P_{k_n}^V(D_n^c)+\E[Z_{k_n}^V\id_{D_n}]}
\le \frac{1-(1/\epsilon_0+1)\epsilon_0}{P_{k_n}^V(D_n^c)+\E[Z_{k_n}^V\id_{D_n}]}\\
&=-\frac{\epsilon_0}{P_{k_n}^V(D_n^c)+\E[Z_{k_n}^V\id_{D_n}]}<0=U(0),
\end{align*}
where the second equality follows from \eqref{eq-rep} and $D_n\in\mathcal F_{k_n}^V$, and the second inequality is due to the  assumption. Hence, we obtain that $-m\id_{B_n}+\id_{B_n^c} \prec 0$ for all $n\in\N$. This contradicts Axiom \ref{ax:MC}, and we conclude that $\{Z_k^V\}_{k\in\N}$ is uniformly integrable. Notice that $P_k^V\le Q_k^V$ implies $0\le Y_k^V\le Z_k^V$ for all $k\in\N$. We have that $\{Y_k^V\}_{k\in\N}$ is also uniformly integrable.

Applying the martingale convergence theorem to these two sequences,
we know that there exist $Y^V, Z^V\in L^\infty(\Omega,\mathcal F,\p)$ such that
$Y_k^V\to Y^V$ and $Z_k^V\to Z^V$ in $L^1(\Omega,\mathcal F,\p)$.
Define $P^V,Q^V\in\mathcal M$ satisfying  $P^V(A)=\int_{A}Y^V\d \p$ and $Q^V(A)=\int_{A}Z^V\d \p$ for all $A\in\mathcal F$, which means $\d P^V/\d\p=Y^V$ and $\d Q^V/\d\p=Z^V$.
It holds that $P^{V}$ coincide with $P_k^{V}$ on $\mathcal F_k^{V}$ for each $k\in \N$, and the same can be said about $Q^{V}$ and $Q_k^{V}$. 
Moreover, $P^V(\Omega)=1$ as $P_k^V(\Omega)=1$ for all $k\in\N$.
Hence, we have 
$$
U(X) = \ex^{P^{V},Q^{V}}(X),~~X\in \mathcal X_{V},
$$ 
where $\mathcal X_{V}:= \bigcup_{k\in \N}L^\infty (\Omega,\mathcal F_k^{V},\p)$. 

\textbf{Step 2}:
In this step, we aim to show that $P^{V}=P^{V_0}$ and $Q^{V}=Q^{V_0}$ on $\mathcal F_\infty^V$ for all $V\in\mathcal U$. To see this, define a filtration as follows 
\begin{align*}
	\mathcal F_k^{V_0,V}=\sigma\left(\Pi_k^{V_0}\cup\Pi_k^{V}\right),~~k\in\N
\end{align*}
and $\mathcal F_\infty^{V_0,V}=\sigma(\bigcup_{k\in\N}\mathcal F_k^{V_0,V})$. 
By \eqref{eq-sigma},
we have $\mathcal  F=\mathcal F_\infty^{V_0}\subseteq\mathcal F_\infty^{V_0,V}\subseteq \mathcal F$, which implies $\mathcal F_\infty^{V_0,V}=\mathcal F$.
Following the same arguments as previous, there exist $P^{V_0,V}$ and $Q^{V_0,V}$ defined on $\mathcal F$ such that 
\begin{align*}
U(X) = \ex^{P^{V_0,V},Q^{V_0,V}}(X),~~ X\in \bigcup_{k\in \N}L^\infty (\Omega,\mathcal F_k^{V_0,V},\p).
\end{align*}
The uniqueness result in Lemma \ref{prop-uniquedef} implies $P^{V_0,V}=P^{V_0}$ on $\bigcup_{k\in\N}\mathcal F_k^{V_0}$. Since $\bigcup_{k\in\N}\mathcal F_k^{V_0}$ is a $\pi$-system, the $\pi$-$\lambda$ theorem implies $P^{V_0,V}=P^{V_0}$ on $\mathcal F_\infty^{V_0}=\mathcal F$.
Similarly, we can obtain $P^{V_0,V}=P^{V}$ on $\mathcal F_\infty^{V}$.
Noting that $\mathcal F_\infty^{V}\subseteq \mathcal F$, we conclude that $P^V=P^{V_0}$ on $\mathcal F_\infty^{V}$. For simplicity, we denote by $P=P^{V_0}$ and $Q=Q^{V_0}$, and it holds that
\begin{align}\label{eq-repUV}
U(X)=\ex^{P,Q}(X),~~X\in \bigcup_{V\in\mathcal U}\mathcal X_V.
\end{align}

\textbf{Step 3}:
In this final step, we aim to show that the representation \eqref{eq-repUV} holds for all $X\in \X$. 
For  $X\in \X$, let  $V\in \mathcal U$ be such that $X=q_X(V)$ (see Lemma A.32 of \cite{FS16} for the existence of such $V$).
Let $\underline{V}_k = 2^{-k} \lceil 2^{k}V -1\rceil  $ and $\overline{V}_k = 2^{-k} \lceil 2^{k}V\rceil $
 for $k\in \N$, where $\lceil x \rceil$ represents the least integer not less than $x$; that is, $(\underline{V}_k,\overline{V}_k]$ is the interval in $\Psi_k$ that contains $V$. Denote by
  $\underline X_k=q_X(\underline{V}_k)\in \mathcal X_V$ and $\overline X_k= q_X(\overline{V}_k)\in \mathcal X_V$.  Clearly $\underline X_k \le X \le \overline X_k$. 
By Axiom \ref{ax:SM}, we have 
$$
U(\underline X_k ) \le U(X) \le U( \overline X_k),
$$
implying 
\begin{align}\label{eq:sandwich}
\ex^{P,Q}(\underline X_k ) \le U(X) \le \ex^{P,Q}( \overline X_k).
\end{align}
Let $$A^{V,k}_j= \left\{ V\in \left(\frac{j-1}{2^k} , \frac{j}{2^k}\right]\right\} \mbox{~~for $j\in [2^k]$.}$$
Note that 
\begin{align*}
\E^{P+Q} [|\overline X_k- X|] &  \le
\E^{P+Q} [|\overline X_k- \underline X_k|] 
\\ &\le
\textstyle\sum_{j=1}^{2^k}  \left(q_X\left(\frac{j}{2^k}\right) - q_X\left(\frac{j-1}{2^k}\right) \right) (P+Q) \left( A^{V,k}_j \right)
\\& \le (q_X(1) - q_X(0)) \max_{j \in [2^k]}  (P+Q) \left(A^{V,k}_j \right) 
\to 0,
\end{align*}
where the convergence holds because $P+Q$ is absolutely continuous with respect to $\p$ and $\p(A^{V,k}_j )=2^{-k}$.
Using the continuity result of Proposition \ref{prop:technical} (i), we have 
$$
|\ex^{P,Q}(\overline  X_k )  - \ex^{P,Q}(  \underline X_k)|\to 0~~{\rm and}~~ 
|\ex^{P,Q}(\overline X_k )  - \ex^{P,Q}( X )|\to 0.
$$
Therefore, by \eqref{eq:sandwich}, we conclude $U(X) = \ex^{P,Q}(X)$. 

\textbf{Step 4}:
Notice that from Proposition \ref{prop:technical} (ii) it follows that Axiom \ref{ax:SM} implies $P\ac Q \ac \mathbb{P}$.
By (vi) and (iv) of Lemma \ref{lm-property}, we know that $U$ is concave, which implies $P\le Q$ by applying Proposition \ref{prop:technical} (iii).
Hence, $\succsim$ can be represented by $\ex_{\alpha}^{\widetilde P,\widetilde Q}=\ex^{P,Q}$ with $\alpha=P(\Omega)/(P(\Omega)+Q(\Omega))$. The uniqueness of $(P,Q,\alpha)$ follows from Lemma \ref{prop-uniquedef}.
This completes the proof.
\end{proof}

Recall that $B_0(\R)$ defined in Section \ref{sec:AA} represents the set of all real-valued measurable simple functions. 
By the proof of Theorem \ref{th:main-general}, the necessity also holds true if we restrict the preference on $B_0(\R)$. This result is summarized as the following corollary, which will be used in the proof of Theorem \ref{thm:AASPS}.

\begin{corollary}\label{cor:thm2}
Let $(\Omega,\mathcal F,\p)$ be a standard probability space. 
For a preference relation $\succsim$ on $B_0(\R)$, if
Axioms \ref{ax:SM}, \ref{ax:C}, \ref{ax:MC} and \ref{ax:DA} hold, then there exist $P,Q\in\mathcal M_1$ and $\alpha\in(0,1/2]$ with $\alpha P\le (1-\alpha)Q$ and $P\ac Q\ac \p$
such that $\succsim$ is represented by $\ex_{\alpha}^{P,Q}$. In this case, $P,Q,\alpha$ are unique.
\end{corollary}

\begin{proof}
Following a similar arguments of Steps 1, 2 and 3 of the proof of Theorem \ref{th:main-general} (note that all acts in Steps 1 and 2 are the elements of $B_0(\R)$), we know that there exists $P,Q\in\mathcal M$ such that $\succsim$ can be represented by $\ex^{P,Q}$ on $B_0(\R)$.
From the proofs of Lemmas \ref{lm-property} and \ref{prop-uniquedef} and Proposition \ref{prop:technical} (iii), 
one can check that these results still hold if we assume that the duet expectiles are defined on $B_0(\R)$.
Thus, we can directly apply them, and sharing the similar arguments in Step 4 of the proof of Theorem \ref{th:main-general} completes the proof.
\end{proof}

\end{appendix}


\newpage

\begin{center}
    \large ~\\
    Supplement to ``Disappointment concordance and duet expectiles"
\end{center}
\begin{center}
 \vspace{0.1cm} Fabio Bellini,   Tiantian Mao,   Ruodu Wang,  and Qinyu Wu
\end{center}
 
~\\

\begin{appendix}
\setcounter{section}{1}

This supplement is organized as follows. Appendix \ref{sec:properties} contains technical properties of duet expectiles, which proves the three statements in Proposition \ref{prop:technical}.  In Appendix \ref{app:omproofs}, we provide the omitted proofs of the results in Sections \ref{sec:solo}, \ref{sec:discuss} and \ref{sec:AA}. In Appendix \ref{app:A}, we provide some counterexamples to support certain claims made in the paper and to clarify several technical conditions.

\section{Properties of duet expectiles}\label{sec:properties}


\subsection{Some standard  properties}\label{subsec:sp}
We first state several standard  properties.
Recall the following notation in Section \ref{sec:7.1}.
For $P\in \M$, let $\|\cdot\|_{\infty}^P$ and
$\|\cdot\|_{1}^P$ be the $L^\infty$-norm and $L^1$-norm under $P$, respectively, i.e., $\Vert X\Vert^P_{\infty}=\esssup_{P}|X|$ and
$\Vert X\Vert^P_1=\E^P[|X|]$.  
Moreover, $P\vee Q$ is  the maximum   of two measures $P$ and $Q$, namely the smallest element $R$ in $\M$ satisfying $R\ge P$ and $R\ge Q$. 
Similarly, $P\wedge Q$ is the minimum   of $P$ and $Q$. 


\begin{proposition}\label{prop-gextrivial}
Let $P,Q\in\mathcal M$.
The duet expectile has the following properties.
\begin{enumerate}[(i)]
\item[(i)]  $\ex^{P,Q}$ is monotone, i.e., $\ex^{P,Q}(Y)\ge \ex^{P,Q}(X)$ whenever $X,Y\in\X$ and $Y\ge X$.
\item[(ii)]  $\ex^{P,Q}$ is translation  invariant, i.e., $\ex^{P,Q}(X+c)= \ex^{P,Q}(X)+c$ for $c\in \R$ and $X\in \X$.
\item[(iii)] $\ex^{P,Q}$ is positively homogeneous, i.e., $\ex^{P,Q}(\lambda X)=\lambda \ex^{P,Q}(X)$ for   $\lambda\ge 0$ and $X\in\X$.
\item[(iv)] $\ex^{P,Q}(X)=-\ex^{Q,P}(-X)$ for all $X\in\X$.
\item[(v)] $\ex^{P,Q}$ is  continuous in the following  senses: for all $X,Y\in\mathcal X$,
\begin{enumerate}[(a)]
\item  $|\ex^{P,Q}(X)- \ex^{P,Q}(Y)|\le \Vert X-Y\Vert_{\infty}^{P\vee Q}$.
\item    If   $ (P\wedge Q)(\Omega) \ne 0$, then 
$$|\ex^{P,Q}(X)- \ex^{P,Q}(Y)|\le  \frac{ \|X-Y\|^{P\vee Q}_1}{(P\wedge Q)(\Omega)} 
.$$
\end{enumerate}
\end{enumerate}
\end{proposition}
 
\begin{remark}
If $P $ and $Q$ are not mutually singular,  then by   
Proposition \ref{prop-gextrivial} (v), $\ex^{P,Q}$ is Lipschitz continuous with respect to  $\Vert \cdot \Vert_{\infty}^{P\vee Q}$ and 
$\Vert\cdot \Vert_1^{P\vee Q}$.
The Lipschitz coefficient in point (b) of Proposition \ref{prop-gextrivial} (v) is generally sharp.
For instance, in the special case of $\ex_\alpha^P$,  by setting  $P\in \M_1$ and  $\alpha Q= (1-\alpha)   P$, 
point (b) yields
		$$|\ex^{P}_\alpha (X)- \ex^{P}_\alpha(Y)|\le \left(\frac{1-\alpha}{\alpha} \vee \frac{\alpha}{1-\alpha}\right)  { \|X-Y\|^{P }_1} 
			,$$
			which gives the known Lipschitz coefficient  for the classic expectile $\ex_{\alpha}^P$ in \citet[Theorem 10]{BKMR14}.
\end{remark}

\begin{proof}[Proof of Proposition \ref{prop-gextrivial}]
The first four statements are straightforward to check from the definition. 
Point (a) of (v) is a direct consequence of (i) and (ii). 
Below we show point (b) of (v). 
Denote by $q=Q(\Omega)$ and  let $e_X=\ex^{P,Q}(X)$ and  $e_Y=\ex^{P,Q}(Y)$ for $X,Y\in\mathcal X$.
Assume without loss of generality that $Y\ge X$; otherwise we can take $Y\vee X$ in place of $Y$.
Clearly, $e_Y\ge e_X$ by (i).
Let the function $f_X:\R\to \R$ be given by
$$
f_X(x) = \E^Q[X] + \E^{P-Q} [(X-x)_+] -  q x, 
$$
and similarly we define $f_Y$. 
By the definition of $e_X$ and $e_Y$, we have $f_X(e_X)=f_Y(e_Y)=0$.
Note that for any $x\in \R$,
\begin{align*}
f_Y(x) - f_X(x) &= \E^Q[Y-X] + \E^{P-Q} [ (Y-x)_+ - (X-x)_+] 
\\ &\le \E^Q[Y-X] + \E^{(P-Q)_+ } [ (Y-x)_+ - (X-x)_+]  
\\& \le \E^Q[Y-X] + \E^{(P-Q)_+ } [  Y-X ]  
 =  \E^{P\vee Q} [Y-X].
\end{align*}
Moreover, for $y\ge x$, 
\begin{align*}
f_Y(x) - f_Y(y) &= \E^{P-Q} [ (Y-x)_+ - (Y-y)_+]  - q( x-y)
\\ &\ge -  \E^{(Q-P)_+} [ (Y-x)_+ - (Y-y)_+]    - \E^Q[x-y]
\\& \ge  -  \E^{(Q-P)_+} [ y-x]    + \E^Q[y-x]\\ 
 &=   \E^{P\wedge Q}[y-x] =  (y-x) (P\wedge Q)(\Omega) .
\end{align*}
Putting the above two inequalities together, and using $f_Y(e_Y)=f_X(e_X)$, we get 
$$
(e_Y-e_X)(P\wedge Q)(\Omega) \le  \E^{P\vee Q} [Y-X],
$$
and this yields the desired statement.  
\end{proof}

\subsection{Strict monotonicity}\label{subsub:SM}
We recall two cases of $\mathcal X$ in Section \ref{sec:axiom}. That is, $\Omega$ is finite, in which case $\X$ can be identified with $\R^n$, 
or $\Omega$ is infinite, in which case $\X=L^\infty(\Omega,\mathcal F,\p)$ is a set of all bounded acts under some fixed reference atomless probability $\p$.
We say briefly that \ref{ax:SM} holds if Axiom \ref{ax:SM} holds for the preference represented by $\ex^{P,Q}$. First, we give a simple equivalent condition on $P$ and $Q$ for \ref{ax:SM} to hold.

\begin{proposition}\label{prop-finSM}
For $P,Q\in\mathcal M$, \ref{ax:SM} holds if and only if 
$P,Q>0$ in case $\Omega$ is finite, or $P\ac Q\ac \p$ in case $\Omega$ is infinite. 
\end{proposition}

\begin{proof}[Proof of Proposition \ref{prop-finSM}]
First, we consider the case that $\Omega=[n]$. Assume that $P(i)$, $Q(i)>0$ for all $i\in[n]$. Let $X,Y\in\X$ be such that $X\le Y$ and $X\neq Y$, and denote by $x^*=\ex^{P,Q}(X)$ and $y^*=\ex^{P,Q}(Y)$. Without loss of generality, we assume that $X(1)<Y(1)$. Define $f_Z(x)=\E^P[(Z-x)_+]-\E^Q[(Z-x)_-]$ for $Z\in\mathcal X$ and $x\in\R$. Denote by
$p_1=P(1)$, $q_1=Q(1)$, $x_1=X(1)$ and $y_1=Y(1)$.
It holds that 
\begin{align}\label{eq-PRSM}
&~~~~f_Y(x^*)
=\E^P[(Y-x^*)_+]-\E^Q[(Y-x^*)_-]\notag\\
&=f_X(x^*)+(\E^P[(Y-x^*)_+]-\E^P[(X-x^*)_+])-(\E^Q[(Y-x^*)_-]-\E^Q[(X-x^*)_-])\notag\\
&\ge p_1[(y_1-x^*)_+-(x_1-x^*)_+]-q_1[(y_1-x^*)_--(x_1-x^*)_-],
\end{align}
where we have used $f_X(x^*)=0$ in the inequality. If $y_1>x^*$, then \eqref{eq-PRSM} implies $f_Y(x^*)\ge p_1(y_1-x_1)>0$. If $y_1\le x^*$, then \eqref{eq-PRSM} implies $f_Y(x^*)\ge q_1(y_1-x_1)>0$. Hence, we have $f_Y(x^*)>0$. It follows from the definition of $\ex^{P,Q}$   that $y^*>x^*$. This implies \ref{ax:SM} holds. Conversely, suppose that \ref{ax:SM} holds. We first show that $P(i)>0$ for all $i\in[n]$ at first.
Assume by contradiction that $P(i)=0$ for some $i\in[n]$.
One can check that $\ex^{P,Q}(k\cdot\id_{\{i\}})=0$ for all $k>0$. This yields a contradiction to \ref{ax:SM}. Hence, we conclude that $P(i)>0$ for all $i\in[n]$. Assume now by contradiction that $Q(i)=0$ for some $i\in[n]$. Noting that $\p(\Omega\setminus \{i\})>0$, we have $\ex^{P,Q}(k\cdot\id_{\{1\}})=0$ for all $k<0$. This also yields a contradiction. Hence, we conclude that $P(i)Q(i)>0$ for all $i\in[n]$. 

Next, we consider the case that $\Omega$ is infinite. Assume that $P\ac Q\ac \p$. Let $X,Y\in \X$ be such that $X\le Y$ and $X\neq Y$. Denote by $A=\{X<Y\}$. Noting that $\p(A)>0$, we have $P(A)Q(A)>0$. Similar to the finite case by substituting the state $\{1\}$ into $A$, we obtain $\ex^{P,Q}(X)<\ex^{P,Q}(Y)$. Conversely, assume by contradiction that $P\not\ac \p$ or $Q\not\ac \p$. We only focus on the case that $P\not\ac \p$ as the other is similar. If $\p(A)>0$ and $P(A)=0$ for some $A\in\mathcal F$, then we have $\ex^{P,Q}(k\cdot \id_{A})=0$ for all $k>0$. This contradicts \ref{ax:SM}, and hence, we have $ \p$ is absolutely continuous with respect to $P$.
Assume now $P(A)>0$ and $\p(A)=0$ for some $A\in\mathcal F$. If $Q(A^c)>0$, then it is easy to calculate that $\ex^{P,Q}(k\cdot \id_{A})=kP(A)/(P(A)+Q(A^c))$ for $k> 0$ which is strictly increasing in $k$. This contradicts to the fact that $\p$-almost surely equal objects are identical. If $Q(A^c)=0$ and $Q(A)>0$, then we have $\ex^{P,Q}(k\cdot\id_{A^c})=0$ for all $k<0$ as $P(A)>0$. Noting that $\p(A^c)=1$, this yields a contradiction to \ref{ax:SM}. Hence, we have $ P$ is absolutely continuous with respect to $\p$. This completes the proof.
\end{proof}

\subsection{Convexity and concavity}\label{subsub:cvcx}

In this section, we aim to investigate concavity and convexity of duet expectiles. We focus on the space and measures satisfying the following assumption. 

\begin{assumption}\label{assum:P}
There exist $S_1,S_2,S_3\in\mathcal F$ such that $\{S_1,S_2,S_3\}$ is a partition of $\Omega$,  with $P(S_1)$, $P(S_2)>0$ and $Q(S_1)$, $Q(S_3)>0$.
\end{assumption}

From Proposition \ref{prop-finSM} it follows that Assumption \ref{assum:P} always holds whenever Axiom \ref{ax:SM} holds on a finite or infinite space.

\begin{proposition}\label{prop-cv}
For $P,Q\in\mathcal M$, 
assume that Assumption \ref{assum:P} holds. The duet expectile $\ex^{P,Q}$ is concave (resp.~convex) if and only if $P\le Q$ (resp.~$P\ge Q$).
\end{proposition}

\begin{remark}
Sufficiency of Proposition \ref{prop-cv} always holds on a general space.  
Theorem \ref{th:riskaverse1} and Proposition \ref{prop-cv} lead to the following equivalent conditions under   \ref{ax:SM}:  
(i) $\ex^{P,Q} \le \E^{\widetilde P}$ (resp.~$\ex^{P,Q} \ge \E^{\widetilde P}$);
 (ii) $\ex^{P,Q} \le \E^{\widetilde Q}$ (resp.~$\ex^{P,Q} \ge \E^{\widetilde Q}$);
(iii) $\ex^{P,Q}$ is concave (resp.~convex);
(iv) $P\le Q$ (resp.~$P\ge Q$). 
\end{remark}

To prove Proposition \ref{prop-cv}, we need the following lemma.
For a functional $U:\mathcal X\to\R$, its \emph{acceptance set} is defined by $\{X\in\mathcal X: U(X)\le 0\}$.

\begin{lemma}[Proposition 4.6 of \cite{FS16}]\label{lm-setcx}
Let $U:\mathcal X\to\R$ be a monotone, translation-invariant and positively homogeneous functional and $\mathcal A$ be its acceptance set. Then, $U:\mathcal X\to\R$ is convex if and only if $\mathcal A$ is a convex cone.
\end{lemma}

\begin{proof}[Proof of Proposition \ref{prop-cv}]
We first consider the convexity case. 
The acceptance set of $\ex^{P,Q}$ has the following form:
\begin{align*}
	\mathcal A:=\{X\in\X: \E^P[X_+]-\E^{Q}[X_-]\le 0\}
	=\{X\in\X: \E^Q[X]+\E^{P-Q}[X_+]\le 0\}.
\end{align*}
Since $\ex^{P,Q}$ is monotone, translation-invariant and positively homogeneous (see Proposition \ref{prop-gextrivial}),
it follows from Lemma \ref{lm-setcx} that $\ex^{P,Q}$ is convex if and only if $\mathcal A$ is a convex cone. Therefore, it suffices to verify:
\begin{align*}
	P\ge Q \Longleftrightarrow \mathcal A~{\rm is~a~convex~cone}.
\end{align*}
Suppose that $P\ge Q$. It is obvious that $\mathcal A$ is a convex cone since $X\mapsto \E^Q[X]$ is linear and $X\mapsto \E^{P-Q}[X_+]$ is convex and positively homogeneous. Conversely, suppose that $\mathcal A$ is a convex cone.
We assume by contradiction that there exists $A\in\mathcal F$ such that $P(A)<Q(A)$. By Assumption \ref{assum:P}, we know that $P(S_1)$,$P(S_2), Q(S_1)$,$Q(S_3)>0$, where $\{S_1,S_2,S_3\}\subseteq \mathcal F$ is a partition of $\Omega$. Hence, we have $P(A\cap S_i)<Q(A\cap S_i)$ for some $i=1,2,3$. Define 
\begin{align*}
B_1=A\cap S_i,~B_2=S_{i_1}\cup (S_i\setminus A)~~{\rm and}~~B_3=S_{i_2},
\end{align*}
where $i_1,i_2\in\{1,2,3\}\setminus \{i\}$ and $i_1\neq i_2$, and we can assume that $P(B_2)>0$ and $Q(B_3)>0$. Indeed, for $i=1,2,3$, $(i_1,i_2)$ can be $(2,3)$, $(1,3)$ and $(2,1)$, respectively. Denote by $p_j=P(B_j)$ and $q_j=Q(B_j)$ for $j=1,2,3$. By construction, we have $p_1<q_1$, $p_2>0$ and $q_3>0$. 
Note that $\{B_1,B_2,B_3\}$ is a partition of $\Omega$. We define $X,Y\in\X$ as follows.
\begin{center}
	\begin{tabular}
		[c]{c|c|c|c}
		& $B_1$ & $B_2$ & $B_3$  \\\hline
		$X$ & $p_2q_3$ & $p_1q_3$ & $-2p_1p_2$ \\
		$Y$ & $-p_2q_3$ & $2q_1q_3$ & $-p_2q_1$ 
	\end{tabular}
\end{center}
One can check that  $\E^P[X_+]-\E^Q[X_-]=\E^P[Y_+]-\E^Q[Y_-]=0$ and 
\begin{align*}
	\E^P[(X+Y)_+]-\E^Q[(X+Y)_-]=(q_1-p_1)p_2q_3>0.
\end{align*}
Hence, we have $X,Y\in\mathcal A$ and $X+Y\notin \mathcal A$. This yields a contradiction, that completes the proof of the convexity case. The concavity case follows by  (iv) of Proposition \ref{prop-gextrivial} that gives $\ex^{P,Q}(X)=-\ex^{Q,P}(-X)$ that is concave if and only if $\ex^{Q,P}$ is convex, that happens if and only if $Q \geq P$.
\end{proof}

\subsection{Dual representation}\label{sec:DR}
The final result in this section is a dual representation of duet expectiles as a worst-case expectation that is the basis of the Gilboa-Schmeidler representation of duet expectiled utilities.  
\begin{proposition}\label{prop:dual}
Let $P,Q\in\mathcal M_1$ and assume that \ref{ax:SM} holds. If $\alpha P\le (1-\alpha) Q$, then the following dual representation holds: for each $X\in \X$,
 $$
\ex_\alpha^{P,Q}(X) = \inf_{R\in \mathcal R}\E^R[X],~\mbox{with~} \mathcal R=\left \{R\in \M_1: \alpha \esssup_Q \frac{\d R}{\d Q} \le (1-\alpha) \essinf_P \frac{\d R}{\d P}\right\},
 $$
 where $\d R/d P$ and $\d R/\d Q$ are Radon-Nikodym derivatives, and $\essinf_H$ and $\esssup_H$ represent the essential infimum and essential supremum under $H\in\mathcal M_1$, respectively. 
Further, it holds that
 $$\ex_\alpha^{P,Q}(X) = \frac{\alpha \E^P [X \id_{E_X} ] + (1-\alpha) \E^Q[X  \id_{D_X}]} {\alpha P(E_X) + (1-\alpha) Q(D_X)},
 $$
 where  $D_X:=\{X < \ex_\alpha^{P,Q}(X) \}$ and $E_X:=\{X \geq \ex_\alpha^{P,Q}(X) \}$. 
 \end{proposition} 
 
\begin{proof}[Proof of Proposition \ref{prop:dual}]
Assume that $\Omega$ is infinite, the finite case being similar. From Proposition \ref{prop-finSM}
 it follows that $P \ac Q \ac \mathbb{P}$. We first check that for each $ R \in \mathcal{R}$ and for each $X \in \mathcal{X}$ it holds that $\E^R[X] \geq \ex_\alpha^{P,Q}(X)$. From translation invariance of $\ex_\alpha^{P,Q}$ (see Proposition \ref{prop-gextrivial}), we can assume without loss of generality that $\ex_\alpha^{P,Q}(X) =0$. Then
\begin{align*}
 \E^R[X] &=\E^R[X_+]-\E^R[X_-] = \E^P \left [\frac{\d R}{\d P}X_+\right]-\E^Q \left [\frac{\d R}{\d Q} X_-\right ] \\
&\geq  
 \left (\essinf_P \frac{\d R}{\d P} \right)  \E^P[X_+] 
 - \left ( \esssup_Q \frac{\d R}{\d Q} \right)  \E^Q[X_-] \\
 &\geq 
\left ( \esssup_Q \frac{\d R}{\d Q} \right)\left ( \frac{\alpha \E^P[X_+]}{1-\alpha} - \E^Q[X_-] \right) =0,
\end{align*}
where the last equality follows from $\ex_\alpha^{P,Q}(X) =0$. Let now $R_0
$ be given by
$$
\frac{\d R_0}{\d\mathbb{P}} =  \frac{\alpha \frac{\d P}{\d \mathbb{P}}  \id_{\{X \geq z\}}+ (1-\alpha) \frac{\d Q}{\d\mathbb{P}} \id_{\{X < z \}} }
{\alpha P(X \geq z )+(1-\alpha) Q(X < z)},
$$
where $z:= \ex_\alpha^{P,Q}(X)$. Since $\frac{\d R_0}{\d\mathbb{P}} >0$, 
it follows that $R_0 \ac \mathbb{P}$, and 
\begin{align*}
\frac{\d R_0}{\d Q} &= \frac{\d R_0}{\d\mathbb{P}} \frac{\d\mathbb{P}}{\d Q} = \frac{\alpha\frac{\d P}{\d Q} \id_{\{X \geq z\}}+ (1-\alpha)  \id_{\{X < z \}}}{ \alpha P(X \geq z)+(1-\alpha) Q(X < z)} \leq \frac{1-\alpha}{{\alpha P(X \geq z)+(1-\alpha) Q(X < z)}}; \\
\frac{\d R_0}{\d P} &= \frac{\d R_0}{\d\mathbb{P}} \frac{\d\mathbb{P}}{\d P} =
 \frac{ \alpha  \id_{\{X \geq z \}}+ (1-\alpha)  \frac{\d Q}{\d P}   \id_{\{X < z \}}}{\alpha P(X \geq z)+(1-\alpha) Q(X < z)} \geq \frac{\alpha}{{\alpha P(X \geq z)+ (1-\alpha) Q(X < z)}},
\end{align*}
where the assumption $\alpha P \leq (1-\alpha) Q$ has been used. Since $Q \ac P$, it follows that
$$
\esssup_Q \frac{\d R_0}{\d Q} \le \essinf_P \frac{\d R_0}{\d P},
$$
which implies $R_0 \in \mathcal{R}$. Finally we show that $\E^{R_0}[X]=\ex^{P,Q}(X)$. Indeed,
\begin{align*}
\E^{R_0}[X] &= \frac{\alpha \E^P\left[X \id_{\{X \geq z \}}\right] + (1-\alpha) \E^Q\left[X \id_{\{X < z \}}\right] }{\alpha P(X \geq z)+(1-\alpha) Q(X < z)} \\
&= z + \frac{\alpha \E^P[(X - z)_+] + (1-\alpha) \E^Q[(X - z)_-] }{\alpha P(X>z)+ (1-\alpha) Q(X \leq z)} 
 = z.  
\end{align*}
\noindent
Under the condition $\alpha P \geq (1-\alpha)Q$ a similar reasoning gives a representation of $\ex_\alpha^{P,Q}$ as a supremum of expected values, with respect to the same set of probability measures. When $P=Q$ the dual representation of expectiles given in \cite{BKMR14} is obtained as a special case. 
\end{proof}

\section{Remaining proofs}\label{app:omproofs}

\subsection{Proofs of results in Section \ref{sec:solo}}\label{app:sec-4-pf}

\subsubsection{Proposition \ref{prop:classic}}

\begin{proof}[Proof of Proposition \ref{prop:classic}]
Sufficiency follows immediately from Theorem \ref{th:main-general}. 
The proof of necessity is similar to the proof of Theorem \ref{th:main-general}. For clarity, we give the complete proof below.
The notation $X\laweq Y$ means that $X$ and $Y$ has the same distribution under $\p$.
We collect some notation used in the proof of Theorem \ref{th:main-general}. Denote by $q_X$ the quantile function of $X\in \X$ under $\p$. Let $V$ be an act with uniform distribution on $(0,1)$. It holds that $X\laweq q_X(V)$ for all $X\in\X$. Define 
\begin{align*}
\Psi_{k}=\left\{  \left(  \frac{0}{2^{k}},\frac{1}{2^{k}}\right]  ,\left(
\frac{1}{2^{k}},\frac{2}{2^{k}}\right]  ,\dots,\left(  \frac{2^{k}-1}{2^{k}%
},\frac{2^{k}}{2^{k}}\right]  \right\},~~k\in\N.
\end{align*}
and 
\begin{align*}
\Pi_{k}=V^{-1}\left(\Psi_{k}\right),~~k\in\N.
\end{align*}
Let $\mathcal F_{k}=\sigma(  \Pi
_{k})  =V^{-1}\left(  \sigma\left(  \Psi_{k}\right)
\right)  $ for all $k\in\mathbb{N}$, and 
 $\left\{  \mathcal F
_{k}\right\}_{k\in\mathbb{N}}$ is a filtration in $\mathcal F$. By Theorem \ref{th-main},
there exist equivalent measures $P_k$ and $Q_k$ on $\mathcal F_k$ with $0<P_k\le Q_k$ such that the certainty equivalent of $\succsim$ is as follows:
\begin{align}\label{eq-rep1}
U = \ex^{P_k,Q_k} \mbox{ on $L^\infty(\Omega,\mathcal F_k,\p)$}.
\end{align}
We aim to verify that $\widetilde{P}_k=\widetilde{Q}_k=\p|_{\mathcal F_k}$, where $\p|_{\mathcal F_k}$ is the restriction of $\p$ on $\mathcal F_k$. 
Note that $\p(A_i)=1/2^k$ for all $i\in[2^k]$, where 
\begin{align*}
A_i=V^{-1}\left(\left(\frac{i-1}{2^k},\frac{i}{2^k}\right]\right),~~i=1,\dots,2^k.
\end{align*}
It suffices to show that $\widetilde{P}_k(A_i)=\widetilde{Q}_k(A_i)=1/2^k$ for all $i\in[2^k]$.
Since the distributions of $\id_{A_i}$ under $\p$ are identical for all $i\in[2^k]$, and the same result also holds for all $-\id_{A_i}$, the probabilistic sophistication implies that there exist $c,c'\in\R$ such that 
\begin{align*}
U(\id_{A_i})=c~~{\rm and }~~U(-\id_{A_i})=c',~\mbox{for every }i\in[2^k].
\end{align*}
Combining with \eqref{eq-rep1} and noting that $A_i\in\mathcal F_k$ for all $i\in[2^k]$, straightforward calculation yields
\begin{align*}
U(\id_{A_i})=\ex^{P_k,Q_k}(\id_{A_i})=\frac{P_k(A_i)}{P_k(A_i)+Q_k(A_i^c)}=c,~~i\in[2^k]
\end{align*}
and
\begin{align*}
U(-\id_{A_i})= \ex^{P_k,Q_k}(-\id_{A_i})=-\frac{Q_k(A_i)}{P_k(A_i^c)+Q_k(A_i)}=c',~~i\in[2^k].
\end{align*}
This implies $Q_k(A_i^c)/P_k(A_i)=(1-c)/c$ and $P_k(A_i^c)/Q_k(A_i)=-(1+c')/c'$ for all $i\in[2^k]$.
Hence, we have
\begin{align*}
&\frac{Q_k(A_i^c)}{P_k(A_i)}=\frac{\textstyle\sum_{j=1}^{2^k}Q_k(A_j^c)}{\textstyle\sum_{j=1}^{2^k}P_k(A_j)}=\frac{(2^k-1)Q_k(\Omega)}{P(\Omega)},~~i\in[2^k];\\
&\frac{P_k(A_i^c)}{Q_k(A_i)}=\frac{\textstyle\sum_{j=1}^{2^k}P_k(A_j^c)}{\textstyle\sum_{j=1}^{2^k}Q_k(A_j)}=\frac{(2^k-1)P_k(\Omega)}{Q(\Omega)},~~i\in[2^k].
\end{align*}
This yields
\begin{align*}
\widetilde{P}_k(A_i)=\frac{1-\widetilde{Q}_k(A_i)}{2^k-1}~~{\rm and}~~\widetilde{Q}_k(A_i)=\frac{1-\widetilde{P}_k(A_i)}{2^k-1},~~i\in[2^k].
\end{align*}
Hence, we conclude that $\widetilde{P}_k(A_i)=\widetilde{Q}_k(A_i)=1/2^k$ for all $i\in[2^k]$, and this implies $\widetilde{P}_k=\widetilde{Q}_k=\p|_{\mathcal F_k}$. Therefore, \eqref{eq-rep1} can be reformulated as
$U = \ex^{\p}_{\alpha_k}$ on $L^\infty(\Omega,\mathcal F_k,\p)$, where $\alpha_k\in(0,1/2]$ for all $k\in\N$.
Note that $\left\{ \mathcal F
_{k}\right\}_{k\in\mathbb{N}}$ is a filtration. The uniqueness result in Lemma \ref{prop-uniquedef} implies that $\alpha_k$ is a constant, denoted by $\alpha\in(0,1/2]$, for all $k\in\N$, and hence, we have
$$
U= \ex_\alpha^{\p}~~{\rm on}~\bigcup_{k\in \N}L^\infty (\Omega,\mathcal F_k,\p).
$$ 
Next, following similar arguments as in step 3 of the proof of Theorem \ref{th:main-general}, we can directly verify that $U(X)=\ex_\alpha^{\p}(X)$ for all $X\in\X$ that satisfies $X=q_X(V)$. Hence, for any  $X\in\X$
\begin{align*}
U(X)=U(q_X(V))=\ex_\alpha^{\p}(q_X(V))=\ex_\alpha^{\p}(X),
\end{align*}
where  we used probabilistic sophistication twice. This completes the proof.
\end{proof}

\begin{remark}
In the proof of Proposition \ref{prop:classic} it is actually not necessary to assume that $(\Omega,\mathcal F,\p)$ is standard, the only requirement needed is that it is atomless.
\end{remark}

\subsubsection{Theorem \ref{th:p-ex}}\label{app:th3}

\begin{proof}[Proof of Proposition \ref{prop-id}]
For any $A \in \mathcal{F}$, let $z_A^{P,Q}= \ex^{P,Q}(\id_A)$. By a straightforward calculation,
$$
z_A^{P,Q}=\frac{P(A)}{P(A)+Q(A^c)}. 
$$
To prove the if part, notice that if $P = \lambda Q$ then for each disjoint $A,B,C \in \mathcal{F}$
$$
z_A^{P,Q} \geq z_B^{P,Q} \iff P(A) \geq P(B) \iff  z_{A \cup C}^{P,Q} \geq z_{B \cup C}^{P,Q} 
$$
that shows that axiom \ref{ax:EI} holds. 

To prove the only if part, notice first that it can be assumed without loss of generality that $P,Q\in\mathcal M_1$. Indeed, recalling that $\widetilde P=P/P(\Omega)$ for any $P\in\mathcal M$, it holds that
\begin{align*}
z_A^{P,Q} \geq z_B^{P,Q} &\iff P(A)Q(B^c) \geq P(B)Q(A^c)\\ 
&\iff \widetilde{P}(A)\widetilde{Q}(B^c) \geq \widetilde{P}(B)\widetilde{Q}(A^c) \iff z_A^{\widetilde{P},\widetilde{Q}} \geq z_B^{\widetilde{P},\widetilde{Q}}
\end{align*}
and similarly 
$$
z_{A \cup C}^{P,Q} \geq z_{B \cup C}^{P,Q} \iff z_{A \cup C}^{\widetilde{P},\widetilde{Q}} \geq z_{B \cup C}^{\widetilde{P},\widetilde{Q}},
$$
so axiom \ref{ax:EI} holds for $(P,Q)$ if and only if it holds for $(\widetilde{P}, \widetilde{Q})$. Further, $P$ and $Q$ are proportional if and only if $\widetilde{P}$ and $\widetilde{Q}$ are equal. 

In order to prove by contradiction the only if part, it thus suffices to show that if $\widetilde{P} \neq \widetilde{Q}$, then axiom \ref{ax:EI} is violated. Let $\varphi = \d \widetilde{Q}/\d \widetilde{P}$. Assuming by contradiction that $\widetilde{P} \neq \widetilde{Q}$ and recalling that $\widetilde{P}(\Omega)=\widetilde{Q}(\Omega)=1$ and that $P \ac Q \ac \mathbb{P}$, it follows that there exists $D_1$ and $D_2$ with $\mathbb{P}(D_1)>0$, $\mathbb{P}(D_2)>0$ and $D_1 \cap D_2 = \varnothing$ such that $\varphi >1$ on $D_1$ and $\varphi <1$ on $D_2$. Since $\mathbb{P}$ is atomless, for each $\epsilon >0$ sufficiently small we can find $A \subseteq D_1$ and $B \subseteq D_2$ such that
$\widetilde{P}(A) = \epsilon$, $\widetilde{Q}(A)=\epsilon \lambda_1 (\epsilon)$, $\widetilde{P}(B)=\epsilon \lambda_2 (\epsilon)$, $\widetilde{Q}(B)=\epsilon$, with $1+c < \lambda_1(\epsilon) \leq \widetilde{Q}(D_1)/\widetilde{P}(D_1)$ and $1+ c < \lambda_2(\epsilon) \leq \widetilde{P}(D_2)/\widetilde{Q}(D_2)$, for some fixed $c>0$. 
Letting $C= (A \cup B)^c$, by a direct computation we get
$$
z_{B}^{P,Q} > z_{A}^{P,Q} \iff \epsilon \cdot \frac{\lambda_1(\epsilon) \lambda_2(\epsilon)-1}{\lambda_2(\epsilon) -1} < 1
$$
and
$$
z_{A \cup C}^{P,Q} > z_{B \cup C}^{P,Q} \iff \epsilon \cdot \frac{\lambda_1(\epsilon) \lambda_2(\epsilon)-1}{\lambda_1(\epsilon) -1} < 1.
$$
This shows that by choosing $\epsilon$ sufficiently small axiom \ref{ax:EI} is violated. 
\end{proof}

Theorem \ref{th:p-ex} follows immediately from Theorem \ref{th:main-general} and Proposition \ref{prop-id}.

\subsection{Proofs of results in Section \ref{sec:discuss}}\label{app:sec5}
\begin{proof}[Proof of Theorem \ref{th:riskaverse1}]

Let $R=\alpha P$ and $H=(1-\alpha) Q$.
We first prove that (iv)  implies   (i) and (ii) using the following property: For $x=\ex^{R,H}(X)$, we have
\begin{align}\label{eq-WRA2} 
\E^R[X] - R(\Omega) x &= \E^{R} [X-x]\notag\\
&= \E^R[(X-x)_+] - \E^R[(x-X)_+]  =  \E^{H-R} [(x-X)_+],
\end{align}
 where the last step follows from $\E^R[(X-x)_+]=\E^H[(x-X)_+]$,
and hence $x \le \E^{P}[X]$  if $H\ge R$. Similarly, 	
\begin{align}\label{eq-WRA1}
	\E^H[X] -H(\Omega) x  &=  \E^H[X-x] \notag\\
 &= \E^H[(X-x)_+] - \E^H[(x-X)_+] = \E^{H-R} [(X-x)_+],
\end{align}
 and hence $x \le \E^{Q}[X]$ if $H\ge R$. The above arguments show that (iv) implies both (i) and (ii). 
	
Next, we show that (i) implies (iv).  Assume by contradiction that $R(A)>H(A)$ for some $A\in\mathcal F$ and $\succsim$ is weakly risk averse under $P$.
Note that we can always take $A$ such that $A\ne \Omega$  
if $\Omega$ is finite and 
 $\p(A)\in (0,1)$  if  $\X=L^\infty(\Omega,\mathcal F,\p)$.
  By Axiom \ref{ax:SM},
$1 \succ \id_A \succ 0$ 
and hence 
 $1>\ex^{R,H}(\id_A)>0$. Then by \eqref{eq-WRA2} we have $\E^R[\id_A] < R(\Omega) \ex^{R,H}(\id_A)$, conflicting weakly risk aversion.
A similar argument using \eqref{eq-WRA1} shows that  (ii) implies (iv).

Finally, we show that (iii) is equivalent to (iv).
By Proposition \ref{prop-cv} in Section \ref{sec:properties}, under Axiom \ref{ax:SM}  concavity of $\ex^{R,H}$ is equivalent to $R\le H$.   
It suffices to verify that quasi-concavity of $\ex^{R,H}$ implies concavity of $\ex^{R,H}$. 
It follows from Proposition \ref{prop-gextrivial} in Section \ref{sec:properties} that $\ex^{R,H}$ is translation invariant, i.e., $\ex^{R,H}(Z+c)=\ex^{R,H}(Z)+c$ for all $Z\in\X$ and $c\in\R$.
This is sufficient for our desired implication (e.g., Proposition 2.1 of \cite{CMMM11}). 
For a self-contained proof of this claim, let $m=\ex^{R,H}(X)-\ex^{R,H}(Y)$ for $X,Y\in\mathcal X$. For $\lambda\in[0,1]$, we have
\begin{align*}
 \ex^{R,H}(\lambda X+(1-\lambda)Y)
 &=\ex^{R,H}(\lambda X+(1-\lambda)(Y+m))-(1-\lambda)m\\
& \ge \ex^{R,H}(X)-(1-\lambda)m\\
&=\lambda\ex^{R,H}(X)+(1-\lambda)\ex^{R,H}(Y),
\end{align*}
where we have used quasi-concavity in the second step since  $\ex^{R,H}(X)=\ex^{R,H}(Y+m)$. Hence, we complete the proof.
\end{proof}

\begin{proof}[Proof of Proposition \ref{prop-jointRA}]
Let $R=\alpha P$, $H=(1-\alpha)Q$ and $T=H-R$. Since  $R\le H$ and $R\neq H$, we have $T\in\mathcal M$.
Note that
\begin{align*}
\ex^{R,H}(X)=\inf\{x\in\R: \E^R[X]-R(\Omega)x\le \E^{T}[(x-X)_+]\},~~X\in\mathcal X.
\end{align*}
For $X,Y\in\mathcal X$,
suppose that $X\ge_{\rm ssd}^P Y$ and $X\ge_{\rm ssd}^{\widetilde T} Y$, and we aim to show that $\ex^{R,H}(X)\ge \ex^{R,H}(Y)$. It is clear that $\E^R[X]\ge \E^R[Y]$. Since $t\mapsto -(x-t)_+$ is an increasing and concave mapping for all  $x\in\R$, we have $\E^{T}[(x-X)_+]\le \E^{T}[(x-Y)_+]$. Therefore,
\begin{align*}
\{x: \E^R[X]-R(\Omega)x\le \E^{T}[(x-X)_+]\}\subseteq \{x: \E^R[Y]-R(\Omega)x\le \E^{T}[(x-Y)_+]\}.
\end{align*}
This implies $\ex^{R,H}(X)\ge \ex^{R,H}(Y)$, showing the desired joint strong risk aversion. 
\end{proof}

\subsection{Proofs of results in Sections \ref{sec:AA}}
\label{sec:proof-AA}
 
\subsubsection{Theorem \ref{thm:diesel}}

\begin{proof}[Proof of Theorem \ref{thm:diesel}]
(ii) $\Rightarrow$ (i). Without loss of generality, we assume that $0\in u\left(  \C\right)^{\circ}$; otherwise a constant shift would do, noting that $u$ is cardinally unique. 
Axioms \ref{ax:AA-SM}, \ref{ax:AA-I}, \ref{ax:AA-C} and \ref{ax:AA-ND}
are trivial. Let us consider Axiom \ref{ax:AA-DA}. We denote by $\succsim'$ the preference on $B_0(u(\mathcal C))$ that has $\ex_\alpha^{P,Q}$ as certainty equivalent and as before we say that 
$X,Y\in B_0(u(\mathcal C))$ are disco if $\{\omega\in \Omega: X(\omega)\prec' X\}=\{\omega\in \Omega: Y(\omega)\prec' Y\}$.
Suppose that $f,g,h\in B_{0}\left(  \C\right)$ satisfies that $f$ and $g$ are disco, and $g\sim h$. 
Let $X=u(f)$, $Y=u(g)$ and $Z=u(h)$. Because $f$ and $g$ are disco, we have
\begin{align*}
\{\omega\in \Omega: X(\omega)<\ex_\alpha^{P,Q}(X)\}=\{\omega\in \Omega: Y(\omega)<\ex_\alpha^{P,Q}(Y)\}.
\end{align*}
Moreover, $g\sim h$ implies $\ex_\alpha^{P,Q}(Y)=\ex_\alpha^{P,Q}(Z)$. Note that the mapping $\ex_\alpha^{P,Q}: B_0(\R)\to\R$ is positive homogeneous (see Proposition \ref{prop-gextrivial}) and $\lambda X,(1-\lambda)Y,(1-\lambda)Z\in B_0(\mathcal C)$ for any $\lambda\in(0,1)$. It holds that $(1-\lambda)Y\sim' (1-\lambda)Z$ and for all $\lambda\in(0,1)$
\begin{align*}
\{\omega\in \Omega: \lambda X(\omega)<\ex_\alpha^{P,Q}(\lambda X)\}=\{\omega\in \Omega: (1-\lambda)Y(\omega)<\ex_\alpha^{P,Q}((1-\lambda)Y)\},
\end{align*}
which implies that $\lambda X$ and $(1-\lambda)Y$ are disco.
It follows from Theorem \ref{th-main} that $\succsim'$ satisfies Axiom \ref{ax:DA}, and hence, 
\begin{align*}
\lambda X+(1-\lambda)Z\succsim' \lambda X+(1-\lambda)Y,~~\forall \lambda\in(0,1).
\end{align*}
This is equivalent to
\begin{align*}
\ex_\alpha^{P,Q}(u(\lambda f+ (1-\lambda)h))&=\ex_\alpha^{P,Q}(\lambda u(f)+(1-\lambda)u(h))\\
&\ge \ex_\alpha^{P,Q}(\lambda u(f)+(1-\lambda)u(g))\\
&=\ex_\alpha^{P,Q}(u(\lambda f+ (1-\lambda)g)),~~\forall \lambda\in(0,1).
\end{align*}
This gives $\lambda f+ (1-\lambda)h\succsim \lambda f+ (1-\lambda)g$ for all $\lambda\in(0,1)$. We have completed this implication.

(i) $\Rightarrow$ (ii). It follows from Lemma \ref{lm-CGMMS11} that there exists an affine function
$u:\C\rightarrow\mathbb{R}$ and a monotone, continuous, and normalized function
$I:B_{0}\left(  u\left(  \C\right)  \right)  \rightarrow\mathbb{R}$ such that,
given any $f,g\in B_{0}\left(  \C\right)$,
\begin{align}\label{eq-AAlm1}
f\succsim g\iff I\left(  u\left(  f\right)  \right)  \geq I\left(  u\left(
g\right)  \right).
\end{align}
Axiom \ref{ax:AA-ND} implies that $u$ is non-constant. Without loss of generality, we can assume $0\in u\left(  \C\right)  ^{\circ}$.
First, we aim to verify that $I$ is strictly monotone and positively homogeneous on $B_0(u(\mathcal C))$. For $X,Y\in B_0(u(\mathcal C))$ with $X\ge Y$ and $X\neq Y$, there exists $f,g\in B_0(\mathcal C)$ such that $X=u(f)$ and $Y=u(g)$. It is obvious that $I(u\circ f (\omega))=X(\omega)\ge Y(\omega)=I(u\circ g(\omega))$ for all $\omega\in\Omega$, which implies $f(\omega)\succsim g(\omega)$ for all $\omega\in \Omega$, and there exists $\omega'$ such that $u\circ f (\omega')>u\circ g (\omega')$, which means that $f(\omega')\succ g(\omega')$.
By Axiom \ref{ax:AA-SM}, we have $f\succ g$, and hence, 
\begin{align*}
I(X)=I(u\circ f)>I(u\circ g)=I(Y).
\end{align*}
This gives the strict monotonicity of $I$ on $B_0(u(\mathcal C))$. To see positive homogeneity of $I$, denote by $\succsim'$ the preference on $B_0(u(\mathcal C))$ that has $I$ as certainty equivalent and define disco for $\succsim'$ as before.
Combining with \eqref{eq-AAlm1},
it is straightforward to check that $\succsim$ satisfies Axiom \ref{ax:AA-DA} if and only if the following statement holds:
\begin{itemize}
\item[(a)] If $X,X',Y\in B_0(u(\mathcal C))$ are such that $X$ and $Y$ are disco, then
\begin{align*}
X'\sim' X~\Longrightarrow~\lambda X'+(1-\lambda)Y\succsim' \lambda X+(1-\lambda)Y,~~\forall \lambda\in(0,1).
\end{align*}
\end{itemize}
Below we will directly use Statement (a).
Let $X\in B_0(u(\mathcal C))$ and $Y=0$. It holds that $I(X)$ and $Y$ are disco because their disappointment event are both empty set. Using Statement (a), we have
\begin{align*}
\lambda X+(1-\lambda)Y\succsim' \lambda I(X)+(1-\lambda)Y,~~\forall \lambda\in(0,1),
\end{align*}
and hence, $I(\lambda X)\ge \lambda I(X)$ for all $\lambda\in(0,1)$. On the other hand, for $X\in B_0(u(\mathcal C))$, denote by $A=\{\omega\in\Omega: X(\omega)\prec' X\}$ and define $Y_\epsilon=\epsilon \id_{A^c}$ for $\epsilon>0$. Since $I$ is strictly monotone, it is straightforward to check that $I(Y_\epsilon)\in(0,\epsilon)$, which implies
\begin{align*}
\{\omega\in \Omega:Y_\epsilon(\omega)\prec' Y_\epsilon\}=\{\omega\in \Omega:Y_\epsilon(\omega)< I(Y_\epsilon)\}=A.
\end{align*}
Therefore, $X$ and $Y_\epsilon$ are disco, and using Statement (a) yields
\begin{align*}
\lambda I(X)+(1-\lambda)Y_\epsilon\succsim'\lambda X+(1-\lambda)Y_\epsilon,~~\forall \lambda\in(0,1).
\end{align*}
Letting $\epsilon \downarrow 0$ and noting that $I$ is continuous, we have $\lambda I(X)\ge I(\lambda X)$ for all $\lambda\in(0,1)$. Therefore, we have concluded that $I(\lambda X)=\lambda I(X)$ for all $X\in B_0(u(\mathcal C))$ and $\lambda\in[0,1]$. Suppose that $\lambda>1$ and $ X\in B_0(u(\mathcal C))$ with $\lambda X\in B_0(u(\mathcal C))$. It holds that 
\begin{align*}
I(X)=I\left(\frac{1}{\lambda}\cdot(\lambda X)\right)=\frac{1}{\lambda}I(\lambda X).
\end{align*}
This gives the positive homogeneity of $I$ on $B_0(u(\mathcal C))$.

Next, we extend $I$ to $B_0(\R)$ in the following way:
\begin{align*}
\widetilde{I}(X):=\frac{1}{\lambda_X}I(\lambda_X X)~~{\rm with}~\lambda_X:=\sup\{\lambda\ge 0:\lambda X\in B_0(u(\mathcal C))\},~~X\in\mathcal B_0(\R).
\end{align*}
In fact, for $X\in B_0(\R)$ and any $\lambda>0$ with $\lambda X\in B_0(u(\mathcal C))$, we have $I(\lambda X)/\lambda=I(\lambda_X X)/\lambda_X$.
It is clear that $\widetilde{I}(X)=I(X)$ for all $X\in B_0(u(\mathcal C))$ and $\widetilde{I}(m)=m$ for all $m\in\R$ because $I$ is positively homogeneous on $B_0(u(\mathcal C))$. We also extend $\succsim'$ to $B_0(\R)$ with $\widetilde{I}$ as its certainty equivalent. Our aim is to verify that $\succsim'$ satisfies Axioms \ref{ax:SM}, \ref{ax:C} and \ref{ax:DA}. For $X,Y\in B_0(\R)$ with $X\ge Y$ and $X\neq Y$, let $\lambda>0$ be such that $\lambda X, \lambda Y\in B_0(u(\mathcal C))$. It holds that 
\begin{align*}
\widetilde{I}(X)=\frac{1}{\lambda}I(\lambda X)>\frac{1}{\lambda}I(\lambda Y)=\widetilde{I}(Y),
\end{align*}
where the inequality follows from the strict monotonicity of $I$. Hence, $\succsim'$ satisfies Axiom \ref{ax:SM}. To see Axiom \ref{ax:C}, it suffices to verify the continuity of $\widetilde{I}$. Let $\{X_n\}_{n\in\N}\subseteq B_0(\R)$ be a sequence such that $X_n\to X$ with $X\in B_0(\R)$. Then, there exists $\lambda>0$ such that $\lambda X,\lambda X_n\in B_0(u(\mathcal C))$ for all $n\in\N$, and hence,
\begin{align*}
\widetilde{I}(X_n)=\frac{1}{\lambda}I(\lambda X_n)\to \frac{1}{\lambda}I(\lambda X)=\widetilde{I}(X),
\end{align*}
where the convergence follows from the continuity of $I$. This gives Axiom \ref{ax:C}. Let us now consider Axiom \ref{ax:DA}. Suppose that $X,X',Y\in B_0(\R)$ satisfy $X\sim' X'$ and $X,Y$ are disco. Let $\lambda>0$ be such that $\lambda X,\lambda X',\lambda Y\in B_0(u(\mathcal C))$. It holds that
\begin{align*}
I(\lambda X)=\lambda \left(\frac{1}{\lambda}I(\lambda X)\right)
=\lambda \widetilde I(X)=\lambda \widetilde I(Y)=\lambda \left(\frac{1}{\lambda}I(\lambda Y)\right)=I(\lambda Y)
\end{align*}
and 
\begin{align*}
\{\omega\in \Omega: \lambda X(\omega)\prec' \lambda X\}&=\{\omega\in \Omega: \lambda X(\omega)<I(\lambda X)\}=\{\omega\in \Omega:  X(\omega)< \widetilde{I}(X)\}\\
&=\{\omega\in \Omega:  Y(\omega)< \widetilde{I}(Y)\}=\{\omega\in \Omega: \lambda Y(\omega)<I(\lambda Y)\}\\
&=\{\omega\in \Omega: \lambda Y(\omega)\prec' \lambda Y\}.
\end{align*}
Therefore, we have that $\lambda X\sim' \lambda X'$ and $\lambda X,\lambda Y$ are disco. Using Statement (a) yields 
\begin{align*}
\frac{1}{2}(\lambda X'+\lambda Y)\succsim' \frac{1}{2}(\lambda X+\lambda Y),
\end{align*}
and this is equivalent to
\begin{align*}
\frac{2}{\lambda}I\left(\frac{1}{2}(\lambda X'+\lambda Y)\right)\ge \frac{2}{\lambda}I\left(\frac{1}{2}(\lambda X+\lambda Y)\right),
\end{align*}
that is, $\widetilde I(X'+Y)\ge \widetilde I(X+Y)$. Hence, Axiom \ref{ax:DA} holds for $\succsim'$.

Finally, applying Theorem \ref{th-main}, we know that $\widetilde I=\ex_\alpha^{P,Q}$, where $P,Q$ are stricly positive probability measures and $\alpha\in\left(0,1/2\right]$ with $0<\alpha P\leq(1-\alpha)Q$. The uniqueness of the triplet $(P,Q,\alpha)$ follows from Lemma \ref{prop-uniquedef}. By Proposition \ref{prop-gextrivial}, we know that $\ex_\alpha^{P,Q}$ is translation-invariant and positively homogeneous, and this implies that $u$ is cardinally unique.
Hence, we complete the proof.
\end{proof}

\subsubsection{Theorem \ref{thm:AASPS}}

\begin{proof}[Proof of Theorem \ref{thm:AASPS}]
(ii) $\Rightarrow$ (i). Axioms \ref{ax:AA-SM}, \ref{ax:AA-I}, \ref{ax:AA-C}. \ref{ax:AA-ND} are trivial. 
We can share a similar proof to Theorem \ref{thm:diesel} to verify Axiom \ref{ax:AA-DA}. Let us focus on Axiom \ref{ax:AA-MC}. 
Suppose that $x,y,z\in\mathcal C$ satisfy $x\succ y$, and $\{A_n\}_{n\in\N}\subseteq \mathcal F$ satisfies $A_1\supseteq A_2\supseteq\dots$ and $\bigcap_{n\in\N}A_n=\varnothing$. 
It is clear that $u(x)>u(y)$ and we assume without loss of generality that $u(x)>0$. By standard calculation and noting that $P,Q$ are probability measures, we have 
\begin{align*}
\ex_\alpha^{P,Q}(u(z\id_{A_{n}}+x\id_{A_{n}^c}))
&=\ex_\alpha^{P,Q}(u(z)\id_{A_{n}}+u(x)\id_{A_{n}^c})\\
&=\frac{\alpha u(x)P(A_n^c)+(1-\alpha)u(z)Q(A_n)}{\alpha P(A_n^c)+(1-\alpha)Q(A_n)}\to u(x)~~{\rm as}~n\to\infty.
\end{align*}
This implies that $z\id_{A_{n}}+x\id_{A_{n}^c}\succ y$ for large enough $n$, and Axiom \ref{ax:AA-MC} holds.


(i) $\Rightarrow$ (ii). By Lemma \ref{lm-CGMMS11}, there exist an affine function
$u:\C\rightarrow\mathbb{R}$ and a monotone, continuous, and normalized function
$I:B_{0}\left(  u\left(  \C\right)  \right)  \rightarrow\mathbb{R}$ such that,
given any $f,g\in B_{0}\left(  \C\right)$,
\begin{align}\label{eq-AAlm}
f\succsim g\iff I\left(  u\left(  f\right)  \right)  \geq I\left(  u\left(
g\right)  \right).
\end{align}
Axiom \ref{ax:AA-ND} implies that $u$ is non-constant. Without loss of generality, we can assume $0\in u\left(  \C\right)  ^{\circ}$.

First, following a similar arguments to Theorem \ref{thm:diesel}, 
we can construct a mapping $\widetilde{I}$ on $B_0(\R)$, which is an extension of $I$ as follows:
\begin{align*}
\widetilde{I}(X):=\frac{1}{\lambda_X}I(\lambda_X X)~~{\rm with}~\lambda_X:=\sup\{\lambda\ge 0:\lambda X\in B_0(u(\mathcal C))\},~~X\in\mathcal B_0(\R).
\end{align*}
Moreover, the following properties holds:
\begin{itemize}
\item[]
For $X\in B_0(\R)$ and any $\lambda>0$ with $\lambda X\in B_0(u(\mathcal C))$, we have $I(\lambda X)/\lambda=I(\lambda_X X)/\lambda_X$; 

\item[]
The preference, denoted by $\succsim'$, on $B_0(\R)$ with  $\widetilde{I}$ as its certainty equivalent satisfies Axioms \ref{ax:SM}, \ref{ax:C} and \ref{ax:DA}. 
\end{itemize}
Next, we aim to verify that $\succsim'$ satisfies Axiom \ref{ax:MC}. Let $\{A_n\}_{n\in\N}\subseteq \mathcal F$ be a sequence such that
 $A_1\supseteq A_2\supseteq\dots$ and $\bigcap_{n\in\N}A_n=\varnothing$.
By Axiom \ref{ax:AA-MC}, it is straightforward to check that for all $a,b,c\in u(\mathcal C)$ with $a>b$, there exists $n_0\in\N$ such that $c\id_{A_{n_0}}+a\id_{A_{n_0}^c}\succ' b$. For any $m\in\R$, there exists $\lambda> 0$ such that $\lambda m,\lambda\in u(\mathcal C)$. 
Hence,
$
\lambda m\id_{A_{n_0}}+\lambda\id_{A_{n_0}^c}\succ' 0
$
for some $n_0\in\N$. Thus, we have
\begin{align*}
\lambda\widetilde{I}( m\id_{A_{n_0}}+\id_{A_{n_0}^c})
&=\lambda\left(\frac{1}{\lambda}I(\lambda m\id_{A_{n_0}}+\lambda\id_{A_{n_0}^c})\right)\\
&=I(\lambda m\id_{A_{n_0}}+\lambda\id_{A_{n_0}^c})
=\widetilde{I}(\lambda m\id_{A_{n_0}}+\lambda\id_{A_{n_0}^c})
> \widetilde{I}(0)
=0.
\end{align*}
This implies $m\id_{A_{n_0}}+\id_{A_{n_0}^c}\succ' 0$, and we have verified Axiom \ref{ax:MC} for $\succsim'$.  Hence, we have concluded that $\succsim'$ satisfies Axioms \ref{ax:SM}, \ref{ax:C}, \ref{ax:DA} and \ref{ax:MC} on $B_0(\R)$.
Finally, applying Corollary \ref{cor:thm2} yields the desired result.
\end{proof}

\section{Counter-examples}
\label{app:A}
In this appendix, we present some counter-examples to justify some claims in the paper and illustrate several technical conditions. 
\subsection{Necessity of dimension larger than 3 in Theorem \ref{th-main}}

In this section, we aim to give a counter-example of the preference relation $\succsim$, which satisfies Axioms \ref{ax:SM}, \ref{ax:C} and \ref{ax:DA} and can not be represented by a duet expectile, in the case $\Omega=[n]$ with $n=3$.
We present the following example of a functional $U$, and in Proposition \ref{prop:cen4} we show that a preference relation $\succsim$ represented by $U$ is a desirable counter-example.

\begin{example}
Define the subsets of $\mathbb{R}^3$:
$$
A_1=\{x\ge y\ge z\},~~A_2=\{x\ge z\ge y,~x+2y\ge 3z\},~~A_3=\{y\ge x\ge z,~y+z\le 2x\},
$$
$$
B_1=\{z\ge x\ge y\},~~B_2=\{x\ge z\ge y,~x+2y\le 3z\},~~B_3=\{z\ge y\ge x,~x+z\ge 2y\},
$$
$$
C_1=\{y\ge z\ge x\},~~C_2=\{z\ge y\ge x,~x+z\le 2y\},~~C_3=\{y\ge x\ge z,~y+z\ge 2x\}.
$$
Let $A=\bigcup_{i=1}^3A_i$, $B=\bigcup_{i=1}^3B_i$ and $C=\bigcup_{i=1}^3C_i$. It can be easily checked that $\R^3=A\cup B\cup C$.
Define a function $f:\R^3\to \R$ as follows:
\begin{align}\label{eq-cefun}
	f(x,y,z)=\begin{cases}
		f_A(x,y,z)=\frac{x+2y+2z}{5},~~& (x,y,z)\in A,\\
		f_B(x,y,z)=\frac{x+2y+z}{4},~~& (x,y,z)\in B,\\
		f_C(x,y,z)=\frac{x+y+z}{3},~~& (x,y,z)\in C.
	\end{cases}
\end{align}
A direct verification shows that $f_A=f_B$ on $A\cap B$, $f_A=f_C$ on $A\cap C$, $f_B=f_C$ on $B\cap C$ and $f_A=f_B=f_C$ on $A\cap B\cap C$.
This implies that $f$ is well defined on $\R^3$, and it is continuous and strictly increasing. Define $U:\mathcal X\to \R$ as $U(X)=f(X(1),X(2),X(3))$. 
\end{example}


\begin{proposition}\label{prop:cen4}
The following two statements hold.
\begin{itemize}
	\item[(i)] $U$ can not be expressed as a duet expectile.
	\item[(ii)] The preference relation represented by $U$ satisfies Axioms \ref{ax:SM}, \ref{ax:C} and \ref{ax:DA}. 
\end{itemize}
\end{proposition}
\begin{proof}
(i) We assume by contradiction that $U=\ex^{P,Q}$ for some $P,Q\in\mathcal M$. Denote by $p_i=P(i)$ and $q_i=Q(i)$ for $i=1,2,3$. For $x,y,z\in\R$, let $X\in\mathcal X$ be such that $X(1)=x$, $X(2)=y$ and $X(3)=z$. If $(x,y,z)\in A_2$, then we have $x\ge f_A(x,y,z)$ and $y,z\le f_A(x,y,z)$, and hence, 
\begin{align*}
	p_1(x-f_A(x,y,z))=q_2(f_A(x,y,z)-y)+q_3(f_A(x,y,z)-z),~~(x,y,z)\in A_2,
\end{align*}
which implies 
\begin{align*}
	f_A(x,y,z)=\frac{p_1x+q_2y+q_3z}{p_1+q_2+q_3},~~(x,y,z)\in A_2.
\end{align*}
Hence, we have 
$$
(a)~p_1:q_2:q_3=1:2:2.
$$
Similarly, by considering $(x,y,z)$ on $A_3$, $B_2$ and $C_2$, we respectively obtain
\begin{align*}
	(b)~p_1:p_2:q_3=1:2:2;~(c)~p_1:q_2:p_3=1:2:1;~(d)~q_1:p_2:p_3=1:1:1.
\end{align*}
It follows from (a) and (b) that $p_2=q_2$. But (c) and (d) imply $q_2=2p_3$ and $p_2=p_3$, respectively. This yields a contradiction. Therefore, we have verified (i). 

(ii) Recall the definition of $U$ and note that $f:\R^3\to\R$ is continuous and strictly increasing. It follows that Axioms \ref{ax:SM} and \ref{ax:C} hold. It remains to verify Axiom \ref{ax:DA}. Suppose that $X_i(1)=x_i$, $X_i(2)=y_i$ and $X_i(3)=z_i$ for $i=1,2$, and $X_1,X_2$ are disco. 
By the definition of $U$,  
it is straightforward to check that $(x_i,y_i,z_i)\in S$ for $i=1,2$ with $S\in\{A,B,C\}$. Below we verify that $A$, $B$ and $C$ are all convex cones. We only consider the case of $A$ as the other two cases are similar. It suffices to show that $(x_i,y_i,z_i)\in A$ for $i=1,2$ implies $(x_1+x_2,y_1+y_2,z_1+z_2)\in A$. To see this, if $(x_i,y_i,z_i)\in A_j$ for all $i=1,2$ and some $j=1,2,3$, then $(x_1+x_2,y_1+y_2,z_1+z_2)\in A_j\subseteq A$ because $A_j$ is a convex cone for all $j=1,2,3$. If $(x_1,y_1,z_1)\in A_1$ and $(x_2,y_2,z_2)\in A_2$, it holds that $x_1+x_2\ge \max\{y_1+y_2,z_1+z_2\}$ and 
\begin{align}\label{eq-ce1}
	x_1+x_2+2(y_1+y_2)= (x_1+2y_1)+(x_2+2y_2)\ge 3z_1+3z_2=3(z_1+z_2).
\end{align}
This implies $(x_1+x_2,y_1+y_2,z_1+z_2)\in A_1\cup A_2\subseteq A$. If $(x_1,y_1,z_1)\in A_1$ and $(x_2,y_2,z_2)\in A_3$, it holds that $\min\{x_1+x_2,y_1+y_2\}\ge z_1+z_2$ and 
\begin{align}\label{eq-ce2}
y_1+y_2+z_1+z_2= (y_1+z_1)+(y_2+z_2)\le 2x_1+2x_2=2(x_1+x_2).
\end{align}
This implies $(x_1+x_2,y_1+y_2,z_1+z_2)\in A_1\cup A_3\subseteq A$.  If $(x_1,y_1,z_1)\in A_2$ and $(x_2,y_2,z_2)\in A_3$, it holds that $x_1+x_2\ge z_1+z_2$, and \eqref{eq-ce1} and \eqref{eq-ce2} both hold. This implies $(x_1+x_2,y_1+y_2,z_1+z_2)\in A_1\cup A_2\cup A_3=A$. Hence, we have concluded that $A$ is a convex cone.
Note that $f$ is a linear function on $A$, $B$ or $C$. Hence, we have $U(X_1+X_2)=U(X_1)+U(X_2)$, which shows that $U$ is additive for any disco acts. It remains to prove that $U$ satisfies superadditivity, and this is equivalent to $f(x_1+x_2,y_1+y_2,z_1+z_2)\ge f(x_1,y_1,z_1)+f(x_2,y_2,z_2)$ for any $(x_i,y_i,z_i)\in\R^3$ with $i=1,2$. To see this, through elementary analysis, one can observe that
\begin{align*}
	f_S(x,y,z)\ge f(x,y,z),~~\forall S\in\{A,B,C\},~(x,y,z)\in\R^3.
\end{align*}
Hence, for any $(x_i,y_i,z_i)\in\R^3$ with $i=1,2$, there exists $S\in\{A,B,C\}$ such that
\begin{align*}
	f(x_1+x_2,y_1+y_2,z_1+z_2)&=f_S(x_1+x_2,y_1+y_2,z_1+z_2)\\
	&=f_S(x_1,y_1,z_1)+f_S(x_2,y_2,z_2)\\
	&\ge f(x_1,y_1,z_1)+f(x_2,y_2,z_2).
\end{align*}
This completes the proof of (ii).
\end{proof}

\subsection{Necessity of Axiom \ref{ax:SM} in Theorem \ref{th-main}}\label{app:SMmain}

In this section, we aim to clarify that Theorem \ref{th-main} fails if 
Axiom \ref{ax:SM} is replaced by the weaker axiom of monotonicity:
\renewcommand\theaxiom{M}
\begin{axiom}[Monotonicity] 
	\label{ax:M} 
	For $X,Y\in \mathcal X$, $X\ge Y$ implies $X\succsim Y$.
\end{axiom} 
To this end, we will give a counter-example of the preference relation $\succsim$, which satisfies Axioms \ref{ax:M}, \ref{ax:C} and \ref{ax:DA} and can not be represented by a duet expectile, in the case $\Omega=[n]$ with $n\ge 4$. We present the following example of a functional $U$, and in Proposition \ref{prop:ceSM} we show that a preference relation $\succsim$ represented by $U$ is a desirable counter-example.


\begin{example}
Let $f \colon \R^3\to \R$ be the function defined in \eqref{eq-cefun}. Define $g \colon \R^n\to \R$ as a cylindrical extension of $f$, i.e.,
\begin{align*}
g(x_1,\dots,x_n)=f(x_1,x_2,x_3),~~\forall (x_1,\dots,x_n)\in \R^n.
\end{align*}
Further, define $U:\mathcal X\to \R$ as $U(X)=g(X(1),\dots,X(n))$. 
\end{example}

\begin{proposition}\label{prop:ceSM}
The following two statements hold.
\begin{itemize}
	\item[(i)] $U$ can not be expressed as a duet expectile.
	\item[(ii)] The preference relation represented by $U$ satisfies Axioms \ref{ax:M}, \ref{ax:C} and \ref{ax:DA}.
\end{itemize}
\end{proposition}
\begin{proof}
 
(i) We assume by contradiction that $U=\ex^{P,Q}$ for some $P,Q\in\mathcal M$. For $x_1,\dots,x_n\in\R$, let $X\in\mathcal X$ be such that $X(i)=x_i$ for $i\in [n]$. Construct four subsets in $\R^n$ as follows:
$$
D_1=\{(x_1,x_2,x_3,t,\dots,t): (x_1,x_2,x_3)\in A_2,~t=f_A(x_1,x_2,x_3)\},
$$
$$
D_2=\{(x_1,x_2,x_3,t,\dots,t): (x_1,x_2,x_3)\in A_3,~t=f_A(x_1,x_2,x_3)\},
$$
$$
D_3=\{(x_1,x_2,x_3,t,\dots,t): (x_1,x_2,x_3)\in B_2,~t=f_B(x_1,x_2,x_3)\},
$$
$$
D_4=\{(x_1,x_2,x_3,t,\dots,t): (x_1,x_2,x_3)\in C_2,~t=f_C(x_1,x_2,x_3)\}.
$$
Similar to Proposition \ref{prop:cen4} (i),
considering $U$ on these four sets yields a contradiction. This completes the proof of (i).
The statement (ii) follows from Proposition \ref{prop:cen4} (ii) and the fact that $g:\R^n\to\R$ is a cylindrical extension of the function $f:\R^3\to\R$.
\end{proof}

\subsection{Duet expectiles and strong risk aversion}\label{app:jointSRA}

For $P,Q\in\mathcal M$, the next example illustrates that $\ex^{P,Q}$ does not satisfy joint strong risk aversion under $(\widetilde{P},\widetilde{Q})$.

\begin{example}
Let $([0,1],\mathcal B([0,1]),L)$ be the standard Borel space, where $L$ is the Lebesgue measure. 
Denote by $A_i=[(i-1)/4,i/4)$ for $i=1,2,3,4$, which are intervals in $[0,1]$. Let
$P$ be a measure on $([0,1],\mathcal B([0,1]))$ such that the Radon-Nikodym derivative $\d P/\d L$ satisfies $\d P/\d L (\omega)=i$ if $\omega\in A_i$. Let $Q=4L$. Obviously, $P\ac Q$ and
$P\le Q$ as $0<\d P/\d L\le \d Q/\d L=4$. Define two acts as follows:

\begin{center}
\begin{tabular}[c]{c|c|c|c|c}
 & $A_1$ & $A_2$ & $A_3$ & $A_4$ \\ \hline
$X$ & 1 & 2 & 4 & 3 \\ \hline
$Y$ & 2 & 1 & 3 & 4 
\end{tabular}
\end{center}
It is clear that $X\laweq_{\widetilde Q} Y$, and hence, $X\ge_{\rm ssd}^{\widetilde Q} Y$. 
One can also check that $\E^P[(x-X)_+]\le \E^P[(x-Y)_+]$ for all $x\in\R$, and this implies $X\ge_{\rm ssd}^{\widetilde P} Y$ (see e.g., Theorem 4.A.2 of \cite{SS07}). By   direct calculations, we obtain
\begin{align*}
\ex^{P,Q}(X)=\frac{12}{5}=2.4~~\mbox{and}~~\ex^{P,Q}(Y)=\frac{37}{15}>2.466,
\end{align*}
and it follows $\ex^{P,Q}(X)<\ex^{P,Q}(Y)$. Therefore, $\ex^{P,Q}$ is not strongly risk averse jointly under $(\widetilde P,\widetilde Q)$.
\end{example}

\subsection{Proposition \ref{prop-cv} fails for a space of two atoms}\label{sec:example}

In this section, 
we will show that the conclusion in Proposition \ref{prop-cv} is not valid in the case $\Omega=\{1,2\}$. Moreover, we obtain the corresponding statement of Proposition \ref{prop-cv}  in this case. 





\begin{proposition}
Let $\Omega=\{1,2\}$. 
For $P,Q\in\mathcal M$, suppose that $P(i),Q(i)>0$ for $i=1,2$.
The duet expectile $\ex^{P,Q}$ is convex if and only if $P(1)P(2)\le Q(1)Q(2)$.
\end{proposition}

\begin{proof}
Denote by $p_i=P(i)$ and $q_i=Q(i)$ for $i=1,2$.
We use $(x,y)$ to represent the random variable that takes the value $x$ on $\{1\}$ and $y$ on $\{2\}$.  By Lemma \ref{lm-setcx}, it is equivalent to prove the claim that $p_1p_2\ge q_1q_2$ if and only if
\begin{align*}
\mathcal A:=\{(x,y)\in \R^2: p_1(x)_++p_2(y)_+-q_1(x)_--q_2(y)_-\le 0\}~{\rm is~a~convex~ cone}.
\end{align*}

{\bf Sufficiency}: Suppose that $\A$ is a convex cone. It is easy to see that $(q_2,-p_1)\in\A$ and $(-p_2,q_1)\in\A$. It holds that $(q_2-p_2,q_1-p_1)\in \A$. If $q_2\ge p_2$, then $q_1\le p_1$, and hence,
$p_1(q_2-p_2)+q_2(q_1-p_1)\le 0$. This implies $p_1p_2\ge q_1q_2$. If $q_2<p_2$ and $q_1\le p_1$, then $p_1p_2\ge q_1q_2$ obviously. If $q_2<p_2$ and $q_1> p_1$, then $p_2(q_1-p_1)+q_1(q_2-p_2)\le 0$ which implies $p_1p_2\ge q_1q_2$. Hence, we complete the proof of  ``$\Longleftarrow$".

{\bf Necessity}:
Suppose that $p_1p_2\ge q_1q_2$.
Note that $(x,y)\in\mathcal A$ means that it must be three cases: $x>0$ and $y<0$; $x<0$ and $y>0$; $x,y\le0$. 
Define $A=\{(x,y): x,y\le 0\}$,
\begin{align*}
B=\{(x,y): x>0,~y<0,~p_1 x+q_2y\le 0\}
\end{align*}
and 
\begin{align*}
C=\{(x,y): x<0,~y>0,~q_1x+p_2y\le 0\}.
\end{align*}
It holds that $\mathcal A=A\cup B\cup C$. Obviously, $A$, $B$ and $C$ are all convex cones. Moreover, for any $(x_1,y_1)\in A$ and $(x_2,y_2)\in B\cup C$, we have $(x_1+x_2,y_1+y_2)\in\mathcal A$. Therefore, it suffices to verify that for any $(x_1,y_1)\in B$ and $(x_2,y_2)\in C$, we have $(x_1+x_2,y_1+y_2)\in\A$. Let $x_1,x_2,y_1, y_2>0$ be such that $(x_1,-y_1)\in B$ and $(-x_2,y_2)\in C$. 
It holds that
\begin{align}\label{eq-cvtwo}
	p_1x_1\le q_2y_1~~{\rm and}~~p_2 y_2\le q_1 x_2.
\end{align}
Define $f(x,y)=p_1(x)_++p_2(y)_+-q_1(x)_--q_2(y)_-$ on $\R^2$. We aim to verify $f(x_1-x_2,y_2-y_1)\le 0$ which means that $(x_1-x_2,y_2-y_1)\in\A$.
If $x_1\ge x_2$, it follows from \eqref{eq-cvtwo} and $p_1p_2\ge q_1q_2$ that 
\begin{align*}
y_1\ge \frac{p_1}{q_2}x_1\ge \frac{p_1}{q_2}x_2\ge \frac{p_1}{q_2}\frac{p_2}{q_1}y_2\ge y_2.
\end{align*}
Therefore,
\begin{align*}
f(x_1-x_2,y_2-y_1)&=p_1(x_1-x_2)+q_2(y_2-y_1)
=p_1 x_1-q_2y_1+(q_2 y_2-p_1x_2)\\
&\le q_2y_2-p_1x_2
\le \frac{q_1q_2}{p_2}x_2-p_1x_2
=p_1x_2\left(\frac{q_1q_2}{p_1p_2}-1\right)\le 0,
\end{align*}
where we have used \eqref{eq-cvtwo} in the first and second inequalities. Hence, we have $(x_1-x_2,y_2-y_1)\in\A$. If $x_1<x_2$ and $y_1\ge y_2$, then $(x_1-x_2,y_2-y_1)\in\A$ obviously. If $x_1<x_2$ and $y_1<y_2$, then we have
\begin{align*}
f(x_1-x_2,y_2-y_1)
&=q_1(x_1-x_2)+p_2(y_2-y-1)
=q_1x_1-p_2y_1+(p_2y_2-q_1x_2)\\
&\le q_1x_1-p_2y_1
\le \frac{q_1q_2}{p_1}y_1-p_2y_1
=p_2 y_1\left(\frac{q_1q_2}{p_1p_2}-1\right)\le 0,
\end{align*}
where the first two inequalities follow from \eqref{eq-cvtwo}.
Hence, we have $(x_1-x_2,y_2-y_1)\in\A$. Therefore, we have verified the direction ``$\Longrightarrow$".
\end{proof}

\subsection{Theorem \ref{th:riskaverse1} and Proposition \ref{prop-cv} fail without \ref{ax:SM}}\label{app:SMprop}



Let $\Omega=\{1,\dots,n\}$ with $n\ge 4$ and $\mathcal F$ be the power set. 
Define 
\begin{align}\label{eq-DE}
\mathcal D=\{R\in \mathcal M: R(1)=R(\Omega)\}~~
\end{align}
Any element in $\mathcal D$ concentrates all mass on one point.
Considering $\ex^{P,Q}$ with $P\in \mathcal D$ or $Q\in \mathcal D$, it follows from Proposition \ref{prop-finSM} that $\ex^{P,Q}$ does not satisfy \ref{ax:SM}.
We have the next proposition for such $\ex^{P,Q}$.

\begin{proposition}\label{prop-ceSM}
Let $\mathcal D$ be defined in \eqref{eq-DE}. The following statements hold.
\begin{itemize}
\item[(i)] For $P\in\mathcal D$ and $Q\in\mathcal M$, $\ex^{P,Q}$ is concave and the preference relation represented by $\ex^{P,Q}$ is weakly risk averse under $\widetilde P$.


\item[(ii)] For $Q\in\mathcal D$ and $P\in\mathcal M$, $\ex^{P,Q}$ is convex and the preference relation represented by $\ex^{P,Q}$ is weakly risk seeking under $\widetilde Q$.

\end{itemize}
\end{proposition}

\begin{proof}

We only prove (ii) as (i) follows the same proof.

(ii) Suppose that $Q\in\mathcal D$ and $P\in\mathcal M$. By Lemma \ref{lm-setcx}, it suffices to verify that the following acceptance set is a convex cone:
\begin{align*}
\A=\left \{(x_1,\dots,x_n)\in\R^n: \sum_{i=1}^n p_i(x_i)_+-q_1 (x_1)_-\le 0\right\},
\end{align*}
where we have denoted by $p_i=P(i)$ and $q_i=Q(i)$ for $i=1,\dots,n$. Define $I=\{i\in\{2,\dots,n\}: p_i>0\}$. Then $\A$ can be divided into the following two sets:
\begin{align}\label{eq-oneA}
A=\left \{(x_1,\dots,x_n)\in\R^n: x_1\ge 0,~p_1 x_1+\sum_{i\in I}p_i(x_i)_+\le 0 \right\}
\end{align}
and
\begin{align}\label{eq-oneB}
B=\left \{(x_1,\dots,x_n)\in\R^n:
x_1<0,~q_1x_1+\sum_{i\in I}p_i(x_i)_+\le 0\right \}.
\end{align}
Obviously, $A$ and $B$ are both convex cones. It remains to verify that $\textbf x:=(x_1,\dots,x_n) \in A$ and $\textbf y=(y_1,\dots,y_n)\in B$ imply $\textbf x+\textbf y\in \mathcal A$. 
To see this,
define $f(x_1,\dots,x_n)=\sum_{i=1}^n p_i(x_i)_+-q_1 (x_1)_-$.
Note that $\textbf x\in A$ implies $x_i=0$ for $i\in I$. If $p_1>0$, then $x_1=0$ as $\textbf x\in A$, and hence, 
\begin{align*}
f(\textbf x+\textbf y)
=\sum_{i=1}^n p_i(x_i+y_i)_+-q_1(x_1+y_1)_-
=\sum_{i\in I}p_i (y_i)_+-q_1(y_1)_-\le 0,
\end{align*}
where the last step holds as $\textbf y\in B$.
If $p_1=0$, then we have
\begin{align*}
f(\textbf x+\textbf y)
=\sum_{i=1}^n p_i(x_i+y_i)_+-q_1(x_1+y_1)_-
=\sum_{i\in I}p_i (y_i)_+-q_1(y_1)_-\le 0.
\end{align*}
Therefore, we have concluded that $\mathcal A$ is a convex cone, and this implies the convexity of $\ex^{P,Q}$. Next we aim to show that $\ex^{P,Q}(X)\ge \E^{\widetilde Q}[X]$ for all $X\in\X$. Note that $\E^{\widetilde Q}[X]=X(1)$, and we have
\begin{align*}
&~~~~~\E^P[(X-\E^{\widetilde Q}[X])_+]-\E^Q[(\E^{\widetilde Q}[X]-X)_+]\\
&=\E^{P-Q}[(X-\E^{\widetilde Q}[X])_+]-\E^{\widetilde Q}[X]Q(\Omega)+\E^Q[X]\\
&=\E^{P-Q}[(X-\E^{\widetilde Q}[X])_+]
=\sum_{i=2}^{n}(X(i)-\E^{\widetilde Q}[X])_+P(i)\ge 0.
\end{align*}
By the definition of $\ex^{P,Q}(X)$, it holds that $\ex^{P,Q}(X)\ge\E^{\widetilde Q}[X]$. This proves part (iii).
\end{proof}

Based on Proposition \ref{prop-ceSM}, constructing counter-examples for Theorem \ref{th:riskaverse1} and Proposition \ref{prop-cv} without \ref{ax:SM} becomes straightforward. 

\begin{example}
Let $\Omega=\{1,2,3,4\}$. Let $P,Q\in\mathcal M$ be such that $P(1)=2$, $P(i)=0$ for $i=2,3,4$ and $Q(j)=1$ for $j\in[4]$. It is clear that $P\in\mathcal D$, $P\not\le Q$ and $\ex^{P,Q}$ does not satisfy \ref{ax:SM}. By Proposition \ref{prop-ceSM} (i), we know that $\ex^{P,Q}$ is concave and the preference relation represented by $\ex^{P,Q}$ is weakly risk averse under $\widetilde P$. This serves as a counter-example for Theorem \ref{th:riskaverse1} ((i) or (iii) $\Rightarrow$ (iv)) and the necessity of Proposition \ref{prop-cv} without \ref{ax:SM}. 

\end{example}

\end{appendix}

\end{document}